\documentclass[reprint, superscriptaddress, nofootinbib, amsmath,amssymb, aps]{revtex4-2}

\usepackage{lineno}

\usepackage[whole]{bxcjkjatype}
\usepackage{graphicx}
\usepackage{dcolumn}
\usepackage{bm}
\usepackage[colorlinks,linkcolor=blue,urlcolor=blue,citecolor=blue,hypertexnames=false]{hyperref}

\usepackage{dsfont}
\usepackage{physics}
\usepackage{nicefrac}
\usepackage{enumitem}
\usepackage{tcolorbox}
\usepackage{mathtools}
\usepackage{tabularx}
\newcolumntype{C}{>{\centering\arraybackslash}X}

\usepackage{amsmath, amssymb, amsfonts, amsthm}
\usepackage{thmtools, thm-restate}

\DeclareMathOperator{\id}{id}

\DeclareMathOperator{\spn}{span}

\newtheorem{theorem}{Theorem}
\newtheorem*{theorem*}{Theorem}

\newtheorem{lemma}[theorem]{Lemma}

\newtheorem{proposition}[theorem]{Proposition}

\newcommand{\red}{\color{red}}

\allowdisplaybreaks

\def \be {\begin{equation}}
\def \ee {\end{equation}}

\def \sofc2{{\cal S}({\mathbb C}^2)}

\def\>{\rangle}
\def\<{\langle}

\DeclareMathOperator\supp{supp}

\begin{document}

\title{Generalized Quantum Stein's Lemma for Classical-Quantum Dynamical Resources}

\author{Masahito Hayashi}
\email{hmasahito@cuhk.edu.cn}
\affiliation{School of Data Science, The Chinese University of Hong Kong, Shenzhen, Longgang District, Shenzhen, 518172, China}
\affiliation{International Quantum Academy, Futian District, Shenzhen 518048, China}
\affiliation{Graduate School of Mathematics, Nagoya University, Chikusa-ku, Nagoya 464--8602, Japan}
\author{Hayata Yamasaki}
\email{hayata.yamasaki@gmail.com}
\affiliation{
Department of Computer Science, Graduate School of Information Science and Technology, The University of Tokyo, 7--3--1 Hongo, Bunkyo-ku, Tokyo, 113--8656, Japan
}

\begin{abstract}
Channel conversion constitutes a pivotal paradigm in information theory and its applications to quantum physics, providing a unified problem setting that encompasses celebrated results such as Shannon's noisy-channel coding theorem. Quantum resource theories (QRTs) offer a general framework to study such problems under a prescribed class of operations, such as those for encoding and decoding. In QRTs, quantum states serve as static resources, while quantum channels give rise to dynamical resources. A recent major advance in QRTs is the generalized quantum Stein's lemma, which characterizes the optimal error exponent in hypothesis testing to discriminate resource states from non-resourceful states, enabling a reversible QRT framework for static resources where asymptotic conversion rates are fully determined by the regularized relative entropy of resource. However, applications of QRTs to channel conversion require a framework for dynamical resources. The earlier extension of the reversible framework to a fundamental class of dynamical resources, represented by classical-quantum (CQ) channels, relied on state-based techniques and imposed an asymptotic continuity assumption on operations, which prevented its applicability to conventional channel coding scenarios. To overcome this problem, we formulate and prove a generalized quantum Stein's lemma directly for CQ channels, by developing CQ-channel counterparts of the core proof techniques used in the state setting. Building on this result, we construct a reversible QRT framework for CQ channel conversion that does not require the asymptotic continuity assumption, and show that this framework applies to the analysis of channel coding scenarios. These results establish a fully general toolkit for CQ channel discrimination and conversion, enabling their broad application to core conversion problems for this fundamental class of channels.
\end{abstract}

\maketitle

\tableofcontents

\section{Introduction}

\paragraph*{Background}

A foundational result in information theory~\cite{cover2012elements} is Shannon's noisy-channel coding theorem~\cite{6773024}, which characterizes the maximum rate at which messages can be transmitted through repeated uses of a classical communication channel.
Physically, this gives an operational measure of how close a noisy, dissipative process, such as that for telecommunication, is to being noiseless.
The same concept is also fundamental in quantum information theory~\cite{hayashi2016quantum,Holevo+2019,watrous_2018,Wilde_2017}, where channels may take quantum states as inputs or outputs.
The central question then becomes: how many identity channels can be simulated per use of a noisy quantum channel, with vanishing error?
The corresponding quantum analogues of capacities~\cite{IEEE-IT-651037,PhysRevA.56.131,PhysRevLett.83.3081,1035117,PhysRevA.55.1613,shor2002quantum,1377491} quantify the noiselessness of quantum dynamics.
Beyond communication, these ideas have found broad applications in physics, offering a unifying tool for the quantitative study of many-body systems~\cite{10.1063/1.1499754} and quantum gravity~\cite{Patrick_Hayden_2007}.
An apparently different yet equally fundamental task is its reverse problem: implementing as many uses of a given noisy channel as possible using only a noiseless
identity channel.
The optimal achievable rate for this reverse process is characterized by the celebrated reverse Shannon theorems~\cite{1035117,6757002}.
Together, these problems are unified in the paradigm of channel conversion, which serves as a foundation for quantum information theory.

The quantum generalization of channels is not unique.
One may consider classical-quantum (CQ) channels with classical inputs and quantum outputs, or fully quantum (QQ) channels with both quantum inputs and outputs~\cite{hayashi2016quantum,Holevo+2019}.
While QQ channels are natural, they are often intractable to analyze due to the lack of analytical techniques to handle quantum inputs.
CQ channels, by contrast, provide a more tractable extension of classical channels while exhibiting distinctive features absent in general QQ channels, as seen for example in channel resolvability~\cite{11005630}.
They therefore constitute a fundamental class of channels in the study of channel conversion problems.

Quantum resource theories (QRTs)~\cite{Kuroiwa2020,Chitambar2018} provide systematic techniques for analyzing the problems of converting resources for quantum information processing.
In QRTs, free operations specify the allowed conversions, and resources may appear either in states, called static resources, or in channels, called dynamical resources~\cite{Chitambar2018}.
A recent major advance in QRTs is the generalized quantum Stein's lemma, originally proposed in Refs.~\cite{Brand_o_2008,brandao2010reversible,Brandao2010} and proven in Refs.~\cite{hayashi2025generalizedquantumsteinslemma,10898013}, following multiple prior attempts~\cite{berta2023gap,yamasaki2024generalized}.\footnote{Note that Refs.~\cite{fang2025generalizedquantumasymptoticequipartition,fang2025errorexponentsquantumstate} more recently studied another variant of composite quantum hypothesis testing; however, their analyses impose an additional assumption on stability of polar sets under tensor product and therefore do not apply to the original setting of the generalized quantum Stein's lemma in Refs.~\cite{Brand_o_2008,brandao2010reversible,Brandao2010}, where the assumptions were later relaxed in Ref.~\cite{hayashi2025generalizedquantumsteinslemma}.}
It shows that the optimal error exponent in hypothesis testing for distinguishing independent and identically distributed (IID) copies of a state from a non-IID set of non-resource states is given by the regularized relative entropy of resource.
Using this, one can construct a reversible QRT framework for state conversion~\cite{Brandao2015,regula2023reversibility,hayashi2025generalizedquantumsteinslemma}, where the optimal rate of asymptotic conversion between resource states is fully characterized by the regularized relative entropy of resource of the states.

However, applications of QRTs to channel conversion problems require a framework for dynamical resources~\cite{Takagi2020}, and extending the reversible QRT framework from static to dynamical resources is generally nontrivial.
Reference~\cite{hayashi2025generalizedquantumsteinslemma} proposed such an extension of the reversible framework to CQ channels, but with two major restrictions.
First, it reduced the channel conversion problem to the state setting by considering the Choi states of channels, so distinguishability was measured by the trace distance between Choi states.
This measure reflects average-case distinguishability over channel inputs, but it is incompatible with the operationally natural diamond distance~\cite{10.1145/276698.276708}, which captures worst-case distinguishability over inputs.
Second, it imposed an asymptotic continuity condition on the free operations for channel conversion.
Because general superchannels for channel conversion do not satisfy this condition, the framework could not be applied to important tasks such as channel coding, where input optimization is essential; problematically, optimization over channel inputs typically violates the asymptotic continuity of the operations used in channel conversion.
As long as one relies on the existing state setting of the generalized quantum Stein's lemma, it remains challenging to derive a channel-conversion framework without these limitations.

\paragraph*{Main results and impact}

In this work, we resolve this challenge by formulating and proving a generalized quantum Stein's lemma directly for CQ channels, a fundamental class of dynamical resources.
Our result applies to the task of quantum hypothesis testing to distinguish IID copies of a CQ channel from a non-IID set of non-resource CQ channels, achieved by choosing an optimal classical input to the channel and performing the corresponding positive operator-valued measure (POVM) on its output.
The optimal performance in this task is characterized by the regularized channel divergence~\cite{cooney2016strong} between multiple copies of the given CQ channel and the closest CQ channel in the set, which serves as a natural extension of the relative entropy of resource from states to channels~\cite{Gour2019a}.
The study of channel discrimination was initiated in Ref.~\cite{5165184}, which mainly considered distinguishing IID copies of two classical channels.
This is extended to replacer channels, i.e., a special class of CQ channels, in Ref.~\cite{cooney2016strong}.
The quantum Stein's lemma for discriminating IID copies of CQ channels was analyzed in Ref.~\cite{wilde2020amortized}.
Similarly, Hoeffding bounds for asymptotic discrimination of CQ channels were obtained in Ref.~\cite{PhysRevA.105.022419}.

Extending beyond these works, our generalized quantum Stein's lemma for CQ channels applies to the discrimination of IID copies of a CQ channel from a non-IID set of CQ channels, in direct analogy with the generalized quantum Stein's lemma for states~\cite{hayashi2025generalizedquantumsteinslemma,10898013,Brandao2010}, which extends the original quantum Stein's lemma~\cite{hiai1991proper,887855} from IID states to a non-IID set of states.
The extension to the CQ-channel setting is nontrivial because CQ channels involve multiple possible inputs, in contrast to the single-state setting.
Nevertheless, we develop CQ-channel counterparts of the key techniques used in the state version of the generalized quantum Stein's lemma in Ref.~\cite{hayashi2025generalizedquantumsteinslemma}, including the pinching technique~\cite{hayashi2002optimal}, the information spectrum method~\cite{4069150}, and error-exponent bounds derived from R\'{e}nyi relative entropies~\cite{887855,cooney2016strong}.

Building on the CQ-channel version of the generalize quantum Stein's lemma, we further construct a reversible QRT framework for CQ channel conversion.
Unlike the earlier framework of Ref.~\cite{hayashi2025generalizedquantumsteinslemma}, which reduced channels to their Choi states and assessed distinguishability via the trace distance between these states, our formulation works directly with CQ channels and characterizes distinguishability using channel divergence and the diamond distance.
This shift is critical: whereas the trace distance between Choi states captures average-case distinguishability, the diamond distance captures worst-case distinguishability, thereby making the approximation requirement in channel conversion strictly more demanding.
At the same time, our framework removes the asymptotic continuity requirement imposed in Ref.~\cite{hayashi2025generalizedquantumsteinslemma}, so that the asymptotic resource-non-generating property now becomes the only condition on free operations, directly paralleling the reversible framework for static resources originally proposed in Refs.~\cite{Brand_o_2008,brandao2010reversible,Brandao2010}.
Taken together, this framework embodies a nontrivial trade-off: it requires addressing a strictly harder channel-conversion task under a more relaxed class of operations, making it a priori unclear whether the achievable conversion rates should be higher or lower than those in the earlier framework of Ref.~\cite{hayashi2025generalizedquantumsteinslemma}.
Nevertheless, we prove that the optimal asymptotic conversion rate between CQ channels in our framework is exactly given by the regularized relative entropy of resource~\cite{Gour2019a}, defined here through channel divergence~\cite{cooney2016strong}.
This establishes a reversible framework for channel conversion that is both conceptually stronger and practically more applicable than previous approaches, directly accommodating conventional channel coding scenarios where input optimization is essential.

As an application, we derive bounds on channel capacities and conversion rates of CQ channels by applying our framework to the case where the set of non-resources consists of replacer channels, which by construction have zero capacity.
In this setting, the resulting reversible framework not only recovers known capacity bounds but also serves as a natural extension of no-signaling, entanglement-assisted, and randomness-assisted scenarios, without requiring additional assumptions on asymptotic continuity.
Overall, this advances the theory of reversible QRTs from static to dynamical resources and establishes a general toolkit for analyzing discrimination and conversion in this fundamental class of channels.

Finally, we emphasize the significance of focusing on CQ channels as a fundamental and tractable class of dynamical resources.
Their classical inputs allow us to extend proof techniques beyond the state setting of the generalized quantum Stein's lemma, while still capturing nontrivial quantum features at the channel output, making them rich enough to model quantum communication scenarios yet sufficiently structured to enable rigorous analysis.
By contrast, in fully quantum settings with QQ channels, it remains unclear whether analogous properties hold at all, and addressing this more general case is left as a natural open question for future work.
Nevertheless, as our results demonstrate, CQ channels serve as a natural bridge between static resources and the full generality of dynamical resources; they recover the state-based results~\cite{hayashi2025generalizedquantumsteinslemma} as the special case of a single channel input, while also encompassing classical channels as further special cases, thereby establishing a unified treatment of discrimination and conversion tasks for states and channels across these settings.
Our results therefore provide not only a powerful tool for analyzing the conversion problems for CQ channels, but also a robust theoretical foundation that clarifies the role of this fundamental class of channels in the broader landscape of QRTs for static and dynamical resources.

\paragraph*{Organization of this paper}

The remainder of this paper is structured as follows.
In Sec.~\ref{sec:QRTs}, we formulate QRTs for CQ channels and introduce the assumptions underlying our setting.
In Sec.~\ref{sec:stein}, we analyze quantum hypothesis testing for CQ dynamical resources and prove the generalized quantum Stein's lemma for CQ channels, characterizing its optimal error exponent.
In Sec.~\ref{sec:reversible}, we present the reversible framework of QRTs for converting CQ channels and characterize the optimal asymptotic conversion rate in this framework.
In Sec.~\ref{sec:application}, we demonstrate applications of the reversible framework to analyzing channel capacities.
Finally, in Sec.~\ref{sec:conclusion} concludes with a summary and outlook.

\section{CQ dynamical resources}
\label{sec:QRTs}

In this section, we introduce QRTs for CQ channels, i.e., CQ dynamical resources, and the notations used in our analysis.
For the basic concepts of quantum information theory, we refer readers to the standard textbooks~\cite{hayashi2016quantum,Holevo+2019,watrous_2018,Wilde_2017}.
In Sec.~\ref{sec:CQ_channels}, we define CQ channels along with the relevant distance measures and divergences.
In Sec.~\ref{sec:superchannels_of_CQ_channels}, we introduce a class of superchannels that convert CQ channels into CQ channels.
Building on this, in Sec.~\ref{sec:QRTs_for_CQ_channels}, we formulate QRTs for CQ channels and present a set of axioms that we will assume throughout our analysis.
Finally, in Sec.~\ref{sec:properties_channel_divergence}, we present properties of channel divergences for CQ channels, showing that under these axioms the regularized relative entropy of resource converges.

\subsection{CQ channels}
\label{sec:CQ_channels}

Let $\mathcal{X}$ be a finite set of alphabets, with its cardinality denoted by $\qty|\mathcal{X}|$.
A quantum system is represented as a finite-dimensional Hilbert space $\mathcal{H}$, with its dimension denoted by $\dim\qty(\mathcal{H})$.
The set of linear operators on $\mathcal{H}$ is denoted by $\mathcal{L}\qty(\mathcal{H})$.
The identity operator on $\mathcal{H}$ is denoted by $\mathds{1}_\mathcal{H}$, which we may write $\mathds{1}$ if the space $\mathcal{H}$ it acts on is obvious from the context.

A quantum state is represented a density operator on $\mathcal{H}$, and the set of density operators on $\mathcal{H}$ is denoted by
\begin{align}
    \mathcal{D}\qty(\mathcal{H})\coloneqq\qty{\rho\in\mathcal{L}\qty(\mathcal{H}):\rho\geq0,\Tr\qty[\rho]=1}.
\end{align}
Following the convention in Refs.~\cite{hayashi2016quantum,Holevo+2019}, we define a CQ channel as a map $\Phi:\mathcal{X}\to\mathcal{D}\qty(\mathcal{H})$.
We may write the input and output as
\begin{align}
    \label{eq:Phi}
    x\in\mathcal{X},~\Phi(x)\in\mathcal{D}\qty(\mathcal{H}).
\end{align}
The dimensions of the input and output spaces are denoted by
\begin{align}
    X=|\mathcal{X}|,~D=\dim\qty(\mathcal{H}),
    \label{eq:X}
\end{align}
which are assumed to be finite throughout our work.
The set of CQ channels is denoted by
\begin{align}
    \mathcal{C}\qty(\mathcal{X}\to\mathcal{H})\coloneqq\qty{\Phi:\mathcal{X}\to\mathcal{D}\qty(\mathcal{H})}.
\end{align}
Note that in Ref.~\cite{hayashi2025generalizedquantumsteinslemma}, CQ channels are defined as measure-and-prepare channels, which are a special case of QQ channels, but with $\mathcal{N}$ denoting such measure-and-prepare QQ channel, the above definition~\eqref{eq:Phi} of CQ channels is equivalent by considering $\mathcal{N}\qty(\rho)=\sum_{x\in\mathcal{X}}\bra{x}\rho\ket{x}\Phi(x)$.

For two states $\rho_1,\rho_2\in\mathcal{D}\qty(\mathcal{H})$, the trace distance is defined as
\begin{align}
\label{eq:d_trace}
    d_{\mathrm{T}}\qty(\rho_1,\rho_2)\coloneqq\frac{1}{2}\left\|\rho_1-\rho_2\right\|_1.
\end{align}
For two CQ channels $\Phi_1,\Phi_2:\mathcal{X}\to\mathcal{D}\qty(\mathcal{H})$, the diamond distance~\cite{10.1145/276698.276708} is defined as
\begin{widetext}
\begin{align}
\label{eq:d_diamond}
    d_{\diamond}\qty(\Phi_1,\Phi_2)&\coloneqq\max_{x\in\mathcal{X}}d_{\mathrm{T}}\qty(\Phi_1\qty(x),\Phi_2\qty(x))\\
    &=\max_{p}d_{\mathrm{T}}\left(\sum_{x\in\mathcal{X}}p(x)\ket{x}\bra{x}\otimes\Phi_1\qty(x),\sum_{x\in\mathcal{X}}p(x)\ket{x}\bra{x}\otimes\Phi_2\qty(x)\right),
\end{align}
\end{widetext}
where $\max_p$ denotes maximization over all probability distributions $p$ over the set of channel inputs.
This makes $\mathcal{C}\qty(\mathcal{X}\to\mathcal{H})$ a metric space in terms of $d_\diamond$.

For two states $\rho_1,\rho_2\in\mathcal{D}\qty(\mathcal{H})$, the quantum relative entropy~\cite{10.2996/kmj/1138844604} is defined as
\begin{align}
    \label{eq:D}
    D\left(\rho_1\middle\|\rho_2\right)\coloneqq\Tr\qty[\rho_1\qty(\log\qty[\rho_1]-\log\qty[\rho_2])],
\end{align}
where $\log$ is the natural logarithm throughout this work, and the right-hand side is considered $\infty$ unless their supports satisfy $\supp\qty(\rho_1)\subseteq\supp\qty(\rho_2)$.
For two CQ channels $\Phi_1,\Phi_2:\mathcal{X}\to\mathcal{D}\qty(\mathcal{H})$, the channel divergence~\cite{cooney2016strong} is defined as
\begin{widetext}
\begin{align}
    \label{eq:D_channel}
    D\left(\Phi_1\middle\|\Phi_2\right)&\coloneqq\max_{x\in\mathcal{X}}D\left(\Phi_1\qty(x)\middle\|\Phi_2\qty(x)\right)\\
    &=\max_{p}D\left(\sum_{x\in\mathcal{X}}p(x)\ket{x}\bra{x}\otimes\Phi_1\qty(x)\middle\|\sum_{x\in\mathcal{X}}p(x)\ket{x}\bra{x}\otimes\Phi_2\qty(x)\right),
\end{align}
\end{widetext}
where the right-hand side is considered $\infty$ unless we have for all $x\in\mathcal{X}$
\begin{align}
\supp\qty(\Phi_1\qty(x))\subseteq\supp\qty(\Phi_2\qty(x)).
\end{align}
Similarly, for any parameter $\alpha>1$, the sandwiched R\'enyi relative entropy~\cite{10.1063/1.4838856,wilde2014strong} is defined as
\begin{align}
    \label{eq:D_alpha}
    \widetilde{D}_\alpha\left(\rho_1\middle\|\rho_2\right)\coloneqq-\frac{1}{1-\alpha}\log\qty[\Tr\qty[\qty(\rho_2^{\frac{1-\alpha}{2\alpha}}\rho_1\rho_2^{\frac{1-\alpha}{2\alpha}})^\alpha]],
\end{align}
and the sandwiched R\'enyi channel divergence as
\begin{align}
    \label{eq:D_alpha_channel}
    \widetilde{D}_\alpha\left(\Phi_1\middle\|\Phi_2\right)\coloneqq\max_{x\in\mathcal{X}}\widetilde{D}_\alpha\left(\Phi_1\qty(x)\middle\|\Phi_2\qty(x)\right).
\end{align}
Note that the sandwiched R\'enyi relative entropy can be defined for a wider parameter region of $\alpha$, but we define it within a parameter region relevant to our analysis.
For the fixed inputs, the sandwiched R\'enyi relative entropy and the sandwiched R\'enyi channel divergence do not decrease as $\alpha$ increases; also, in the limit, we obtain~\cite{10.1063/1.4838856,wilde2014strong}
\begin{align}
    \label{eq:D_alpha_limit}
    \lim_{\alpha\to 1}\widetilde{D}_\alpha\left(\Phi_1\middle\|\Phi_2\right)&=D\left(\Phi_1\middle\|\Phi_2\right).
\end{align}
Note that our analysis would also work if we use the Petz-R\'enyi relative entropy
$D_\alpha\left(\rho_1\middle\|\rho_2\right)\coloneqq\frac{1}{\alpha-1}\log\Tr\qty[\rho_1^\alpha\rho_2^{1-\alpha}]$~\cite{PETZ198657} instead of the sandwiched R\'enyi relative entropy; however, we here use the sandwiched R\'enyi relative entropy since it provides a tighter bound than the Petz-R\'enyi relative entropy in the relevant parameter region due to the Araki-Lieb-Thirring inequality~\cite{lieb1976inequalities,araki1990inequality}.

For any CQ channels $\Phi_1,\Phi_{1,\delta},\Phi_2\in\mathcal{C}\qty(\mathcal{X}\to\mathcal{H})$ satisfying
\begin{align}
    \lim_{\delta\to 0}d_\diamond\qty(\Phi_{1,\delta},\Phi_1)&=0,\\
    \supp\qty(\Phi_1\qty(x))&\subseteq\supp\qty(\Phi_2\qty(x))~\text{for all $x\in\mathcal{X}$},\\
    \supp\qty(\Phi_{1,\delta}\qty(x))&\subseteq\supp\qty(\Phi_2\qty(x))~\text{for all $x\in\mathcal{X}$ and $\delta>0$},
\end{align}
the channel divergence $D$ in~\eqref{eq:D_channel} satisfies the continuity with respect to the first argument
\begin{align}
\label{eq:D_continuity}
    \lim_{\delta\to 0}\qty|D\left(\Phi_{1,\delta}\middle\|\Phi_2\right)-D\left(\Phi_{1}\middle\|\Phi_2\right)|=0.
\end{align}
due to the continuity bound of the quantum relative entropy with respect to the first argument~\cite{10129917,10504886}.

\subsection{Superchannels of CQ channels}
\label{sec:superchannels_of_CQ_channels}

We define superchannels that convert CQ channels into CQ channels.
In Ref.~\cite{hayashi2025generalizedquantumsteinslemma}, superchannels that convert measure-and-prepare QQ channels into measure-and-prepare QQ channels are formulated as a subclass of superchannels for QQ channel conversion, which can be represented as a quantum comb~\cite{PhysRevLett.101.060401,Chiribella_2008,PhysRevA.80.022339}.
Here, we provide an equivalent formulation of the subclass of superchannels for convertion between CQ channels defined as~\eqref{eq:Phi}.

To define a superchannel, we consider a supermap $\Theta$ that converts an input map $\Phi_\mathrm{in}:\mathcal{X}_\mathrm{in}\to\mathcal{L}\qty(\mathcal{H}_\mathrm{in})$ to an output map $\Phi_\mathrm{out}=\Theta[\Phi_\mathrm{in}]:\mathcal{X}_\mathrm{out}\to\mathcal{L}\qty(\mathcal{H}_\mathrm{out})$.
For $N$ maps $\Phi_1,\ldots,\Phi_N:\mathcal{X}\to\mathcal{L}\qty(\mathcal{H})$, their linear combination $\sum_{n=1}^N \alpha(n)\Phi_n:\mathcal{X}\to\mathcal{L}\qty(\mathcal{H})$ with coefficients $\alpha(n)\in\mathbb{C}$ is defined as a map satisfying for all $x\in\mathcal{X}$
\begin{align}
\label{eq:linear_combination}
    \qty(\sum_{n=1}^N \alpha(n)\Phi_n)\qty(x)\coloneqq\sum_{n=1}^N \alpha(n)\Phi_n(x).
\end{align}
A supermap $\Theta$ is said to be linear if it satisfies
\begin{align}
    \Theta\qty[\sum_{n=1}^N \alpha(n)\Phi_n]=\sum_{n=1}^N \alpha(n)\Theta[\Phi_n].
\end{align}
For $N$ maps $\Phi_1:\mathcal{X}_1\to\mathcal{L}\qty(\mathcal{H}_1),\ldots,\Phi_N:\mathcal{X}_N\to\mathcal{L}\qty(\mathcal{H}_N)$, their tensor product is defined as a map $\bigotimes_{n=1}^N\Phi_n:\mathcal{X}_1\times\cdots\times\mathcal{X}_N\to\mathcal{L}\qty(\bigotimes_{n=1}^{N}\mathcal{H}_n)$ satisfying for any input $x^{(N)}\coloneqq\qty(x_1,\ldots,x_N)\in\mathcal{X}_1\times\cdots\times\mathcal{X}_N$
\begin{align}
    \label{eq:tensor_product_CQ_channel}
    \qty(\bigotimes_{n=1}^N \Phi_n)\qty(x^{(N)})\coloneqq\bigotimes_{n=1}^N \Phi_n\qty(x_n).
\end{align}
Using the same notation, the tensor product $\bigotimes_{n=1}^N \Theta_n$ of linear supermaps $\Theta_1,\ldots,\Theta_N$ is defined as a map satisfying
\begin{align}
    \qty(\bigotimes_{n=1}^N \Theta_n)\qty[\bigotimes_{n=1}^N \Phi_n]\coloneqq\bigotimes_{n=1}^N \qty(\Theta_n\qty[\Phi_n]).
\end{align}
The action of $\bigotimes_{n=1}^N \Theta_n$ extends to any map $\Phi^{(N)}:\mathcal{X}^N\to\mathcal{L}\qty(\mathcal{H}^{\otimes N})$ due to the linearity of $\bigotimes_{n=1}^N \Theta_n$, since $\Phi^{(N)}\qty(x^{(N)})$ for every input $x^{(N)}\in\mathcal{X}^{(N)}$ can be represented as a linear combination of the tensor product of $N$ linear operators acting on $\mathcal{H}$.

As in the superchannels for converting QQ channels in Refs.~\cite{PhysRevLett.101.060401,Chiribella_2008,PhysRevA.80.022339}, we write the set of superchannels that convert CQ channels from $\mathcal{C}\qty(\mathcal{X}_\mathrm{in}\to\mathcal{H}_\mathrm{in})$ to $\mathcal{C}\qty(\mathcal{X}_\mathrm{out}\to\mathcal{H}_\mathrm{out})$ as $\mathcal{C}\qty(\qty(\mathcal{X}_\mathrm{in}\to\mathcal{H}_\mathrm{in})\to\qty(\mathcal{X}_\mathrm{in}\to\mathcal{H}_\mathrm{out}))$, where each superchannel
\begin{align}
\label{eq:C}
    \Theta\in\mathcal{C}\qty(\qty(\mathcal{X}_\mathrm{in}\to\mathcal{H}_\mathrm{in})\to\qty(\mathcal{X}_\mathrm{out}\to\mathcal{H}_\mathrm{out}))
\end{align}
is a linear supermap that converts any input CQ channel $\Phi_\mathrm{in}\in\mathcal{C}\qty(\mathcal{X}_\mathrm{in}\times\mathcal{X}_\mathrm{aux}\to\mathcal{H}_\mathrm{in}\otimes\mathcal{H}_\mathrm{aux})$ to an output CQ channel $\Phi_\mathrm{out}=\qty(\Theta\otimes\id)\qty[\Phi_\mathrm{in}]\in\mathcal{C}\qty(\mathcal{X}_\mathrm{out}\times\mathcal{X}_\mathrm{aux}\to\mathcal{H}_\mathrm{out}\otimes\mathcal{H}_\mathrm{aux})$, with $\mathcal{X}_\mathrm{aux}$ and $\mathcal{H}_\mathrm{aux}$ representing the classical and quantum auxiliary systems, and $\id$ denoting the identity supermap that converts any map $\Phi_\mathrm{aux}:\mathcal{X}_\mathrm{aux}\to\mathcal{L}\qty(\mathcal{H}_\mathrm{aux})$ to $\Phi_\mathrm{aux}$ itself.
As in the superchannel for QQ channels~\cite{Chiribella_2008,PhysRevA.80.022339},
the superchannels for CQ channels can be written in the form of
\begin{align}
    &\qty(\qty(\Theta\otimes\id)\qty[\Phi_\mathrm{in}])\qty(x_\mathrm{out},x_\mathrm{aux})\notag\\
    \label{eq:theta}
    &=\sum_{x_\mathrm{in}\in\mathcal{X}_\mathrm{in}}p_\Theta\qty(x_\mathrm{in}|x_\mathrm{out})\qty(\qty(\mathcal{N}_{\Theta,x_\mathrm{in},x_\mathrm{out}}\otimes\id)\circ\Phi_\mathrm{in})\qty(x_\mathrm{in},x_\mathrm{aux}),
\end{align}
where $x_\mathrm{in}\in\mathcal{X}_\mathrm{in}$, $x_\mathrm{out}\in\mathcal{X}_\mathrm{out}$, and $x_\mathrm{aux}\in\mathcal{X}_\mathrm{aux}$ are classical inputs to $\Phi_\mathrm{in}$ and $\Phi_\mathrm{out}$, $p_\Theta\qty(x_\mathrm{in}|x_\mathrm{out})$ is a conditional probability distribution, $\mathcal{N}_{\Theta,x_\mathrm{in},x_\mathrm{out}}:\mathcal{L}\qty(\mathcal{H}_\mathrm{in})\to\mathcal{L}\qty(\mathcal{H}_\mathrm{out})$ is a completely positive and trace-preserving (CPTP) linear map depending on $x_\mathrm{in}$ and $x_\mathrm{out}$, and $\id:\mathcal{L}\qty(\mathcal{H}_\mathrm{aux})\to\mathcal{L}\qty(\mathcal{H}_\mathrm{aux})$ is the identity map.
The dimensions of the input and output spaces are denoted by
\begin{align}
\label{eq:dim}
    &X_\mathrm{in}=\left|\mathcal{X}_\mathrm{in}\right|,~X_\mathrm{out}=\left|\mathcal{X}_\mathrm{out}\right|,\notag\\
    &D_\mathrm{in}=\dim\qty(\mathcal{H}_\mathrm{in}),~D_\mathrm{out}=\dim\qty(\mathcal{H}_\mathrm{out}),
\end{align}
and all dimensions are assumed to be finite throughout our work as in~\eqref{eq:X}.

For any superchannel $\Theta$ for converting CQ channels, the channel divergence defined in~\eqref{eq:D_channel} satisfies the monotonicity
\begin{align}
\label{eq:monotonicity_D_channel}
    D\left(\Theta[\Phi_1]\middle\|\Theta[\Phi_2]\right)\leq D\left(\Phi_1\middle\|\Phi_2\right),
\end{align}
which follows from the monotonicity of the quantum relative entropy under CPTP linear maps along with the form~\eqref{eq:theta} of $\Theta$.
The diamond distance in~\eqref{eq:d_diamond} also satisfies the monotonicity
\begin{align}
\label{eq:monotonicity_diamond}
    d_\diamond\qty(\Theta[\Phi_1],\Theta[\Phi_2])\leq d_\diamond\qty(\Phi_1,\Phi_2),
\end{align}
which also follows from the monotonicity of the trace distance under CPTP linear maps along with the form~\eqref{eq:theta} of $\Theta$.
The diamond distance between CQ channels is convex; that is, for any $p\in[0,1]$ and any CQ channels $\Phi_1,\Phi_1',\Phi_2,\Phi_2'\in\mathcal{C}\qty(\mathcal{X}\to\mathcal{H})$, we have
\begin{align}
\label{eq:d_diamond_convexity}
    &d_\diamond\qty(p\Phi_1+(1-p)\Phi_1',p\Phi_2+(1-p)\Phi_2')\notag\\
    &\leq pd_\diamond\qty(\Phi_1,\Phi_2)+\qty(1-p)d_\diamond\qty(\Phi_1',\Phi_2'),
\end{align}
which follows from the convexity of the trace distance between quantum states.

\subsection{QRTs for CQ dynamical resources}
\label{sec:QRTs_for_CQ_channels}

QRTs are formulated by specifying free operations~\cite{Kuroiwa2020,Chitambar2018}.
We take a set of free operations as a subset of superchannels converting CQ channels.
In particular, the set of free operations are denoted by
\begin{align}
\label{eq:O}
    &\mathcal{O}\qty(\qty(\mathcal{X}_\mathrm{in}\to\mathcal{H}_\mathrm{in})\to\qty(\mathcal{X}_\mathrm{in}\to\mathcal{H}_\mathrm{out}))\notag\\
    &\subseteq\mathcal{C}\qty(\qty(\mathcal{X}_\mathrm{in}\to\mathcal{H}_\mathrm{in})\to\qty(\mathcal{X}_\mathrm{in}\to\mathcal{H}_\mathrm{out})),
\end{align}
where $\mathcal{C}$ is given by~\eqref{eq:C}, and the left-hand side may be written as $\mathcal{O}$ if the arguments are obvious from the context.
Given $\mathcal{O}$, a family of sets of free CQ channels is specified by the CQ channels that can be obtained from any (non-resourceful) CQ channels by some operation in $\mathcal{O}$; in particular, we write this family as
\begin{align}
\label{eq:F}
    \mathcal{F}\qty(\mathcal{X}\to\mathcal{H})&\coloneqq\{\Phi\in\mathcal{C}\qty(\mathcal{X}\to\mathcal{H}):\notag\\
    &\quad\forall \Phi_\mathrm{in}\in\mathcal{C}\qty(\mathcal{X}_\mathrm{in}\to\mathcal{H}_\mathrm{in}),\notag\\
    &\quad\exists\Theta\in\mathcal{O}\qty(\qty(\mathcal{X}_\mathrm{in}\to\mathcal{H}_\mathrm{in})\to\qty(\mathcal{X}\to\mathcal{H})),\notag\\
    &\quad\Phi=\Theta[\Phi_\mathrm{in}]\},
\end{align}
which we may write $\mathcal{F}$ if the argument specifying each set is obvious from the context.

Using $\mathcal{F}$, we can define various resource measures~\cite{Chitambar2018}.
For a CQ channel $\Phi\in\mathcal{C}\qty(\mathcal{X}\to\mathcal{H})$, the relative entropy of resource is denoted by
\begin{align}
\label{eq:D_F}
    D\left(\Phi\middle\|\mathcal{F}\right)\coloneqq\min_{\Phi_\mathrm{free}\in\mathcal{F}}D\left(\Phi\middle\|\Phi_\mathrm{free}\right),
\end{align}
where $D$ is defined as~\eqref{eq:D_channel}.
We may also write this as a function of $\Phi$
\begin{align}
\label{eq:relative_entropy_of_resource}
    R_\mathrm{R}\qty(\Phi)\coloneqq D\left(\Phi\middle\|\mathcal{F}\right).
\end{align}
For any family $\mathcal{F}$ of sets of free CQ channels satisfying Axioms~\ref{p:full_rank},~\ref{p:compact}, and~\ref{p:tensor_product}, and~\ref{p:convex}, and any sequences $\qty{\Phi^{(n)}\in\mathcal{C}\qty(\mathcal{X}^n\to\mathcal{H}^{\otimes n})}_{n=1,2,\ldots}$ and $\qty{\Phi^{(n)\prime}\in\mathcal{C}\qty(\mathcal{X}^n\to\mathcal{H}^{\otimes n})}_{n=1,2,\ldots}$ of CQ channels satisfying
\begin{align}
    &\lim_{n\to \infty}d_\diamond\qty(\Phi^{(n)},\Phi^{(n)\prime})=0,
\end{align}
the relative entropy of resource satisfies the asymptotic continuity
\begin{align}
\label{eq:D_asymptotic_continuity}
    \liminf_{n\to 0}\frac{1}{n}R_\mathrm{R}\qty(\Phi^{(n)})&=\liminf_{n\to 0}\frac{1}{n}R_\mathrm{R}\qty(\Phi^{(n)\prime}),\notag\\
    \limsup_{n\to 0}\frac{1}{n}R_\mathrm{R}\qty(\Phi^{(n)})&=\limsup_{n\to 0}\frac{1}{n}R_\mathrm{R}\qty(\Phi^{(n)\prime}),
\end{align}
which follows from the argument in Ref.~\cite[Lemma~7]{Winter2016}, by replacing $D_C$ in Ref.~\cite{Winter2016} with $R_\mathrm{R}$ and states in Ref.~\cite{Winter2016} with CQ channels, and then taking the limit. 
The regularized relative entropy of resource is defined as
\begin{align}
\label{eq:regularized_relative_entropy_of_resource}
    R_\mathrm{R}^\infty\qty(\Phi)\coloneqq \lim_{n\to\infty}\frac{1}{n}D\left(\Phi^{\otimes n}\middle\|\mathcal{F}\right).
\end{align}
Using $\widetilde{D}_\alpha$ in~\eqref{eq:D_alpha_channel} instead of $D$, we also define
\begin{align}
\label{eq:D_alpha_F}
    \widetilde{D}_\alpha\left(\Phi\middle\|\mathcal{F}\right)&\coloneqq\min_{\Phi_\mathrm{free}\in\mathcal{F}}\widetilde{D}_\alpha\left(\Phi\middle\|\Phi_\mathrm{free}\right).
\end{align}
The generalized robustness is defined as
\begin{align}
\label{eq:R_G}
    R_\mathrm{G}\qty(\Phi)\coloneqq\min\qty{s:\frac{\Phi+s\Phi'}{1+s}\in\mathcal{F},\Phi'\in\mathcal{C}\qty(\mathcal{X}\to\mathcal{H})}.
\end{align}

We introduce the following axioms on QRTs, which are imposed on $\mathcal{F}$ while $\mathcal{F}$ is determined from $\mathcal{O}$ by~\eqref{eq:F}.
\begin{enumerate}[label={CQ\arabic*}]
    \item\label{p:full_rank}For any CQ channel $\Phi\in\mathcal{F}\qty(\mathcal{X}\to\mathcal{H})$, there exists a free CQ channel $\Phi_\mathrm{full}\in\mathcal{F}\qty(\mathcal{X}\to\mathcal{H})$ such that, for every input $x\in\mathcal{X}$, the output $\Phi_\mathrm{full}(x)$ satisfies the relation $\supp\qty(\Phi\qty(x))\subseteq\supp\qty(\Phi_\mathrm{full}\qty(x))$ on their supports; equivalently, each set $\mathcal{F}\qty(\mathcal{X}\to\mathcal{H})$ includes the free CQ channel $\Phi_\mathrm{full}$ such that, for every input $x\in\mathcal{X}$, the output $\Phi_\mathrm{full}(x)$ is a full-rank state $\supp\qty(\Phi_\mathrm{full}(x))=\mathcal{H}$.
    \item\label{p:compact}Each set $\mathcal{F}\qty(\mathcal{X}\to\mathcal{H})$ is compact.
    \item\label{p:tensor_product}The family $\mathcal{F}$ of sets is closed under the tensor product in~\eqref{eq:tensor_product_CQ_channel}; that is, if $\Phi_\mathrm{free}\in\mathcal{F}\qty(\mathcal{X}\to\mathcal{H})$ and $\Phi_\mathrm{free}'\in\mathcal{F}\qty(\mathcal{X}'\to\mathcal{H}')$, then $\Phi_\mathrm{free}\otimes\Phi_\mathrm{free}^\prime\in\mathcal{F}\qty(\mathcal{X}\times\mathcal{X}'\to\mathcal{H}\otimes\mathcal{H}^\prime)$.
    \item\label{p:convex}Each set $\mathcal{F}\qty(\mathcal{X}\to\mathcal{H})$ is convex with respect to the convex combination in the form of~\eqref{eq:linear_combination} with $\alpha$ being a probability distribution.
\end{enumerate}
As discussed in Refs.~\cite{Kuroiwa2020,Chitambar2018}, it would also be conventional to additionally impose axioms directly on operations in $\mathcal{O}$ to ensure that QRTs are physically well-motivated, such as the requirement that the composition of multiple operations in $\mathcal{O}$ remains in $\mathcal{O}$.
From this requirement, an essential property of QRTs follows: free operations in $\mathcal{O}$ always map free CQ channels in $\mathcal{F}$ to free CQ channels in $\mathcal{F}$ and, therefore, never generate resources from any free CQ channel.
However, in our analysis, rather than focusing solely on $\mathcal{O}$ under such axioms, we will also introduce a relaxed version $\tilde{O}$ of free operations, using the resource-non-generating property of free operations as a guiding principle, as discussed in more detail in Sec.~\ref{sec:reversible}.

The full-rank condition in Axiom~\ref{p:full_rank} ensures that for all CQ channel $\Phi\in\mathcal{C}\qty(\mathcal{X}\to\mathcal{H})$, the relative entropy of resource $R_\mathrm{R}\qty(\Phi)$ in~\eqref{eq:relative_entropy_of_resource} and their variant in~\eqref{eq:D_alpha_F} with the sandwiched R\'enyi channel divergence do not diverge to infinity, making the equations appearing our analysis well-defined throughout.
Axiom~\ref{p:full_rank} is equivalent to ensuring that there exists a positive real constant $\Lambda_{\min}\in(0,1]$, along with $\Phi_\mathrm{full}\in\mathcal{F}\qty(\mathcal{X}\to\mathcal{H})$, such that, for every $x\in\mathcal{X}$,
\begin{align}
\label{eq:Lambda_min}
    \Phi_\mathrm{full}\qty(x)\geq\Lambda_{\min}\mathds{1}.
\end{align}
The compactness in Axiom~\ref{p:compact} guarantees that the minimum in the definition~\eqref{eq:D_F} of $R_\mathrm{R}$ exists.
As we will show below in Sec.~\ref{sec:properties_channel_divergence}, Axiom~\ref{p:tensor_product} is essential for the existence of the limit in the definition~\eqref{eq:regularized_relative_entropy_of_resource} of $R_\mathrm{R}^\infty$.
Finally, Axiom~\ref{p:convex} is necessary for the generalized quantum Stein's lemma to hold, due to the counterexamples with a nonconvex QRT for static resources shown in Ref.~\cite{hayashi2025generalizedquantumsteinslemma}.

When the input dimension of CQ channels is one, the QRTs for CQ channels reduce to those for states, leading to the state version of Axioms~\ref{p:full_rank},~\ref{p:compact},~\ref{p:tensor_product}, and~\ref{p:convex}.
In the original proposal of the generalized quantum Stein's lemma in Ref.~\cite{Brandao2010}, two additional axioms were imposed on $\mathcal{F}$, namely, closedness under partial trace and closedness under permutation of subsystems.
We also note that Ref.~\cite{Brandao2010} originally presented Axioms~\ref{p:compact} and~\ref{p:convex} jointly as a single axiom, but here we distinguish them since some of our results depend only on one of these conditions.
Assuming the generalized quantum Stein's lemma, Ref.~\cite{Brandao2015} then developed a reversible QRT framework for static resources, relying on the same set of axioms as Ref.~\cite{Brandao2010}.
Note that Ref.~\cite{Brandao2015} does not explicitly mention the full-rank condition in Axiom~\ref{p:full_rank}, but it is necessary for their arguments.
In Ref.~\cite{hayashi2025generalizedquantumsteinslemma}, the generalized quantum Stein's lemma for states is proven under Axioms~\ref{p:full_rank},~\ref{p:compact},~\ref{p:tensor_product}, and~\ref{p:convex}, while Ref.~\cite{10898013} provides an alternative proof of the weaker, original version of the lemma by also using the two additional axioms in Ref.~\cite{Brandao2010}.
Note that Refs.~\cite{fang2025generalizedquantumasymptoticequipartition,fang2025errorexponentsquantumstate} provide a similar yet different generalization of quantum Stein's lemma, but their analysis do not apply to the original setting of the generalized quantum Stein's lemma, especially in the case of entanglement theory as originally envisioned in Refs.~\cite{Brand_o_2008,brandao2010reversible,Brandao2010}, since they require an additional assumption on  stability of polar sets under tensor product.
In contrast, our analysis builds upon a CQ-channel extension of the minimal set of axioms under which the generalized quantum Stein's lemma for states was proven in Ref.~\cite{hayashi2025generalizedquantumsteinslemma}, thereby advancing from static resources to the fundamental class of dynamical resources without imposing any stronger assumptions.

We consider the task of CQ channel conversion under a class of superchannels.
Under $\mathcal{O}$, this task involves converting many copies of a CQ channel $\Phi_\mathrm{in}$ into as many copies as possible of another CQ channel $\Phi_\mathrm{out}$ using operations in $\mathcal{O}$, within errors that vanish asymptotically.
In our work, the errors are measured in terms of the diamond distance $d_\diamond$ in~\eqref{eq:d_diamond}.
The conversion rate $r_\mathcal{O}\qty(\Phi_\mathrm{in}\to\Phi_\mathrm{out})$ under $\mathcal{O}$ is the supremum of achievable rates in this asymptotic conversion, i.e.,
\begin{align}
\label{eq:conversion_rate_O}
    &r_\mathcal{O}\qty(\Phi_\mathrm{in}\to\Phi_\mathrm{out})\coloneqq\sup\left\{r\geq 0:\notag\right.\\
    &\left.\exists\qty{\Theta^{(n)}\in\mathcal{O}}_n,\liminf_{n\to\infty}d_\diamond\qty(\Theta^{(n)}\qty(\Phi_\mathrm{in}^{\otimes n}),\Phi_\mathrm{out}^{\otimes \lceil rn\rceil})=0\right\},
\end{align}
where $\lceil{}\cdots{}\rceil$ is the ceiling function, and the diamond distance $d_\diamond$ is given by~\eqref{eq:d_diamond}.

A fundamental problem in QRTs is when the conversion rate has a simple characterization by a single function representing a resource measure~\cite{Brand_o_2008,brandao2010reversible,Brandao2010,Brandao2015,hayashi2025generalizedquantumsteinslemma}; if this is the case, any two resources that have the same amount of resource should always be convertible into each other, and we call such a framework of QRTs a reversible framework.
To address this issue, as in the previous works~\cite{Brand_o_2008,brandao2010reversible,Brandao2010,Brandao2015,hayashi2025generalizedquantumsteinslemma}, we will consider a broader, axiomatically defined class $\tilde{\mathcal{O}}$ of operations as a relaxation of $\mathcal{O}$.
A fundamental question in QRTs is whether it is possible to establish a general reversible framework of QRTs with an appropriate choice of such a class $\tilde{\mathcal{O}}$ of operations and a single resource measure $R$ such that the resource measure, i.e., $R(\Phi_\mathrm{in})$ and  $R(\Phi_\mathrm{out})$, fully determines the convertibility from $\Phi_\mathrm{in}$ to $\Phi_\mathrm{out}$ at a given conversion rate, which is called the second law of QRTs~\cite{Brand_o_2008,brandao2010reversible,Brandao2010,Brandao2015,hayashi2025generalizedquantumsteinslemma}.

\subsection{Properties of channel divergences for CQ channels}
\label{sec:properties_channel_divergence}

We show properties of channel divergences for CQ channels that are relevant to our proof of the generalized quantum Stein's lemma for CQ channels; especially, we will show the existence of the limit in the definition~\eqref{eq:regularized_relative_entropy_of_resource} of $R_\mathrm{R}^\infty$ under Axioms~\ref{p:full_rank},~\ref{p:compact}, and~\ref{p:tensor_product}.
Importantly, the argument in this section relies critically on the fact that the channel inputs are classical, making it nontrivial in a sense that the same properties cannot be generally shown for QQ channels in the same way.

A crucial property for our proof, especially for the strong converse part of the generalized quantum Stein's lemma, will be the additivity of channel divergences for CQ channels.
For states $\rho_1,\rho_1',\rho_2,\rho_2'$, it is known that the sandwidged R\'{e}nyi relative entropy satisfies the additivity~\cite{10.1063/1.4838856,wilde2014strong}
\begin{align}
    \label{eq:additivity_state}
    \widetilde{D}_\alpha\left(\rho_1\otimes\rho_1'\middle\|\rho_2\otimes\rho_2'\right)=\widetilde{D}_\alpha\left(\rho_1\middle\|\rho_2\right)+\widetilde{D}_\alpha\left(\rho_1'\middle\|\rho_2'\right).
\end{align}
The following lemma shows that this additivity extends to CQ channels as well.

\begin{lemma}[\label{lem:additivity_CQ_channels}Additivity of channel divergences for CQ channels]
    For any CQ channels $\Phi_1,\Phi_2\in\mathcal{C}\qty(\mathcal{X}\to\mathcal{H})$ and $\Phi_1',\Phi_2'\in\mathcal{C}\qty(\mathcal{X}'\to\mathcal{H}')$, we have the additivity
    \begin{align}
        D\left(\Phi_1\otimes\Phi_1'\middle\|\Phi_2\otimes\Phi_2'\right)&=D\left(\Phi_1\middle\|\Phi_2\right)+D\left(\Phi_1'\middle\|\Phi_2'\right),\\
        \widetilde{D}_\alpha\left(\Phi_1\otimes\Phi_1'\middle\|\Phi_2\otimes\Phi_2'\right)&=\widetilde{D}_\alpha\left(\Phi_1\middle\|\Phi_2\right)+\widetilde{D}_\alpha\left(\Phi_1'\middle\|\Phi_2'\right),
    \end{align}
    where $D$ is defined as~\eqref{eq:D_channel},
    and $\widetilde{D}_\alpha$ is defined as~\eqref{eq:D_alpha_channel}.
\end{lemma}

\begin{proof}
    Due to~\eqref{eq:D_alpha_limit}, we will show the statement on $\widetilde{D}_\alpha$ for any $\alpha>1$, which also indicates that for $D$.
    For $\widetilde{D}_\alpha$, we have
    \begin{widetext}
    \begin{align}
        \label{eq:additivity_1}
        \widetilde{D}_\alpha\left(\Phi_1\otimes\Phi_1'\middle\|\Phi_2\otimes\Phi_2'\right) &=\max_{(x,x')\in\mathcal{X}\times\mathcal{X}'}\widetilde{D}_\alpha\left(\qty(\Phi_1\otimes\Phi_1')\qty(x,x')\middle\|\qty(\Phi_2\otimes\Phi_2')\qty(x,x')\right)\\
        \label{eq:additivity_2}
        &=\max_{(x,x')\in\mathcal{X}\times\mathcal{X}'}\widetilde{D}_\alpha\left(\Phi_1\qty(x)\otimes\Phi_1'\qty(x')\middle\|\Phi_2\qty(x)\otimes\Phi_2'\qty(x)\right)\\
        \label{eq:additivity_3}
        &=\max_{(x,x')\in\mathcal{X}\times\mathcal{X}'}\left\{\widetilde{D}_\alpha\left(\Phi_1\qty(x)\middle\|\Phi_2\qty(x)\right)+\widetilde{D}_\alpha\left(\Phi_1'\qty(x')\middle\|\Phi_2'\qty(x')\right)\right\},
    \end{align}
    \end{widetext}
    where~\eqref{eq:additivity_1} follows~\eqref{eq:D_alpha_channel},~\eqref{eq:additivity_2} is the definition~\eqref{eq:tensor_product_CQ_channel} of the tensor product of CQ channels, and~\eqref{eq:additivity_3} uses the additivity~\eqref{eq:additivity_state} for states.
    Then, as the channel inputs are classical, we have
    \begin{widetext}
    \begin{align}
        \max_{(x,x')\in\mathcal{X}\times\mathcal{X}'}\left\{\widetilde{D}_\alpha\left(\Phi_1\qty(x)\middle\|\Phi_2\qty(x)\right)+\widetilde{D}_\alpha\left(\Phi_1'\qty(x')\middle\|\Phi_2'\qty(x')\right)\right\}
        &=\max_{x\in\mathcal{X}}\widetilde{D}_\alpha\left(\Phi_1\qty(x)\middle\|\Phi_2\qty(x)\right)+\max_{x'\in\mathcal{X}'}\widetilde{D}_\alpha\left(\Phi_1'\qty(x')\middle\|\Phi_2'\qty(x')\right)\\
        \label{eq:additivity_5}
        &=\widetilde{D}_\alpha\left(\Phi_1\middle\|\Phi_2\right)+\widetilde{D}_\alpha\left(\Phi_1'\middle\|\Phi_2'\right),
    \end{align}
    \end{widetext}
    leading to the conclusion.
\end{proof}

Using this additivity, we have the following lemma on the subadditivity of the channel divergence between CQ channels and the set $\mathcal{F}$ of free CQ channels.

\begin{lemma}[\label{lem:subadditivity_CQ_channels}Subadditivity of channel divergences between CQ channels and the sets of free CQ channels]
    For any $n,n'\in\qty{1,2,\ldots}$, any family $\mathcal{F}$ of sets of free CQ channels satisfying Axioms~\ref{p:full_rank},~\ref{p:compact}, and~\ref{p:tensor_product}, and any CQ channel $\Phi\in\mathcal{C}\qty(\mathcal{X}\to\mathcal{H})$, it holds that
\begin{align}
    D\left(\Phi^{\otimes n+n'}\middle\|\mathcal{F}\right)&\leq D\left(\Phi^{\otimes n}\middle\|\mathcal{F}\right)+D\left(\Phi^{\otimes n'}\middle\|\mathcal{F}\right),\\
    \widetilde{D}_\alpha\left(\Phi^{\otimes n+n'}\middle\|\mathcal{F}\right)&\leq\widetilde{D}_\alpha\left(\Phi^{\otimes n}\middle\|\mathcal{F}\right)+\widetilde{D}_\alpha\left(\Phi^{\otimes n'}\middle\|\mathcal{F}\right),
\end{align}
    where $D$ is defined as~\eqref{eq:D_F},
    and $\widetilde{D}_\alpha$ is defined as~\eqref{eq:D_alpha_F}.
\end{lemma}

\begin{proof}
    As in the proof of Lemma~\ref{lem:additivity_CQ_channels}, due to~\eqref{eq:D_alpha_limit}, it suffices to prove the statement for $\widetilde{D}_\alpha$.

    Axiom~\ref{p:full_rank} ensures that $\widetilde{D}_\alpha$ does not diverge, and Axiom~\ref{p:compact} guarantees that the minimum in its definition~\eqref{eq:D_alpha_F} exists.
    Let $\Phi_\mathrm{free}^{(n)}$ and $\Phi_\mathrm{free}^{(n')}$ be free CQ channels achieving the minima in
\begin{align}
    \widetilde{D}_\alpha\left(\Phi^{\otimes n}\middle\|\Phi_\mathrm{free}^{(n)}\right)&=\min_{\Phi_\mathrm{free}\in\mathcal{F}}\widetilde{D}_\alpha\left(\Phi^{\otimes n}\middle\|\Phi_\mathrm{free}\right),\\
    \widetilde{D}_\alpha\left(\Phi^{\otimes n}\middle\|\Phi_\mathrm{free}^{(n')}\right)&=\min_{\Phi_\mathrm{free}'\in\mathcal{F}}\widetilde{D}_\alpha\left(\Phi^{\otimes n'}\middle\|\Phi_\mathrm{free}'\right).
\end{align}
Axiom~\ref{p:tensor_product} ensures that
\begin{align}
    \Phi_\mathrm{free}^{(n)}\otimes\Phi_\mathrm{free}^{(n')}\in\mathcal{F}.
\end{align}

Then, using the additivity in Lemma~\ref{lem:additivity_CQ_channels}, we have
\begin{align}
    &\widetilde{D}_\alpha\left(\Phi^{\otimes n+n'}\middle\|\mathcal{F}\right)\notag\\
    &=\min_{\Phi_\mathrm{free}\in\mathcal{F}}\widetilde{D}_\alpha\left(\Phi^{\otimes n+n'}\middle\|\Phi_\mathrm{free}\right)\\
    &\leq\widetilde{D}_\alpha\left(\Phi^{\otimes n+n'}\middle\|\Phi_\mathrm{free}^{(n)}\otimes\Phi_\mathrm{free}^{(n')}\right)\\
    &=\widetilde{D}_\alpha\left(\Phi^{\otimes n}\middle\|\Phi_\mathrm{free}^{(n)}\right)+\widetilde{D}_\alpha\left(\Phi^{\otimes n'}\middle\|\Phi_\mathrm{free}^{(n')}\right)\\
    &=\min_{\Phi_\mathrm{free}\in\mathcal{F}}\widetilde{D}_\alpha\left(\Phi^{\otimes n}\middle\|\Phi_\mathrm{free}\right)+\min_{\Phi_\mathrm{free}'\in\mathcal{F}}\widetilde{D}_\alpha\left(\Phi^{\otimes n'}\middle\|\Phi_\mathrm{free}'\right)\\
    &=\widetilde{D}_\alpha\left(\Phi^{\otimes n}\middle\|\mathcal{F}\right)+\widetilde{D}_\alpha\left(\Phi^{\otimes n'}\middle\|\mathcal{F}\right).
\end{align}
\end{proof}

In the following proposition, as a corollary of the subadditivity, we see that the limit in the definition~\eqref{eq:regularized_relative_entropy_of_resource} of $R_\mathrm{R}^\infty$ exists under our axioms.
The proof is based on Fekete's subadditive lemma~\cite{fekete1923verteilung} (see also Ref.~\cite[Lemma~A.1]{hayashi2016quantum}): for any subadditive sequence $\qty{a_n}_{n=1,2\ldots}$, i.e.,
\begin{align}
\label{eq:Fekete_assumption}
    a_{n+n'}\leq a_n+a_{n'}~\text{for all $n,n'$},
\end{align}
the limit
\begin{align}
    \lim_{n\to\infty}\frac{a_n}{n}
\end{align}
exists.

\begin{proposition}[\label{prp:existence_limit}Existence of the regularized relative entropy of resource for CQ channels]
For any family $\mathcal{F}$ of sets of free CQ channels satisfying
Axioms~\ref{p:full_rank},~\ref{p:compact}, and~\ref{p:tensor_product}, and any CQ channel $\Phi\in\mathcal{C}\qty(\mathcal{X}\to\mathcal{H})$, the limit
\begin{align}
    \lim_{n\to\infty}\frac{1}{n}D\left(\Phi^{\otimes n}\middle\|\mathcal{F}\right)
\end{align}
exists, where $D$ is defined as~\eqref{eq:D_F}.
\end{proposition}

\begin{proof}
    Due to the subadditivity shown in Lemma~\ref{lem:subadditivity_CQ_channels}, 
    by setting $a_n\coloneqq D\left(\Phi^{\otimes n}\middle\|\mathcal{F}\right)$ in~\eqref{eq:Fekete_assumption}, Fekete's subadditive lemma shows that the limit exists.
\end{proof}

\section{Analysis of generalized quantum Stein's lemma for CQ channels}
\label{sec:stein}

In this section, we formulate and prove the generalized quantum Stein's lemma for CQ channels.
In Sec.~\ref{sec:formulation_stein}, we formulate a quantum hypothesis testing task for CQ dynamical resources and present the generalized quantum Stein's lemma for CQ channels.
In Sec.~\ref{sec:minimax}, we analyze properties of the error exponent appearing in this lemma, which will be useful for our analysis.
In Sec.~\ref{sec:proof_generalized_quantum_steins_lemma}, we present proof of the generalized quantum Stein's lemma for CQ channels.

\subsection{Formulation of generalized quantum Stein's lemma for CQ channels and its property}
\label{sec:formulation_stein}

In this section, we formulate quantum hypothesis testing for discriminating CQ dynamical resources and correspondingly present the generalized quantum Stein's lemma for CQ channels.

We define a task of quantum hypothesis testing for CQ dynamical resources, which aims to distinguish a given CQ channel from any free CQ channel in a set $\mathcal{F}$.
In this task, we are initially given an integer $n\in\{1,2,\ldots\}$ and classical descriptions of a CQ channel $\Phi\in\mathcal{C}\qty(\mathcal{X}\to\mathcal{H})$ and the family $\mathcal{F}$ of sets of free CQ channels satisfying Axioms~\ref{p:full_rank},~\ref{p:compact},~\ref{p:tensor_product}, and~\ref{p:convex}, along with an unknown CQ channel in $\mathcal{C}\qty(\mathcal{X}^n\to\mathcal{H}^{\otimes n})$.
The goal of the task is to distinguish the following two cases:
\begin{itemize}
    \item Null hypothesis: The given unknown CQ channel is $n$-fold copies $\Phi^{\otimes n}$ of $\Phi$.
    \item Alternative hypothesis: The given unknown CQ channel is some free CQ channel $\Phi_\mathrm{free}\in\mathcal{F}\qty(\mathcal{X}^n\to\mathcal{H}^{\otimes n})$, where $\Phi_\mathrm{free}$ may have a composite form over the $n$-fold input and output spaces.
\end{itemize}
To achieve this goal, we choose a probability distribution $p\qty(x^{(n)})$ over classical inputs $x^{(n)}\in\mathcal{X}^n$ to the unknown CQ channel, so that we can perform a two-outcome measurement of the corresponding output quantum state on $\mathcal{H}^{\otimes n}$ by a POVM $\{T_{x^{(n)}},\mathds{1}-T_{x^{(n)}}\}$, where $0\leq T_{x^{(n)}}\leq\mathds{1}$.
The POVM $\{T_{x^{(n)}},\mathds{1}-T_{x^{(n)}}\}$ for every input $x^{(n)}\in\mathcal{X}^{(n)}$ is specified by choosing a family $\qty{T_{x^{(n)}}}_{x^{(n)}\in\mathcal{X}^{(n)}}$ of POVM elements.

When we input $x^{(n)}$ sampled from $p$, if the measurement outcome is $T_{x^{(n)}}$, we conclude that the unknown CQ channel state was $\Phi^{\otimes n}$, and if $\mathds{1}-T_{x^{(n)}}$, then was some free CQ channel $\Phi_\mathrm{free}\in\mathcal{F}\qty(\mathcal{X}^n\to\mathcal{H}^{\otimes n})$.
For this hypothesis testing, we define the following two types of errors.
\begin{itemize}
    \item Type I error: The mistaken conclusion that the given CQ channel was some free state $\Phi_\mathrm{free}\in\mathcal{F}\qty(\mathcal{X}^n\to\mathcal{H}^{\otimes n})$ when it was $\Phi^{\otimes n}$, which happens with probability
    \begin{align}
    \label{eq:type_I_error}
        \sum_{x^{(n)}\in\mathcal{X}^n}p\qty(x^{(n)})\Tr\qty[\qty(\mathds{1}-T_{x^{(n)}})\Phi^{\otimes n}\qty(x^{(n)})].
    \end{align}
    \item Type II error: The mistaken conclusion that the given CQ channel was $\Phi^{\otimes n}$ when it was some free state $\Phi_\mathrm{free}\in\mathcal{F}\qty(\mathcal{X}^n\to\mathcal{H}^{\otimes n})$, which happens,  in the worst case, with probability
    \begin{align}
    \label{eq:type_II_error}
        \max_{\Phi_\mathrm{free}\in\mathcal{F}\qty(\mathcal{X}^n\to\mathcal{H}^{\otimes n})}\sum_{x^{(n)}\in\mathcal{X}^n}p\qty(x^{(n)})\Tr\qty[T_{x^{(n)}}\Phi_\mathrm{free}\qty(x^{(n)})].
    \end{align}
\end{itemize}

In the setting of the generalized quantum Stein's lemma, we constrain that the type I error should be below a fixed parameter $\epsilon$, and the task aims to minimize the type II error under this constraint.
As $n$ goes to infinity, the type II error decreases exponentially in $n$, and its exponent characterizes how fast the type II error may decay.
The generalized quantum Stein's lemma characterizes the optimal exponent of the type II error.
To analyze these errors, for a parameter $\epsilon\geq 0$, a CQ channel $\Phi\in\mathcal{C}\qty(\mathcal{X}\to\mathcal{H})$, and a probability distribution $p$ over $\mathcal{X}$, we let $\mathcal{T}_{\epsilon,\Phi,p}$ denote the set of the POVM elements achieving the type I error~\eqref{eq:type_I_error} below $\epsilon$ when the input to the CQ channel $\Phi$ is given according to $p$, i.e.,
\begin{align}
\label{eq:T_set}
    \mathcal{T}_{\epsilon,\Phi,p}&\coloneqq\left\{\qty{T_x}_{x\in\mathcal{X}}:\right.\notag\\
    &\quad\left.0\leq T_x\leq\mathds{1},\right.\notag\\
    &\quad\left.\sum_{x\in\mathcal{X}}p\qty(x)\Tr\qty[\qty(\mathds{1}-T_x)\Phi\qty(x)]\leq\epsilon\right\}.
\end{align}
With this set, we represent the optimal type II error~\eqref{eq:type_II_error} using a function
\begin{align}
    \label{eq:beta}
    &\beta_\epsilon\left(\Phi\middle\|\mathcal{F}\right)\notag\\
    &\coloneqq\min_{p}\min_{\qty{T_{x}}_x\in\mathcal{T}_{\epsilon,\Phi,p}}\max_{\Phi_\mathrm{free}\in\mathcal{F}}\sum_{x\in\mathcal{X}}p(x)\Tr\qty[T_{x}\Phi_\mathrm{free}\qty(x)],
\end{align}
where $\min_p$ denotes minimization over all probability distributions $p$ over the set of channel inputs.
Then, the generalized quantum Stein's lemma is stated as follows.

\begin{theorem}[\label{thm:main}The generalized quantum Stein's lemma for CQ channels]
For any error parameter $\epsilon\in(0,1)$, any family $\mathcal{F}$ of sets of free CQ channels satisfying Axioms~\ref{p:full_rank},~\ref{p:compact},~\ref{p:tensor_product}, and~\ref{p:convex}, and any CQ channel $\Phi\in\mathcal{C}\qty(\mathcal{X}\to\mathcal{H})$, it holds that
\begin{align}
    \label{eq:Stein}
    &\lim_{n\to\infty}-\frac{1}{n}\log\qty[\beta_\epsilon\left(\Phi^{\otimes n}\middle\|\mathcal{F}\right)]\notag\\
    &=\lim_{n\to\infty}\frac{1}{n}D\left(\Phi^{\otimes n}\middle\|\mathcal{F}\right).
\end{align}
where $\beta_\epsilon$ is defined as~\eqref{eq:beta}, and $D$ is defined as~\eqref{eq:D_F}.
\end{theorem}

The main difficulty in analyzing Theorem~\ref{thm:main} lies in the fact that the channel inputs can be arbitrary; even when considering $n$-fold copies $\Phi^{\otimes n}$ of a CQ channel in the first argument of~\eqref{eq:Stein}, the corresponding outputs are not IID copies of a single state.
By contrast, in the state version of the generalized quantum Stein's lemma~\cite{Brand_o_2008,brandao2010reversible,Brandao2010,hayashi2025generalizedquantumsteinslemma,10898013}, the analysis essentially exploits the fact that the first argument is an IID state, and therefore, the same proof techniques cannot be directly applied in our CQ-channel setting.
Nevertheless, in the remainder of this section, we extend the proof techniques developed in the state case~\cite{hayashi2025generalizedquantumsteinslemma} to CQ channels, thereby overcoming this challenge.
In Sec.~\ref{sec:minimax}, we provide a minimax characterization of $\beta_\epsilon\left(\Phi^{\otimes n}\middle\|\mathcal{F}\right)$.
Then, in Sec.~\ref{sec:proof_generalized_quantum_steins_lemma}, we develop proof techniques for Theorem~\ref{thm:main} based on this characterization.
Combining these ingredients, we summarize the proof of the theorem as follows.

\begin{proof}[Proof of Theorem~\ref{thm:main}]
Proposition~\ref{prp:existence_limit} ensures that the limit
\begin{align}
    \lim_{n\to\infty}\frac{1}{n}D\left(\Phi^{\otimes n}\middle\|\mathcal{F}\right)
\end{align}
exists.  
In Sec.~\ref{sec:strong_converse_stein}, we prove Proposition~\ref{prp:strong_converse}, yielding
\begin{align}
    &\limsup_{n\to\infty}-\frac{1}{n}\log\qty[\beta_\epsilon\left(\Phi^{\otimes n}\middle\|\mathcal{F}\right)]\notag\\
    &\leq\lim_{n\to\infty}\frac{1}{n}D\left(\Phi^{\otimes n}\middle\|\mathcal{F}\right).
\end{align}
In Sec.~\ref{sec:direct_stein}, we prove Proposition~\ref{prp:direct}, which establishes
\begin{align}
    &\liminf_{n\to\infty}-\frac{1}{n}\log\qty[\beta_\epsilon\left(\Phi^{\otimes n}\middle\|\mathcal{F}\right)]\notag\\
    &\geq\lim_{n\to\infty}\frac{1}{n}D\left(\Phi^{\otimes n}\middle\|\mathcal{F}\right).
\end{align}
Together, these results imply that the limit
\begin{align}
    \lim_{n\to\infty}-\frac{1}{n}\log\qty[\beta_\epsilon\left(\Phi^{\otimes n}\middle\|\mathcal{F}\right)]
\end{align}
exists and coincides with the right-hand side of~\eqref{eq:Stein}.
\end{proof}

\subsection{Properties of type II errors}
\label{sec:minimax}

In this section, we present several properties of the type II error $\beta_\epsilon\left(\Phi\middle|\mathcal{F}\right)$ defined in~\eqref{eq:beta}, which will be useful for analyzing the generalized quantum Stein's lemma for CQ channels.

We first provide a minimax characterization of the type II error $\beta_\epsilon\left(\Phi\middle\|\mathcal{F}\right)$ in~\eqref{eq:beta} for a CQ channel $\Phi$ and a set $\mathcal{F}$.
To this end, we introduce an auxiliary quantity between two CQ channels $\Phi_1$ and $\Phi_2$:
\begin{align}
\label{eq:beta_CQ_channels}
\beta_\epsilon\left(\Phi_1\middle|\Phi_2\right)\coloneqq\min_{p}\min_{\qty{T_{x}}x\in\mathcal{T}_{\epsilon,\Phi_1,p}}\sum_{x\in\mathcal{X}}p(x)\Tr\qty[T_{x}\Phi_2(x)],
\end{align}
where the minimization is over input distributions $p$ and POVM elements ${T_x}_x$ satisfying the error constraint specified by $\mathcal{T}_{\epsilon,\Phi_1,p}$ in~\eqref{eq:T_set} with $\Phi_1$.
In the proof of the state version of the generalized quantum Stein's lemma~\cite{hayashi2025generalizedquantumsteinslemma}, the type II error involves only two optimizations, i.e., minimization over POVM elements and maximization over the set $\mathcal{F}$.
By contrast, in our CQ-channel setting, $\beta_\epsilon$ in~\eqref{eq:beta} requires an additional minimization over input distributions.
This results in a nested optimization problem involving several minimization and maximization across distinct sets, which complicates obtaining a simple characterization.
Nevertheless, when the set $\mathcal{F}$ is compact and convex, we will show that this difficulty can be overcome by representing the inputs and POVM elements as a single set, so the type II error admits the following minimax characterization.

\begin{proposition}[\label{prp:minimax}Minimax characterization of type II errors for CQ channels]
    For any $\epsilon\geq 0$,
    any family $\mathcal{F}$ of sets of free CQ channels satisfying Axioms~\ref{p:compact} and~\ref{p:convex},
    and any CQ channel $\Phi\in\mathcal{C}\qty(\mathcal{X}\to\mathcal{H})$,
    it holds that
    \begin{align}
        \beta_\epsilon\left(\Phi\middle\|\mathcal{F}\right)=
        \max_{\Phi_\mathrm{free}\in\mathcal{F}}\beta_\epsilon\left(\Phi\middle\|\Phi_\mathrm{free}\right),
    \end{align}
    where $\beta_\epsilon$ on the left-hand side is defined as~\eqref{eq:beta}, and $\beta_\epsilon$ on the right-hand side is defined as~\eqref{eq:beta_CQ_channels}.
\end{proposition}

\begin{proof}
    We will show that
    \begin{align}
        &\beta_\epsilon\left(\Phi\middle\|\mathcal{F}\right)\notag\\
        &=
        \min_{p}\min_{\qty{T_{x}}_x\in\mathcal{T}_{\epsilon,\Phi,p}}\max_{\Phi_\mathrm{free}\in\mathcal{F}}\Tr\qty[T_{x}\Phi_\mathrm{free}\qty(x)]
    \end{align}
    and
    \begin{align}
        &\max_{\Phi_\mathrm{free}\in\mathcal{F}}\beta_\epsilon\left(\Phi\middle\|\Phi_\mathrm{free}\right)\notag\\
        &=
        \max_{\Phi_\mathrm{free}\in\mathcal{F}}\min_{p}\min_{\qty{T_{x}}_x\in\mathcal{T}_{\epsilon,\Phi,p}}\Tr\qty[T_{x}\Phi_\mathrm{free}\qty(x)]
    \end{align}
    coincide.
    The proof is based on the minimax theorem~\cite{v1928theorie,sion1958general,10.2996/kmj/1138038812}, which shows that a pair of $\min$ and $\max$ may commute if these optimizations are over compact convex sets and the objective function to be optimized is bilinear.
    However, in our case, we have three optimizations.
    Hence, it is essential for our proof to take appropriate variables to ensure that we have a pair of compact convex sets.

    To this end, instead of a CQ channel $\Phi\in\mathcal{C}\qty(\mathcal{X}\to\mathcal{H})$, we consider its Choi operator $J\qty(\Phi)\in\mathcal{L}\qty(\mathcal{H}_\mathcal{X}\otimes\mathcal{H})$ with $\mathcal{H}_\mathcal{X}\coloneqq\spn\qty{\ket{x}:x\in\mathcal{X}}$
    \begin{align}
        J\qty(\Phi)\coloneqq\sum_{x\in\mathcal{X}}\ket{x}\bra{x}\otimes \Phi(x).
    \end{align}
    Instead of the input $x\in\mathcal{X}$ and the POVM measurement $\qty{T_{x},\mathds{1}-T_{x}}$, we introduce an operator in the form of
    \begin{align}
        K=\sum_{x\in\mathcal{X}}p(x)\ket{x}\bra{x}\otimes T_x\in\mathcal{L}\qty(\mathcal{H}_\mathcal{X}\otimes\mathcal{H}),
    \end{align}
    so that we have
    \begin{align}
        \Tr\qty[KJ\qty(\Phi)]=\sum_{x\in\mathcal{X}}p(x)\Tr\qty[T_x\Phi(x)].
    \end{align}
    The set of these operators satisfying the constraints in~\eqref{eq:T_set} on the type I errors is given by
    \begin{align}
        \mathcal{K}_{\epsilon,J\qty(\Phi)}\coloneqq
        &\left\{K=\sum_{x\in\mathcal{X}}p(x)\ket{x}\bra{x}\otimes T_x:\right.\notag\\
        &\Tr\qty[KJ\qty(\Phi)]\geq1-\epsilon,\notag\\
        &p(x)\geq 0,~\sum_x p(x)=1,\notag\\
        &\left.0\leq T_x\leq\mathds{1}\right\},
    \end{align}
    which is a compact convex set.

    Then, we have
    \begin{align}
        \beta_\epsilon\left(\Phi\middle\|\mathcal{F}\right)&=\min_{K\in\mathcal{K}_{\epsilon,J\qty(\Phi)}}\max_{\Phi_\mathrm{free}\in\mathcal{F}}\Tr\qty[KJ\qty(\Phi_\mathrm{free})],\\
        \max_{\Phi_\mathrm{free}\in\mathcal{F}}\beta_\epsilon\left(\Phi\middle\|\Phi_\mathrm{free}\right)&=\max_{\Phi_\mathrm{free}\in\mathcal{F}}\min_{K\in\mathcal{K}_{\epsilon,J\qty(\Phi)}}\Tr\qty[KJ\qty(\Phi_\mathrm{free})].
    \end{align}
    Due to the convexity and the compactness of $\mathcal{K}_{\epsilon,J\qty(\Phi)}$ and those of $\mathcal{F}$ from Axioms~\ref{p:compact} and~\ref{p:convex}, by the bilinearity of $\Tr\qty[KJ\qty(\Phi_\mathrm{free})]$ in terms of $K$ and $\Phi_\mathrm{free}$, the minimax theorem~\cite{v1928theorie,sion1958general,10.2996/kmj/1138038812} shows that
    \begin{align}
        &\min_{K\in\mathcal{K}_{\epsilon,J\qty(\Phi)}}\max_{\Phi_\mathrm{free}\in\mathcal{F}}\Tr\qty[KJ\qty(\Phi_\mathrm{free})]\notag\\
        &=\max_{\Phi_\mathrm{free}\in\mathcal{F}}\min_{K\in\mathcal{K}_{\epsilon,J\qty(\Phi)}}\Tr\qty[KJ\qty(\Phi_\mathrm{free})],
    \end{align}
    indicating the conclusion.
\end{proof}

Additionally, for the quantity $\beta_\epsilon\left(\Phi_1\middle\|\Phi_2\right)$ defined in~\eqref{eq:beta_CQ_channels}, we will show a monotonicity under the action of superchannels, as presented below.
In Ref.~\cite{hayashi2025generalizedquantumsteinslemma}, the analogous monotonicity of type II errors in the state version of the generalized quantum Stein's lemma,  under the action of channels, was used.
By contrast, in our CQ-channel setting, the proof requires a more refined analysis, since we must account not only for the action of channels on the outputs but also for the conversion of input probability distributions by superchannels, and for the auxiliary input and output systems that superchannels may involve.

\begin{lemma}[\label{lem:monotonicity_beta_CQ_channel}Monotonicity of type II errors for CQ channels under superchannels]
    For any parameter $\epsilon\geq 0$, any CQ channels $\Phi_1,\Phi_2\in\mathcal{C}\qty(\mathcal{X}_\mathrm{in}\times\mathcal{X}_\mathrm{aux}\to\mathcal{H}_\mathrm{in}\otimes\mathcal{H}_\mathrm{aux})$, and any superchannel $\Theta\in\mathcal{C}\qty(\qty(\mathcal{X}_\mathrm{in}\to\mathcal{H}_\mathrm{in})\to\qty(\mathcal{X}_\mathrm{out}\to\mathcal{H}_\mathrm{out}))$, it holds that
    \begin{align}
        \beta_\epsilon\left(\qty(\Theta\otimes\id)[\Phi_1]\middle\|\qty(\Theta\otimes\id)[\Phi_2]\right)\geq\beta_\epsilon\left(\Phi_1\middle\|\Phi_2\right),
    \end{align}
    where $\beta_\epsilon$ is defined as~\eqref{eq:beta_CQ_channels}, and $\id$ is the identity supermap as in the definition~\eqref{eq:C} of superchannels.
\end{lemma}

\begin{proof}
    As in~\eqref{eq:theta}, we represent the superchannel $\Theta$ as
    \begin{align}
        &\qty(\qty(\Theta\otimes\id)[\Phi])\qty(x_\mathrm{out},x_\mathrm{aux})\notag\\
        &=\sum_{x_\mathrm{in}\in\mathcal{X}_\mathrm{in}}p_{\Theta}\qty(x_\mathrm{in}|x_\mathrm{out})\qty(\qty(\mathcal{N}_{\Theta,x_\mathrm{in},x_\mathrm{out}}\otimes\id)\circ\Phi)\qty(x_\mathrm{in},x_\mathrm{aux}).
    \end{align}
    Then, by the definition~\eqref{eq:beta_CQ_channels} of $\beta_\epsilon$, we have
    \begin{widetext}
    \begin{align}
        &\beta_\epsilon\left(\qty(\Theta\otimes\id)[\Phi_1]\middle\|\qty(\Theta\otimes\id)[\Phi_2]\right)\notag\\
        &=\min_{p}\min_{\qty{T_{x_\mathrm{out},x_\mathrm{aux}}}_{x_\mathrm{out},x_\mathrm{aux}}\in\mathcal{T}_{\epsilon,\qty(\Theta\otimes\id)[\Phi_1],p}}\notag\\
        &\quad\sum_{x_\mathrm{out},x_\mathrm{aux}}p\qty(x_\mathrm{out},x_\mathrm{aux})\Tr\qty[T_{x_\mathrm{out},x_\mathrm{aux}}\qty(\sum_{x_\mathrm{in}\in\mathcal{X}_\mathrm{in}}p_{\Theta}\qty(x_\mathrm{in}|x_\mathrm{out})\qty(\qty(\mathcal{N}_{\Theta,x_\mathrm{in},x_\mathrm{out}}\otimes\id)\circ\Phi_2)\qty(x_\mathrm{in},x_\mathrm{aux}))]\\
        &=\min_{p'}\min_{\qty{T_{x_\mathrm{in},x_\mathrm{aux}}^{\prime}}_{x_\mathrm{in},x_\mathrm{aux}}}\sum_{x_\mathrm{in},x_\mathrm{aux}}p'\qty(x_\mathrm{in},x_\mathrm{aux})\Tr\qty[T_{x_\mathrm{in},x_\mathrm{aux}}^{\prime}\Phi_2\qty(x_\mathrm{in},x_\mathrm{aux})],
    \end{align}
    where $p'$ and $\qty{T_{x_\mathrm{in},x_\mathrm{aux}}^{\prime}}_{x_\mathrm{in},x_\mathrm{aux}}$ are optimized over those having particular forms of
    \begin{align}
        p'\qty(x_\mathrm{in},x_\mathrm{aux})&=\sum_{x_\mathrm{out}\in\mathcal{X}_\mathrm{out}}p_{\Theta}\qty(x_\mathrm{in}|x_\mathrm{out})p\qty(x_\mathrm{out},x_\mathrm{aux}),\\
        T_{x_\mathrm{in},x_\mathrm{aux}}^{\prime}&=\frac{\sum_{x_\mathrm{out}\in\mathcal{X}_\mathrm{out}}p_{\Theta}\qty(x_\mathrm{in}|x_\mathrm{out})p\qty(x_\mathrm{out},x_\mathrm{aux})\qty(\mathcal{N}_{\Theta,x_\mathrm{in},x_\mathrm{out}}\otimes\id)^\dag\qty(T_{x_\mathrm{out},x_\mathrm{aux}})}{p'\qty(x_\mathrm{in},x_\mathrm{aux})}\notag\\
        &\quad\text{for $\qty{T_{x_\mathrm{out},x_\mathrm{aux}}}_{x_\mathrm{out},x_\mathrm{aux}}\in\mathcal{T}_{\epsilon,\qty(\Theta\otimes\id)[\Phi_1],p}$}.
    \end{align}
    Since $p'$ is optimized over a subset of the set of probability distributions over $\mathcal{X}_\mathrm{in}\times\mathcal{X}_\mathrm{aux}$, and $\qty{T_{x_\mathrm{in},x_\mathrm{aux}}^{\prime}}_{x_\mathrm{in},x_\mathrm{aux}}$ is optimized over a subset of $\mathcal{T}_{\epsilon,\Phi_1,p'}$, it follows that
    \begin{align}
    \label{eq:beta_monotonicity_1}
        \beta_\epsilon\left(\qty(\Theta\otimes\id)[\Phi_1]\middle\|\qty(\Theta\otimes\id)[\Phi_2]\right) &\geq
        \min_{p}\min_{\qty{T_{x_\mathrm{in},x_\mathrm{aux}}}_{x_\mathrm{in},x_\mathrm{aux}}\in\mathcal{T}_{\epsilon,\Phi_1,p}}\sum_{x_\mathrm{in},x_\mathrm{aux}}p\qty(x_\mathrm{in},x_\mathrm{aux})\Tr\qty[T_{x_\mathrm{in},x_\mathrm{aux}}\Phi_2\qty(x_\mathrm{in},x_\mathrm{aux})]\\
        &=\beta_\epsilon\left(\Phi_1\middle\|\Phi_2\right),
    \end{align}
    where the minimization of $p$ on the right-hand side of~\eqref{eq:beta_monotonicity_1} is over the set of all probability distributions over $\mathcal{X}_\mathrm{in}\times\mathcal{X}_\mathrm{aux}$.
    \end{widetext}
\end{proof}

\subsection{Main parts of proof of the generalized quantum Stein's lemma for CQ channels}
\label{sec:proof_generalized_quantum_steins_lemma}

In this section, we present the techniques to prove the generalized quantum Stein's lemma for CQ channels in Theorem~\ref{thm:main}.
The proof is composed of two parts: one is the strong converse part to establish the optimality, and the other is the direct part to demonstrate the achievability.
In Sec.~\ref{sec:strong_converse_stein}, we analyze the strong converse part.
In Sec.~\ref{sec:direct_stein}, we develop techniques to prove the direct part.
In each part, we construct a CQ-channel toolkit that extends the techniques used in proving the state version of the generalized quantum Stein's lemma in Ref.~\cite{hayashi2025generalizedquantumsteinslemma}, overcoming the inherent challenges posed by the presence of multiple channel inputs.

\subsubsection{Strong converse part}
\label{sec:strong_converse_stein}

We now prove the strong converse part of the generalized quantum Stein's lemma for CQ channels in Theorem~\ref{thm:main}.
In the state version of the generalized quantum Stein's lemma, Ref.~\cite{hayashi2025generalizedquantumsteinslemma} provided a simple proof of the strong converse using upper bounds on type II errors expressed in terms of the sandwiched Rényi relative entropy of states.
In this section, we generalize these upper bounds from the state setting to the CQ-channel setting, which yields a streamlined proof for the CQ-channel version of the generalized quantum Stein's lemma.
All the resulting bounds naturally reduce to the state case when the input dimension for CQ channels is one.
Our derivation relies on the properties of CQ-channel divergences presented in Sec.~\ref{sec:properties_channel_divergence}.
Whereas further extending such arguments to QQ channels is generally highly nontrivial, we demonstrate that the extension is feasible for CQ channels, precisely because their inputs are classical.

To this end, we prepare the lemma shown below.
For any parameter $\epsilon\in[0,1)$ and any density operator $\rho\in\mathcal{D}\qty(\mathcal{H})$, we write
\begin{align}
    \label{eq:T_set_states}
    \mathcal{T}_{\epsilon,\rho}\coloneqq\qty{T:0\leq T\leq\mathds{1},\Tr\qty[\qty(\mathds{1}-T)\rho]\leq\epsilon}.
\end{align}
We write the type II error between two density operators $\rho_1$ and $\rho_2$ as
\begin{align}
    \label{eq:beta_states}
    \beta_\epsilon\left(\rho_1\middle\|\rho_2\right)\coloneqq\min_{T\in\mathcal{T}_{\epsilon,\rho_1}}\Tr\qty[T\rho_2].
\end{align}
Then, for any parameter $\alpha>1$, the type II error is bounded by~\cite{cooney2016strong,887855}
\begin{align}
    \label{eq:beta_states_bound}
    -\log\qty[\beta_\epsilon\left(\rho_1\middle\|\rho_2\right)]\leq \widetilde{D}_\alpha\left(\rho_1\middle\|\rho_2\right)+\frac{\alpha}{\alpha-1}\log\qty[\frac{1}{1-\epsilon}],
\end{align}
where $\widetilde{D}_\alpha$ on the right-hand side is defined as~\eqref{eq:D_alpha}.
The following lemma generalizes this bound to CQ channels.

\begin{lemma}[\label{lem:bound_type_II_error_CQ_channels}Bound on type II errors between CQ channels]
    For any parameters $\alpha>1$ and $\epsilon\in[0,1)$,
    and any CQ channels $\Phi_1,\Phi_2\in\mathcal{C}\qty(\mathcal{X}\to\mathcal{H})$,
    it holds that
    \begin{align}
        &-\log\qty[\beta_\epsilon\left(\Phi_1\middle\|\Phi_2\right)]\leq\widetilde{D}_\alpha\left(\Phi_1\middle\|\Phi_2\right)+\frac{\alpha}{\alpha-1}\log\qty[\frac{1}{1-\epsilon}],
    \end{align}
    where $\beta_\epsilon$ is defined as~\eqref{eq:beta_CQ_channels}, and $\widetilde{D}_\alpha$ is defined as~\eqref{eq:D_alpha_channel}.
\end{lemma}

\begin{proof}
    We take $x^{\ast}\in\mathcal{X}$ achieving the maximum in the definition~\eqref{eq:D_alpha_channel} of $\widetilde{D}_\alpha$, i.e.,
    \begin{align}
    \widetilde{D}_\alpha\left(\Phi_1\qty(x^\ast)\middle\|\Phi_2\qty(x^\ast)\right)= \widetilde{D}_\alpha\left(\Phi_1\middle\|\Phi_2\right),
    \end{align}
    where $\widetilde{D}_\alpha$ on the left-hand side is given by~\eqref{eq:D_alpha}.
    By the definitions~\eqref{eq:beta_CQ_channels} and~\eqref{eq:beta_states} of $\beta_\epsilon$, along with~\eqref{eq:beta_states_bound}, we have
    \begin{align}
        &-\log\qty[\beta_\epsilon\left(\Phi_1\middle\|\Phi_2\right)]\notag\\
        &\leq-\log\qty[\beta_\epsilon\left(\Phi_1\qty(x^\ast)\middle\|\Phi_2\qty(x^\ast)\right)]\\
        &\leq\widetilde{D}_\alpha\left(\Phi_1\qty(x^\ast)\middle\|\Phi_2\qty(x^\ast)\right)+\frac{\alpha}{\alpha-1}\log\qty[\frac{1}{1-\epsilon}]\\
        &=\widetilde{D}_\alpha\left(\Phi_1\middle\|\Phi_2\right)+\frac{\alpha}{\alpha-1}\log\qty[\frac{1}{1-\epsilon}].
    \end{align}
\end{proof}

The following lemma extends this bound to the one between a CQ channel and the set $\mathcal{F}$.

\begin{lemma}[\label{lem:bound_type_II_error}Bound on type II errors between a CQ channel and a set of free CQ channel]
    For any parameters $\alpha>1$ and $\epsilon\in[0,1)$, any family $\mathcal{F}$ of sets of free CQ channels satisfying Axioms~\ref{p:full_rank},~\ref{p:compact} and~\ref{p:convex},
    and any CQ channel $\Phi\in\mathcal{C}\qty(\mathcal{X}\to\mathcal{H})$,
    it holds that
    \begin{align}
        &-\log\qty[\beta_\epsilon\left(\Phi\middle\|\mathcal{F}\right)]\leq\widetilde{D}_\alpha\left(\Phi\middle\|\mathcal{F}\right)+\frac{\alpha}{\alpha-1}\log\qty[\frac{1}{1-\epsilon}],
    \end{align}
    where $\beta_\epsilon$ is defined as~\eqref{eq:beta}, and $\widetilde{D}_\alpha$ is defined as~\eqref{eq:D_alpha_F}.
\end{lemma}

\begin{proof}
    Due to Axioms~\ref{p:compact} and~\ref{p:convex},
    Proposition~\ref{prp:minimax} shows that
    \begin{align}
        \beta_\epsilon\left(\Phi\middle\|\mathcal{F}\right)&=\max_{\Phi_\mathrm{free}\in\mathcal{F}}\beta_\epsilon\left(\Phi\middle\|\Phi_\mathrm{free}\right).
    \end{align}
    Then, due to Lemma~\ref{lem:bound_type_II_error_CQ_channels}, we have
    \begin{align}
        &-\log\qty[\max_{\Phi_\mathrm{free}\in\mathcal{F}}\beta_\epsilon\left(\Phi\middle\|\Phi_\mathrm{free}\right)]\notag\\
        &=\min_{\Phi_\mathrm{free}\in\mathcal{F}}-\log\qty[\beta_\epsilon\left(\Phi\middle\|\Phi_\mathrm{free}\right)]\\
        &\leq\min_{\Phi_\mathrm{free}\in\mathcal{F}}\widetilde{D}_\alpha\left(\Phi\middle\|\Phi_\mathrm{free}\right)+\frac{\alpha}{\alpha-1}\log\qty[\frac{1}{1-\epsilon}]\\
        &=\widetilde{D}_\alpha\left(\Phi\middle\|\mathcal{F}\right)+\frac{\alpha}{\alpha-1}\log\qty[\frac{1}{1-\epsilon}],
    \end{align}
    where $\widetilde{D}_\alpha$ is finite due to Axioms~\ref{p:full_rank}.
\end{proof}

Using Lemma~\ref{lem:bound_type_II_error}, we prove the strong converse part of the generalized quantum Stein's lemma for CQ channels as follows.
\begin{proposition}[\label{prp:strong_converse}The strong converse part of the generalized quantum Stein's lemma for CQ channels]
    For any parameter $\epsilon\in[0,1)$, 
    any family $\mathcal{F}$ of sets of free CQ channels satisfying Axioms~\ref{p:full_rank},~\ref{p:compact},~\ref{p:tensor_product}, and~\ref{p:convex},
    and any CQ channel $\Phi\in\mathcal{C}\qty(\mathcal{X}\to\mathcal{H})$,
    it holds that
    \begin{align}
        &\limsup_{n\to\infty}-\frac{1}{n}\log\qty[\beta_\epsilon\left(\Phi^{\otimes n}\middle\|\mathcal{F}\right)]\leq\lim_{n\to\infty}\frac{1}{n}D\left(\Phi^{\otimes n}\middle\|\mathcal{F}\right),
    \end{align}
    where $\beta_\epsilon$ is defined as~\eqref{eq:beta}, and $D$ is defined as~\eqref{eq:D_F}.
\end{proposition}

\begin{proof}
    We fix a positive integer $M$ and, for each $n$, choose $q_n$ and $r_n$ such that
    \begin{align}
    \label{eq:strong_converse_n}
        n=q_n M + r_n,~0\leq r_n<M.
    \end{align}
    Then, using Lemma~\ref{lem:bound_type_II_error} with Axioms~\ref{p:full_rank},~\ref{p:compact} and~\ref{p:convex}, we have
    \begin{align}
        &-\frac{1}{n}\log\qty[\beta_\epsilon\left(\Phi^{\otimes n}\middle\|\mathcal{F}\right)]\notag\\
        &\leq\frac{1}{n}\widetilde{D}_\alpha\left(\Phi^{\otimes q_n M+r_n}\middle\|\mathcal{F}\right)+\frac{\alpha}{\alpha-1}\log\qty[\frac{1}{1-\epsilon}]\\
        &\leq\frac{q_n}{n}\widetilde{D}_\alpha\left(\Phi^{\otimes M}\middle\|\mathcal{F}\right)+\frac{r_n}{n}\widetilde{D}_\alpha\left(\Phi\middle\|\mathcal{F}\right)+\notag\\
        &\quad\frac{\alpha}{n(\alpha-1)}\log\qty[\frac{1}{1-\epsilon}],
    \end{align}
    where the last line follows from the subadditivity in Lemma~\ref{lem:subadditivity_CQ_channels} with Axioms~\ref{p:full_rank},~\ref{p:compact}, and~\ref{p:tensor_product}.
    By taking the limit $n\to\infty$ with~\eqref{eq:strong_converse_n}, we obtain
    \begin{align}
        &\limsup_{n\to\infty}-\frac{1}{n}\log\qty[\beta_\epsilon\left(\Phi^{\otimes n}\middle\|\mathcal{F}\right)]\notag\\
        &\leq\frac{1}{M}\widetilde{D}_\alpha\left(\Phi^{\otimes M}\middle\|\mathcal{F}\right),
    \end{align}
    which holds for any $\alpha>1$.
    Due to~\eqref{eq:D_alpha_limit}, taking the limit $\alpha\to 1$ yields
    \begin{align}
        &\limsup_{n\to\infty}-\frac{1}{n}\log\qty[\beta_\epsilon\left(\Phi^{\otimes n}\middle\|\mathcal{F}\right)]\notag\\
        &\leq\frac{1}{M}D\left(\Phi^{\otimes M}\middle\|\mathcal{F}\right),
    \end{align}
    which holds for any choice of $M$.
    In the limit $M\to\infty$, it holds that
    \begin{align}
        &\limsup_{n\to\infty}-\frac{1}{n}\log\qty[\beta_\epsilon\left(\Phi^{\otimes n}\middle\|\mathcal{F}\right)]\notag\\
        &\leq\lim_{n\to\infty}\frac{1}{n}D\left(\Phi^{\otimes n}\middle\|\mathcal{F}\right),
    \end{align}
    where the limit on the right-hand side exists due to Lemma~\ref{prp:existence_limit}.
\end{proof}

\subsubsection{Direct part}
\label{sec:direct_stein}

We now turn to the proof of the direct part of the generalized quantum Stein's lemma for CQ channels in Theorem~\ref{thm:main}.
In the proof of the state version of the lemma in Ref.~\cite{hayashi2025generalizedquantumsteinslemma}, two techniques played a central role: the pinching technique~\cite{hayashi2002optimal} and the information spectrum method~\cite{4069150}.
These methods were originally developed for the state setting, whereas our current task lies in the more general CQ-channel setting.
Here, we extend these techniques to the CQ-channel setting.
In Sec.~\ref{sec:pinching}, we present an adaptation of the pinching method to CQ channels together with its implications.
In Sec.~\ref{sec:information_spectrum}, we develop the information spectrum method for CQ channels also with its implications.
Finally, in Sec.~\ref{sec:proof_direct}, we complete the proof of the direct part of the generalized quantum Stein's lemma using the implications from these CQ-channel techniques.

\paragraph{The pinching technique for CQ channels}
\label{sec:pinching}

The pinching technique is crucial because it enables one to render non-commutative operators commuting without disturbing the error exponent in quantum hypothesis testing~\cite{hayashi2002optimal}. Here, we present a method to extend this technique to CQ channels. To this end, analogous to the pinching channel introduced in Ref.~\cite{hayashi2025generalizedquantumsteinslemma} for states, we define a pinching superchannel for CQ channels that accommodates the presence of multiple possible inputs.
While a further extension to QQ channels would in general be highly nontrivial, our construction is feasible precisely because the channel inputs are classical, which allows us to generalize the state-based results of Ref.~\cite{hayashi2025generalizedquantumsteinslemma} to the CQ-channel setting.

We begin with summarizing the pinching technique in the state setting.
Let $\rho\in\mathcal{D}\qty(\mathcal{H})$ be a quantum state in the spectral decomposition
\begin{align}
\label{eq:rho_spectral_decomposition}
    \rho=\sum_{j=0}^{J_\rho-1}\lambda_j\Pi_j,
\end{align}
where $J_\rho$ is the number of distinct eigenvalues of $\rho$, $\{\lambda_j\}_j$ is the set of distinct eigenvalues, and $\Pi_j$ is the projection onto the eigenspace associated with each eigenvalue $\lambda_j$.
With $\qty{\Pi_j}_{j=0,\ldots,J_\rho-1}$ in~\eqref{eq:rho_spectral_decomposition}, the pinching map with respect to the state $\rho$ is defined as
\begin{align}
\label{eq:pinching_map}
    \mathcal{P}_{\rho}\qty(\sigma)\coloneqq\sum_{j=0}^{J_\rho-1}\Pi_j\sigma\Pi_j.
\end{align}
For any state $\sigma\in\mathcal{D}\qty(\mathcal{H})$, the pinching inequality~\cite{hayashi2002optimal} yields an operator inequality
\begin{align}
\label{eq:pinching_inequality}
    \sigma\leq J_\rho\mathcal{P}_{\rho}\qty(\sigma),
\end{align}
where $J_\rho$ is defined as~\eqref{eq:rho_spectral_decomposition}.
Also, even if $\rho$ and $\sigma$ do not commute, the pinching makes $\rho$ and $\mathcal{P}_{\rho}\qty(\sigma)$ commute with each other
\begin{align}
\label{eq:pinching_commute}
    \rho\mathcal{P}_{\rho}\qty(\sigma)=\mathcal{P}_{\rho}\qty(\sigma)\rho.
\end{align}
The state $\rho$ remains invariant under the pinching
\begin{align}
\label{eq:pinching_invariant}
    \mathcal{P}_{\rho}\qty(\rho)=\rho.
\end{align}

Given any CQ channels $\Phi_1,\Phi_2\in\mathcal{C}\qty(\mathcal{X}\to\mathcal{H})$, we let $\mathcal{P}_{\Phi_2}\qty[\Phi_1]$ denote a superchannel that converts $\Phi_1$ to a CQ channel outputting, for every $x\in\mathcal{X}$,
\begin{align}
\label{eq:pinching_superchannel}
    \qty(\mathcal{P}_{\Phi_2}\qty[\Phi_1])\qty(x)=\mathcal{P}_{\Phi_2\qty(x)}\qty(\Phi_1\qty(x)),
\end{align}
which we call a pinching superchannel.
Note that when a classical description of the CQ channel $\Phi_2$ is available, depending on $x$, we can implement each pinching channel $\mathcal{P}_{\Phi_2(x)}$ with respect to $\Phi_2(x)$ to realize this pinching superchannel.
We emphasize that the feasibility of this pinching superchannel crucially relies on the fact that the inputs to CQ channels are classical; thus, extending the same definition to QQ channels is generally challenging.

To effectively apply this pinching technique without assuming an IID structure, we introduce a rounding lemma, which plays a key role in our analysis.

\begin{lemma}[\label{lem:rounding}Rounding lemma of full-rank quantum states.]
For any sequence $\qty{C_n>0}_{n=1,2,\ldots}$ of parameters, and any sequence $\qty{\rho^{(n)}\in\mathcal{D}\qty(\mathcal{H}^{\otimes n})}_{n=1,2,\ldots}$ of full-rank quantum states satisfying
\begin{align}
\label{eq:rounding_eigenvalue_lower_bound}
    \rho^{(n)}\geq e^{-C_n n}\mathds{1},
\end{align}
there exists a sequence $\qty{\tilde{\rho}^{(n)}\in\mathcal{D}\qty(\mathcal{H}^{\otimes n})}_{n=1,2,\ldots}$ of quantum states such that, for every $n$, we have
\begin{align}
\label{eq:operator_inequality_modified_state}
    e^{-C_n}\rho^{(n)}\leq\tilde{\rho}^{(n)}\leq e^{C_n}\rho^{(n)},
\end{align}
and the spectral decomposition of $\tilde{\rho}^{(n)}$ is in the form of
\begin{align}
    \tilde{\rho}^{(n)}=\sum_{\tilde{j}=0}^{\tilde{J}_{\tilde{\rho}^{(n)}}-1}\tilde{\lambda}_{\tilde{j}}\tilde{\Pi}_{\tilde{j}},
\end{align}
with the number $\tilde{J}_{\tilde{\rho}^{(n)}}$ of distinct eigenvalues bounded by
\begin{align}
\label{eq:D_n_modified_state}
    \tilde{J}_{\tilde{\rho}^{(n)}}\leq n+1,
\end{align}
where $\tilde{\lambda}_{\tilde{j}}$ is each of the distinct eigenvalues, and $\tilde{\Pi}_{\tilde{j}}$ is the projection onto the eigenspace associated with $\tilde{\lambda}_{\tilde{j}}$.
\end{lemma}

\begin{proof}
    We first provide a construction of $\tilde{\rho}^{(n)}$, followed by proving~\eqref{eq:operator_inequality_modified_state} and~\eqref{eq:D_n_modified_state}.

    \textbf{Construction of $\tilde{\rho}^{(n)}$.}
    We write the spectral decomposition of $\rho^{(n)}$ as
    \begin{align}
        \rho^{(n)}=\sum_{j}\lambda_j\Pi_j.
    \end{align}
    For $\lambda\geq e^{-C_n n}$, using the ceiling function $\lceil{}\cdots{}\rceil$ for rounding, we define a real function $f_n\qty(\lambda)$ as
    \begin{align}
        f_n\qty(\lambda)\coloneqq\exp[-C_n n + C_n\left\lceil n-\frac{\log\qty[\frac{1}{\lambda}]}{C_n}\right\rceil],
    \end{align}
    so that it holds that
    \begin{align}
    \label{eq:inequality_f_n}
        \lambda\leq f_n\qty(\lambda)\leq e^{C_n} \lambda.
    \end{align}
    We define $\tilde{\rho}^{(n)}$ as
    \begin{align}
    \label{eq:tilde_rho_n_definition}
        \tilde{\rho}^{(n)}\coloneqq\frac{\sum_{j}f_n\qty(\lambda_j)\Pi_j}{\Tr\qty[\sum_{j}f_n\qty(\lambda_j)\Pi_j]}.
    \end{align}

    \textbf{Proof of~\eqref{eq:operator_inequality_modified_state}.} We obtain from~\eqref{eq:inequality_f_n}
    \begin{align}
        &\rho^{(n)}\leq \sum_{j}f_n\qty(\lambda_j)\Pi_j\leq e^{C_n} \rho^{(n)},\\
        &1\leq\Tr\qty[\sum_{j}f_n\qty(\lambda_j)\Pi_j]\leq e^{C_n}.
    \end{align}
    Therefore, we have~\eqref{eq:operator_inequality_modified_state} by the definition~\eqref{eq:tilde_rho_n_definition} of $\tilde{\rho}^{(n)}$.
    
    \textbf{Proof of~\eqref{eq:D_n_modified_state}.} Due to~\eqref{eq:rounding_eigenvalue_lower_bound}, we see that the eigenvalues $\lambda_d$ of $\rho^{(n)}$ is within the range of
    \begin{align}
        -C_n n\leq\log\qty[\lambda_j]\leq 1.
    \end{align}
    Thus, $\log\qty[f_n\qty(\lambda_j)]$ takes values in $\qty{a_{\tilde{j}}}_{\tilde{j}=0,1,\ldots,n}$ with
    \begin{align}
        a_{\tilde{j}}\coloneqq-C_n n+C_n\tilde{j},
    \end{align}
    leading to~\eqref{eq:D_n_modified_state}.
\end{proof}

In this rounding lemma, the approximate states $\tilde{\rho}^{(n)}$ have only $O(n)$ distinct eigenvalues.
This property will be instrumental in applying the pinching inequality in~\eqref{eq:pinching_inequality} to obtain a useful operator inequality.
To proceed, we show the following bound on the difference between type II errors obtained from an operator inequality.

\begin{lemma}[\label{lem:bound_on_beta_from_operator_inequality}Difference between type II errors for CQ channels from operator inequalities]
    For any parameters $\epsilon,\underline{C}\geq 0$, and any CQ channels $\Phi_1,\Phi_2,\Phi_2'\in\mathcal{C}\qty(\mathcal{X}\to\mathcal{H})$,
    if we have, for every input $x\in\mathcal{X}$,
    \begin{align}
    \label{eq:bound_on_beta_from_operator_inequality_assumption}
        \Phi_2^{\prime}\qty(x)\geq e^{-\underline{C}}\Phi_2\qty(x),
    \end{align}
    then it holds that
    \begin{align}
        &-\log\qty[\beta_\epsilon\left(\Phi_1\middle\|\Phi_2^{\prime}\right)]\leq-\log\qty[\beta_\epsilon\left(\Phi_1\middle\|\Phi_2\right)]+\underline{C},
    \end{align}
    where $\beta_\epsilon$ is defined as~\eqref{eq:beta_CQ_channels}.
    For any parameters $\epsilon,\overline{C}\geq 0$, and any CQ channels $\Phi_1,\Phi_2,\Phi_2'\in\mathcal{C}\qty(\mathcal{X}\to\mathcal{H})$,
    if we have, for every input $x\in\mathcal{X}$,
    \begin{align}
    \label{eq:bound_on_beta_from_operator_inequality_assumption2}
        \Phi_2^{\prime}\qty(x)\leq e^{\overline{C}}\Phi_2\qty(x),
    \end{align}
    then it holds that
    \begin{align}
        &-\log\qty[\beta_\epsilon\left(\Phi_1\middle\|\Phi_2^{\prime}\right)]\geq -\log\qty[\beta_\epsilon\left(\Phi_1\middle\|\Phi_2\right)]-\overline{C}.
    \end{align}
\end{lemma}

\begin{proof}
Due to~\eqref{eq:bound_on_beta_from_operator_inequality_assumption} and~\eqref{eq:bound_on_beta_from_operator_inequality_assumption2}, for any probability distribution $p$ over $\mathcal{X}$ and any family $\qty{T_x}_{x\in\mathcal{X}}$ of POVM elements, it holds, respectively, that
\begin{align}
    \sum_{x\in\mathcal{X}}p\qty(x)\Tr\qty[T_x\Phi_2'\qty(x)]&\geq e^{-\underline{C}}\sum_{x\in\mathcal{X}}p\qty(x)\Tr\qty[T_x\Phi_2\qty(x)],\\
    \sum_{x\in\mathcal{X}}p\qty(x)\Tr\qty[T_x\Phi_2'\qty(x)]&\leq e^{\overline{C}}\sum_{x\in\mathcal{X}}p\qty(x)\Tr\qty[T_x\Phi_2\qty(x)].
\end{align}
Thus, by the definition~\eqref{eq:beta_CQ_channels} of $\beta_\epsilon$, we have each conclusion.
\end{proof}

From an operator inequality, we can also derive a useful bound on the difference between the channel divergences.

\begin{lemma}[\label{lem:difference_between_channel_divergence_from_operator_inequality}Difference between channel divergence for CQ channels from operator inequalities]
    For any parameters $\epsilon,\underline{C},\overline{C}\geq 0$, and any CQ channels $\Phi_1,\Phi_2,\Phi_2'\in\mathcal{C}\qty(\mathcal{X}\to\mathcal{H})$,
    suppose that, for every input $x\in\mathcal{X}$, 
    \begin{align}
        \supp\qty(\Phi_1\qty(x))\subseteq\supp\qty(\Phi_2\qty(x))=\supp\qty(\Phi_2'\qty(x)).
    \end{align}
    If it holds, for every input $x\in\mathcal{X}$, that
    \begin{align}
    \label{eq:difference_between_channel_divergence_from_operator_inequality_assumption1}
        \Phi_2^{\prime}\qty(x)\geq e^{-\underline{C}}\Phi_2\qty(x),
    \end{align}
    then we have
    \begin{align}
        D\left(\Phi_1\middle\|\Phi_2^{\prime}\right) \leq D\left(\Phi_1\middle\|\Phi_2\right)+\underline{C}.
    \end{align}
    If it holds, for every input $x\in\mathcal{X}$, that
    \begin{align}
    \label{eq:difference_between_channel_divergence_from_operator_inequality_assumption2}
        \Phi_2^{\prime}\qty(x)\leq e^{\overline{C}}\Phi_2\qty(x),
    \end{align}
    then we have
    \begin{align}
        D\left(\Phi_1\middle\|\Phi_2^{\prime}\right)\geq
        &D\left(\Phi_1\middle\|\Phi_2\right)-\overline{C}.
    \end{align}
\end{lemma}

\begin{proof}
Due to the assumptions~\eqref{eq:difference_between_channel_divergence_from_operator_inequality_assumption1} and~\eqref{eq:difference_between_channel_divergence_from_operator_inequality_assumption2}, for every input $x\in\mathcal{X}$, taking $\log$ on the support $\supp\qty(\Phi_2\qty(x))=\supp\qty(\Phi_2'\qty(x))$ yields, respectively,
\begin{align}
    \log\qty[\Phi_2^{\prime}\qty(x)]
    &\geq\log\qty[e^{-\underline{C}}\Phi_2\qty(x)],\\
    \log\qty[\Phi_2^{\prime}\qty(x)]
    &\leq\log[e^{\overline{C}}\Phi_2\qty(x)].
\end{align}
Thus, by the definition~\eqref{eq:D} of $D$, we have, respectively,
\begin{align}
    &D\left(\Phi_1\qty(x)\middle\|\Phi_2^{\prime}\qty(x)\right)\notag\\
    &=\Tr\qty[\Phi_1\qty(x)\qty(\log\qty[\Phi_1\qty(x)]-\log\qty[\Phi_2'\qty(x)])]\\
    &\geq\Tr\qty[\Phi_1\qty(x)\qty(\log\qty[\Phi_1\qty(x)]-\log\qty[\Phi_2\qty(x)])]-\overline{C}\\
    &=D\left(\Phi_1\qty(x)\middle\|\Phi_2\qty(x)\right)-\overline{C},
\end{align}
and
\begin{align}
    &D\left(\Phi_1\qty(x)\middle\|\Phi_2^{\prime}\qty(x)\right)\notag\\
    &=\Tr\qty[\Phi_1\qty(x)\qty(\log\qty[\Phi_1\qty(x)]-\log\qty[\Phi_2'\qty(x)])]\\
    &\leq\Tr\qty[\Phi_1\qty(x)\qty(\log\qty[\Phi_1\qty(x)]-\log\qty[\Phi_2\qty(x)])]+\underline{C}\\
    &=D\left(\Phi_1\qty(x)\middle\|\Phi_2\qty(x)\right)+\underline{C},
\end{align}
which hold for any choice of $x\in\mathcal{X}$.
By the definition~\eqref{eq:D_channel} of $D$ for CQ channels, taking the maximum over $x\in\mathcal{X}$ yields the conclusions.
\end{proof}

Apart from this bound, the operator inequality can also be directly translated into a bound on the channel divergence.

\begin{lemma}[\label{lem:bound_channel_divergence_operator_inequality}Bound on channel divergence for CQ channels from an operator inequality]
    For any parameter $C\geq 0$, and any CQ channels $\Phi_1,\Phi_2\in\mathcal{C}\qty(\mathcal{X}\to\mathcal{H}))$, if it holds, for every input $x\in\mathcal{X}$, that
    \begin{align}
    \label{eq:bound_channel_divergence_operator_inequality_assumption}
        \Phi_1(x)\leq e^C\Phi_2(x),
    \end{align}
    then we have
    \begin{align}
        D\left(\Phi_1\middle\|\Phi_2\right)\leq C,
    \end{align}
    where $D$ is defined as~\eqref{eq:D_channel}.
\end{lemma}

\begin{proof}
    The assumption~\eqref{eq:bound_channel_divergence_operator_inequality_assumption} implies
    \begin{align}
        \supp\qty(\Phi_1(x))\subseteq\supp\qty(\Phi_2(x)),
    \end{align}
    but the supports may be different, and addressing this potential mismatch is the focus of the remainder of the proof.
    To ensure that $\log$ can be taken even if $\Phi_1\qty(x)$ and $\Phi_2\qty(x)$ may have different supports, we introduce a parameter $\delta>0$ and define a full-rank density operator
    \begin{align}
        \Phi_{1,\delta}\qty(x)\coloneqq\frac{\Phi_1(x)+\delta\Phi_2(x)}{1+\delta},
    \end{align}
    which satisfies
    \begin{align}
    \label{eq:continuity_assumption}
        &\lim_{\delta\to 0}d_\diamond\qty(\Phi_{1,\delta},\Phi_1)=0,\\
        &\supp\qty(\Phi_{1,\delta}(x))=\supp\qty(\Phi_2(x)),
    \end{align}
    where $d_\diamond$ is defined as~\eqref{eq:d_diamond}.
    Under the assumption~\eqref{eq:bound_channel_divergence_operator_inequality_assumption},
    we have
    \begin{align}
        \Phi_{1,\delta}(x)\leq \frac{e^C+\delta}{1+\delta}\Phi_{2,\delta}(x).
    \end{align}
    Due to the operator monotonicity of $\log$~\cite{Lowner1934} (see also Ref.~\cite[Example 2.5.9.]{FumioHIAI2010IIS160201}),
    on the support of $\Phi_2(x)$,
    it holds that
    \begin{align}
        \log\qty[\Phi_{1,\delta}\qty(x)]&\leq\log\qty[e^C\Phi_{2}(x)].
    \end{align}
    Thus, for $D$ in~\eqref{eq:D}, we have
    \begin{align}
        &D\left(\Phi_{1,\delta}\qty(x)\middle\|\Phi_{2}\qty(x)\right)\notag\\
        &=\Tr\qty[\Phi_{1,\delta}\qty(x)\qty(\log\qty[\Phi_{1,\delta}\qty(x)]-\log\qty[\Phi_{2}(x)])]\\
        &\leq\log\qty[\frac{e^C+\delta}{1+\delta}],
    \end{align}
    which holds for any choice of $\delta>0$.
    Therefore, by taking the limit $\delta\to 0$, the continuity~\eqref{eq:D_continuity} of $D$ with respect to the first argument yields
    \begin{align}
        &D\left(\Phi_{1}\qty(x)\middle\|\Phi_{2}\qty(x)\right)=\lim_{\delta\to 0}D\left(\Phi_{1,\delta}\qty(x)\middle\|\Phi_{2}\qty(x)\right)\leq C,
    \end{align}
    which holds for any choice of $x\in\mathcal{X}$.
    By the definition~\eqref{eq:D_channel} of $D$ for CQ channels, taking the maximum over $x\in\mathcal{X}$ yields
    \begin{align}
        D\left(\Phi_1\middle\|\Phi_2\right)\leq C.
    \end{align}
\end{proof}

Using these ingredients, we obtain the following approximation of CQ channels via the pinching technique.
An analogous approximation for states played a central role in proving the state version of the generalized quantum Stein's lemma in Ref.~\cite{hayashi2025generalizedquantumsteinslemma}.
Our contribution here is to extend this technique beyond the state setting and demonstrate that it applies to CQ channels even in the presence of multiple possible inputs.
Importantly, this extension is made possible by exploiting the classical nature of the inputs of CQ channels, which allows us to overcome obstacles that make a similar generalization to QQ channels inherently difficult.

\begin{lemma}[\label{lem:approximation}Approximation of CQ channels under full-rank condition]
For any sequence $\qty{C_n>0}_{n=1,2,\ldots}$ of parameters, and any sequences $\qty{\Phi_1^{(n)}\in\mathcal{C}\qty(\mathcal{X}^n\to\mathcal{H}^{\otimes n}))}_n$ and $\qty{\Phi_2^{(n)}\in\mathcal{C}\qty(\mathcal{X}^n\to\mathcal{H}^{\otimes n}))}_n$ of CQ channels satisfying, for every input $x^{(n)}\in\mathcal{X}^n$,
\begin{align}
    \label{eq:approximation_assumption}
    \Phi_2^{(n)}\qty(x^{(n)})\geq e^{-C_n n}\mathds{1},
\end{align}
there exist sequences $\qty{\tilde{\Phi}_1^{(n)}\in\mathcal{C}\qty(\mathcal{X}^n\to\mathcal{H}^{\otimes n}))}_n$ and $\qty{\tilde{\Phi}_2^{(n)}\in\mathcal{C}\qty(\mathcal{X}^n\to\mathcal{H}^{\otimes n}))}_n$ of CQ channels satisfying the following:
\begin{enumerate}
    \item (commutativity) for every input $x^{(n)}\in\mathcal{X}^n$, $\tilde{\Phi}_1^{(n)}$ and $\tilde{\Phi}_2^{(n)}$ commute with each other
    \begin{align}
    \label{eq:approximation_1}
        \tilde{\Phi}_1^{(n)}\qty(x^{(n)})\tilde{\Phi}_2^{(n)}\qty(x^{(n)})=\tilde{\Phi}_2^{(n)}\qty(x^{(n)})\tilde{\Phi}_1^{(n)}\qty(x^{(n)});
    \end{align}
    \item (approximation in operator inequality) for every input $x^{(n)}\in\mathcal{X}^n$, $\tilde{\Phi}_2^{(n)}$ satisfies
    \begin{align}
    \label{eq:approximation_2}
        e^{-C_n}\Phi_2^{(n)}\qty(x^{(n)})&\leq\tilde{\Phi}_2^{(n)}\qty(x^{(n)})\leq e^{C_n}\Phi_2^{(n)}\qty(x^{(n)});
    \end{align}
    \item (distinguishability bounds) for any $\epsilon\geq 0$,
     if we have
    \begin{align}
    \label{eq:approximation_assumption_3_4}
        C_n=o(n)~\text{as $n\to\infty$},
    \end{align}
    it holds that
    \begin{align}
    \label{eq:approximation_3}
        &\liminf_{n\to\infty}-\frac{1}{n}\log\qty[\beta_\epsilon\left(\tilde{\Phi}_1^{(n)}\middle\|\tilde{\Phi}_2^{(n)}\right)]\notag\\
        &\leq\liminf_{n\to\infty}-\frac{1}{n}\log\qty[\beta_\epsilon\left(\Phi_1^{(n)}\middle\|\Phi_2^{(n)}\right)],\\
    \label{eq:approximation_4}
        &\limsup_{n\to\infty}-\frac{1}{n}\log\qty[\beta_\epsilon\left(\tilde{\Phi}_1^{(n)}\middle\|\tilde{\Phi}_2^{(n)}\right)]\notag\\
        &\leq\limsup_{n\to\infty}-\frac{1}{n}\log\qty[\beta_\epsilon\left(\Phi_1^{(n)}\middle\|\Phi_2^{(n)}\right)],
    \end{align}
    where $\beta_\epsilon$ is defined as~\eqref{eq:beta_CQ_channels};
    \item (invariance of regularized quantum relative entropy) if we have
    \begin{align}
    \label{eq:approximation_assumption_5_6}
        C_n=o(n)~\text{as $n\to\infty$},
    \end{align}
    it holds that
    \begin{align}
    \label{eq:approximation_5}
        \liminf_{n\to\infty}\frac{1}{n}D\left(\tilde{\Phi}_1^{(n)}\middle\|\tilde{\Phi}_2^{(n)}\right)&=\liminf_{n\to\infty}\frac{1}{n}D\left(\Phi_1^{(n)}\middle\|\Phi_2^{(n)}\right),\\
    \label{eq:approximation_6}
        \limsup_{n\to\infty}\frac{1}{n}D\left(\tilde{\Phi}_1^{(n)}\middle\|\tilde{\Phi}_2^{(n)}\right)&=\limsup_{n\to\infty}\frac{1}{n}D\left(\Phi_1^{(n)}\middle\|\Phi_2^{(n)}\right),
    \end{align}
    where $D$ is defined as~\eqref{eq:D_channel}.
\end{enumerate}
\end{lemma}
\begin{proof}
We first provide construction of $\qty{\tilde{\Phi}_1^{(n)}}_n$ and $\qty{\tilde{\Phi}_2^{(n)}}_n$, followed by proving~\eqref{eq:approximation_1},~\eqref{eq:approximation_2},~\eqref{eq:approximation_3},~\eqref{eq:approximation_4},~\eqref{eq:approximation_5}, and~\eqref{eq:approximation_6}.

\textbf{Construction of $\qty{\tilde{\Phi}_1^{(n)}}_n$ and $\qty{\tilde{\Phi}_2^{(n)}}_n$.}
For every state $\Phi_2^{(n)}\qty(x^{(n)})$ with $x^{(n)}\in\mathcal{X}^n$, due to~\eqref{eq:approximation_assumption}, Lemma~\ref{lem:rounding} provides
$\tilde{\Phi}_2^{(n)}\qty(x^{(n)})$ satisfying
\begin{align}
\label{eq:approximation_tilde_Phi_2_operator_inequality}
    &e^{-C_n}\Phi_2^{(n)}\qty(x^{(n)})\leq\tilde{\Phi}_2^{(n)}\qty(x^{(n)})\leq e^{C_n}\Phi_2^{(n)}\qty(x^{(n)}),\\
\label{eq:approximation_tilde_Phi_2_J}
    &J_{\tilde{\Phi}_2^{(n)}\qty(x^{(n)})}\leq n+1,
\end{align}
which yields a CQ channel $\tilde{\Phi}_2^{(n)}$.
Using the superchannel $\mathcal{P}_{\tilde{\Phi}_2^{(n)}}$ as in~\eqref{eq:pinching_superchannel}, we define
\begin{align}
\label{eq:approximation_tilde_Phi_1}
    \tilde{\Phi}_1^{(n)}\coloneqq\mathcal{P}_{\tilde{\Phi}_2^{(n)}}\qty[\Phi_1^{(n)}].
\end{align}

\textbf{Proof of~\eqref{eq:approximation_1}.}
Due to the commutativity of the operators after the pinching as in~\eqref{eq:pinching_commute}, by the definition~\eqref{eq:approximation_tilde_Phi_1} of $\tilde{\Phi}_1^{(n)}$, we obtain~\eqref{eq:approximation_1}.

\textbf{Proof of~\eqref{eq:approximation_2}.}
We have~\eqref{eq:approximation_2} due to~\eqref{eq:approximation_tilde_Phi_2_operator_inequality}.

\textbf{Proof of~\eqref{eq:approximation_3} and~\eqref{eq:approximation_4}.}
It holds that
\begin{align}
    \label{eq:approximation_7}
    &-\frac{1}{n}\log\qty[\beta_{\epsilon}\left(\tilde{\Phi}_1^{(n)}\middle\|\tilde{\Phi}_2^{(n)}\right)]\notag\\
    &=-\frac{1}{n}\log\qty[\beta_{\epsilon}\left(\mathcal{P}_{\tilde{\Phi}_2^{(n)}}\qty[\Phi_1^{(n)}]\middle\|\tilde{\Phi}_2^{(n)}\right)]\\
    \label{eq:approximation_8}
    &=-\frac{1}{n}\log\qty[\beta_{\epsilon}\left(\mathcal{P}_{\tilde{\Phi}_2^{(n)}}\qty[\Phi_1^{(n)}]\middle\|\mathcal{P}_{\tilde{\Phi}_2^{(n)}}\qty[\tilde{\Phi}_2^{(n)}]\right)]\\
    \label{eq:approximation_9}
    &\leq-\frac{1}{n}\log\qty[\beta_{\epsilon}\left(\Phi_1^{(n)}\middle\|\tilde{\Phi}_2^{(n)}\right)],
\end{align}
where~\eqref{eq:approximation_7} is the definition~\eqref{eq:approximation_tilde_Phi_1} of $\tilde{\Phi}_1^{(n)}$,~\eqref{eq:approximation_8} follows from~\eqref{eq:pinching_invariant}, and~\eqref{eq:approximation_9} follows from Lemma~\ref{lem:monotonicity_beta_CQ_channel}.
By taking the limit $n\to\infty$, we obtain
\begin{align}
\label{eq:liminf_beta1}
    &\liminf_{n\to\infty}-\frac{1}{n}\log\qty[\beta_{\epsilon}\left(\tilde{\Phi}_1^{(n)}\middle\|\tilde{\Phi}_2^{(n)}\right)]\notag\\
    &\leq\liminf_{n\to\infty}-\frac{1}{n}\log\qty[\beta_{\epsilon}\left(\Phi_1^{(n)}\middle\|\tilde{\Phi}_2^{(n)}\right)],\\
\label{eq:limsup_beta1}
    &\limsup_{n\to\infty}-\frac{1}{n}\log\qty[\beta_{\epsilon}\left(\tilde{\Phi}_1^{(n)}\middle\|\tilde{\Phi}_2^{(n)}\right)]\notag\\
    &\leq\limsup_{n\to\infty}-\frac{1}{n}\log\qty[\beta_{\epsilon}\left(\Phi_1^{(n)}\middle\|\tilde{\Phi}_2^{(n)}\right)].
\end{align}

On the other hand, due to~\eqref{eq:approximation_tilde_Phi_2_operator_inequality}, Lemma~\ref{lem:bound_on_beta_from_operator_inequality} shows
\begin{align}
    &-\frac{1}{n}\log\qty[\beta_{\epsilon}\left(\Phi_1^{(n)}\middle\|\Phi_2^{(n)}\right)]-\frac{C_n}{n}\notag\\
    &\leq-\frac{1}{n}\log\qty[\beta_{\epsilon}\left(\Phi_1^{(n)}\middle\|\tilde{\Phi}_2^{(n)}\right)]\notag\\
    &\leq-\frac{1}{n}\log\qty[\beta_{\epsilon}\left(\Phi_1^{(n)}\middle\|\Phi_2^{(n)}\right)]+\frac{C_n}{n}.
\end{align}
Under the assumption~\eqref{eq:approximation_assumption_3_4} of $C_n=o(n)$, taking the limit $n\to\infty$ yields
\begin{align}
\label{eq:liminf_beta2}
    &\liminf_{n\to\infty}-\frac{1}{n}\log\qty[\beta_{\epsilon}\left(\Phi_1^{(n)}\middle\|\tilde{\Phi}_2^{(n)}\right)]\notag\\
    &=\liminf_{n\to\infty}-\frac{1}{n}\log\qty[\beta_{\epsilon}\left(\Phi_1^{(n)}\middle\|\Phi_2^{(n)}\right)],\\
\label{eq:limsup_beta2}
    &\limsup_{n\to\infty}-\frac{1}{n}\log\qty[\beta_{\epsilon}\left(\Phi_1^{(n)}\middle\|\tilde{\Phi}_2^{(n)}\right)]\notag\\
    &=\limsup_{n\to\infty}-\frac{1}{n}\log\qty[\beta_{\epsilon}\left(\Phi_1^{(n)}\middle\|\Phi_2^{(n)}\right)].
\end{align}
Consequently,~\eqref{eq:approximation_3} follows from~\eqref{eq:liminf_beta1} and~\eqref{eq:liminf_beta2}, and~\eqref{eq:approximation_4} follows from~\eqref{eq:limsup_beta1} and~\eqref{eq:limsup_beta2}.

\textbf{Proof of~\eqref{eq:approximation_5} and~\eqref{eq:approximation_6}.}
We evaluate the difference between $D\left(\tilde{\Phi}_1^{(n)}\middle\|\tilde{\Phi}_2^{(n)}\right)$ and $D\left(\Phi_1^{(n)}\middle\|\tilde{\Phi}_2^{(n)}\right)$.
As in the case of states in Ref.~\cite[Lemma 3.1]{hiai1991proper}, for every input $x^{(n)}\in\mathcal{X}^n$, it holds that
\begin{widetext}
\begin{align}
    &D\left(\Phi_1^{(n)}\qty(x^{(n)})\middle\|\tilde{\Phi}_2^{(n)}\qty(x^{(n)})\right)-D\left(\tilde{\Phi}_1^{(n)}\qty(x^{(n)})\middle\|\tilde{\Phi}_2^{(n)}\qty(x^{(n)})\right)\notag\\
    &=\Tr\qty[\Phi_1^{(n)}\qty(x^{(n)})\qty(\log\qty[\Phi_1^{(n)}\qty(x^{(n)})]-\log\qty[\tilde{\Phi}_2^{(n)}\qty(x^{(n)})])]-\notag\\
    &\quad\Tr\qty[\qty(\mathcal{P}_{\tilde{\Phi}_2^{(n)}}\qty[\Phi_1^{(n)}])\qty(x^{(n)})\qty(\log\qty[\tilde{\Phi}_1^{(n)}\qty(x^{(n)})]-\log\qty[\tilde{\Phi}_2^{(n)}\qty(x^{(n)})])]\\
    \label{eq:channel_divergence_under_pinching_2}
    &=\Tr\qty[\Phi_1^{(n)}\qty(x^{(n)})\qty(\log\qty[\Phi_1^{(n)}\qty(x^{(n)})]-\log\qty[\tilde{\Phi}_2^{(n)}\qty(x^{(n)})])]-\notag\\
    &\quad\Tr\qty[\Phi_1^{(n)}\qty(x^{(n)})\qty(\log\qty[\tilde{\Phi}_1^{(n)}\qty(x^{(n)})]-\log\qty[\tilde{\Phi}_2^{(n)}\qty(x^{(n)})])]\\
    \label{eq:channel_divergence_under_pinching_1}
    &=\Tr\qty[\Phi_1^{(n)}\qty(x^{(n)})\qty(\log\qty[\Phi_1^{(n)}\qty(x^{(n)})]-\log\qty[\tilde{\Phi}_1^{(n)}\qty(x^{(n)})])]\\
    \label{eq:channel_divergence_under_pinching_3}
    &=D\left(\Phi_1^{(n)}\qty(x^{(n)})\middle\|\tilde{\Phi}_1^{(n)}\qty(x^{(n)})\right)\geq 0,
\end{align}
\end{widetext}
where~\eqref{eq:channel_divergence_under_pinching_2} holds since the pinching makes $\tilde{\Phi}_1^{(n)}\qty(x^{(n)})$ and $\tilde{\Phi}_2^{(n)}\qty(x^{(n)})$ commutative, so the trace can be taken using their simultaneous eigenbasis,
and~\eqref{eq:channel_divergence_under_pinching_1} follows from the linearity of the trace.
Thus, by taking $x^{(n)\ast}\in\mathcal{X}^n$ as that achieving the maximum in the definition~\eqref{eq:D_channel} of $D\left(\tilde{\Phi}_1^{(n)}\middle\|\tilde{\Phi}_2^{(n)}\right)$, i.e.,
\begin{align}
    D\left(\tilde{\Phi}_1^{(n)}\qty(x^{(n)\ast})\middle\|\tilde{\Phi}_2^{(n)}\qty(x^{(n)\ast})\right)=D\left(\tilde{\Phi}_1^{(n)}\middle\|\tilde{\Phi}_2^{(n)}\right),
\end{align}
we have
\begin{align}
    &\frac{1}{n}D\left(\tilde{\Phi}_1^{(n)}\middle\|\tilde{\Phi}_2^{(n)}\right)\notag\\
    &=\frac{1}{n}D\left(\tilde{\Phi}_1^{(n)}\qty(x^{(n)\ast})\middle\|\tilde{\Phi}_2^{(n)}\qty(x^{(n)\ast})\right)\\
    &\leq\frac{1}{n}D\left(\Phi_1^{(n)}\qty(x^{(n)\ast})\middle\|\tilde{\Phi}_2^{(n)}\qty(x^{(n)\ast})\right)\\
    \label{eq:channel_divergence_under_pinching_4}
    &\leq\frac{1}{n}D\left(\Phi_1^{(n)}\middle\|\tilde{\Phi}_2^{(n)}\right).
\end{align}

On the other hand, we will derive the opposite inequality.
To this end, by the definition~\eqref{eq:approximation_tilde_Phi_1} of $\tilde{\Phi}_1^{(n)}$, the pinching inequality~\eqref{eq:pinching_inequality} yields
\begin{align}
    \Phi_1^{(n)}\qty(x^{(n)})\leq (n+1)\tilde{\Phi}_1^{(n)}\qty(x^{(n)}),
\end{align}
where we use~\eqref{eq:approximation_tilde_Phi_2_J}.
Thus, Lemma~\ref{lem:bound_channel_divergence_operator_inequality} shows
\begin{align}
    D\left(\Phi_1^{(n)}\qty(x^{(n)})\middle\|\tilde{\Phi}_1^{(n)}\qty(x^{(n)})\right)\leq\log\qty[n+1].
\end{align}
Therefore, due to~\eqref{eq:channel_divergence_under_pinching_3}, we have
\begin{align}
    &\frac{1}{n}D\left(\Phi_1^{(n)}\qty(x^{(n)})\middle\|\tilde{\Phi}_2^{(n)}\qty(x^{(n)})\right)\notag\\
    &\leq\frac{1}{n}D\left(\tilde{\Phi}_1^{(n)}\qty(x^{(n)})\middle\|\tilde{\Phi}_2^{(n)}\qty(x^{(n)})\right)+\frac{\log\qty[n+1]}{n},
\end{align}
which holds for any choice of $x^{(n)}\in\mathcal{X}^n$.
By taking $x^{(n)\ast\ast}\in\mathcal{X}^n$ as that achieving the maximum in the definition~\eqref{eq:D_channel} of $D\left(\Phi_1^{(n)}\middle\|\tilde{\Phi}_2^{(n)}\right)$, i.e.,
\begin{align}
    D\left(\Phi_1^{(n)}\qty(x^{(n)\ast\ast})\middle\|\tilde{\Phi}_2^{(n)}\qty(x^{(n)\ast\ast})\right)=D\left(\Phi_1^{(n)}\middle\|\tilde{\Phi}_2^{(n)}\right),
\end{align}
we obtain
\begin{align}
    &\frac{1}{n}D\left(\Phi_1^{(n)}\middle\|\tilde{\Phi}_2^{(n)}\right)\notag\\
    &=\frac{1}{n}D\left(\Phi_1^{(n)}\qty(x^{(n)\ast\ast})\middle\|\tilde{\Phi}_2^{(n)}\qty(x^{(n)\ast\ast})\right)\\
    &\leq\frac{1}{n}D\left(\tilde{\Phi}_1^{(n)}\qty(x^{(n)\ast\ast})\middle\|\tilde{\Phi}_2^{(n)}\qty(x^{(n)\ast\ast})\right)+\frac{\log\qty[n+1]}{n}\\
    \label{eq:channel_divergence_under_pinching_5}
    &\leq\frac{1}{n}D\left(\tilde{\Phi}_1^{(n)}\middle\|\tilde{\Phi}_2^{(n)}\right)+\frac{\log\qty[n+1]}{n}.
\end{align}

Consequently,~\eqref{eq:channel_divergence_under_pinching_4} and~\eqref{eq:channel_divergence_under_pinching_5} lead to
\begin{align}
&\frac{1}{n}D\left(\tilde{\Phi}_1^{(n)}\middle\|\tilde{\Phi}_2^{(n)}\right)\notag\\
&\leq\frac{1}{n}D\left(\Phi_1^{(n)}\middle\|\tilde{\Phi}_2^{(n)}\right)\notag\\
&\leq\frac{1}{n}D\left(\tilde{\Phi}_1^{(n)}\middle\|\tilde{\Phi}_2^{(n)}\right)+\frac{\log\qty[n+1]}{n}.
\end{align}
Thus, taking the limit $n\to\infty$ yields
\begin{align}
\label{eq:liminf_D1}
\liminf_{n\to\infty}\frac{1}{n}D\left(\tilde{\Phi}_1^{(n)}\middle\|\tilde{\Phi}_2^{(n)}\right)&=\liminf_{n\to\infty}\frac{1}{n}D\left(\Phi_1^{(n)}\middle\|\tilde{\Phi}_2^{(n)}\right),\\
\label{eq:limsup_D1}
\limsup_{n\to\infty}\frac{1}{n}D\left(\tilde{\Phi}_1^{(n)}\middle\|\tilde{\Phi}_2^{(n)}\right)&=\limsup_{n\to\infty}\frac{1}{n}D\left(\Phi_1^{(n)}\middle\|\tilde{\Phi}_2^{(n)}\right).
\end{align}

On the other hand, due to~\eqref{eq:approximation_tilde_Phi_2_operator_inequality}, Lemma~\ref{lem:difference_between_channel_divergence_from_operator_inequality} shows
\begin{align}
    &\frac{1}{n}D\left(\Phi_1^{(n)}\middle\|\Phi_2^{(n)}\right)-\frac{C_n}{n}\notag\\
    &\leq\frac{1}{n}D\left(\Phi_1^{(n)}\middle\|\tilde{\Phi}_2^{(n)}\right)\notag\\
    &\leq\frac{1}{n}D\left(\Phi_1^{(n)}\middle\|\Phi_2^{(n)}\right)+\frac{C_n}{n}.
\end{align}
Under the assumption~\eqref{eq:approximation_assumption_5_6} of $C_n=o(n)$, taking the limit $n\to\infty$ yields
\begin{align}
\label{eq:liminf_D2}
\liminf_{n\to\infty}\frac{1}{n}D\left(\Phi_1^{(n)}\middle\|\tilde{\Phi}_2^{(n)}\right)&=\liminf_{n\to\infty}\frac{1}{n}D\left(\Phi_1^{(n)}\middle\|\Phi_2^{(n)}\right),\\
\label{eq:limsup_D2}
\limsup_{n\to\infty}\frac{1}{n}D\left(\Phi_1^{(n)}\middle\|\tilde{\Phi}_2^{(n)}\right)&=\limsup_{n\to\infty}\frac{1}{n}D\left(\Phi_1^{(n)}\middle\|\Phi_2^{(n)}\right).
\end{align}
Consequently,~\eqref{eq:approximation_5} follows from~\eqref{eq:liminf_D1} and~\eqref{eq:liminf_D2}, and~\eqref{eq:approximation_6} follows from~\eqref{eq:limsup_D1} and~\eqref{eq:limsup_D2}.
\end{proof}

\paragraph{The information spectrum method for CQ channels}
\label{sec:information_spectrum}

Given the approximate CQ channels obtained from the pinching technique, we next apply the information spectrum method, which provides projection operators onto typical subspaces based on bounds for type II errors in quantum hypothesis testing~\cite{4069150}.
For any operators $A$ and $B$, let
\begin{align}
\label{eq:projection_A_B}
\qty{A \geq B}
\end{align}
denote the projection onto the eigenspace corresponding to the non-negative eigenvalues of $A-B$.
In Ref.~\cite{hayashi2025generalizedquantumsteinslemma}, this projection was employed to analyze the state version of the generalized quantum Stein's lemma.
Here, we establish an analogous method for CQ channels, extending beyond the single-state setting to handle multiple possible inputs, thereby demonstrating its applicability in a context of dynamical resources.

By extending the information spectrum method for states used in Ref.~\cite{hayashi2025generalizedquantumsteinslemma}, we obtain the following lemma.
This extension plays a crucial role in analyzing the CQ-channel version of the generalized quantum Stein's lemma.

\begin{lemma}[\label{lem:information_spectrum}Information spectrum method for CQ channels]
    For any parameters $\epsilon\geq 0$, $\underline{R}$, $\overline{R}$, any sequences $\qty{\Phi_1^{(n)}\in\mathcal{C}\qty(\mathcal{X}^n\to\mathcal{H}^{\otimes n})}_{n=1,2,\ldots}$ and $\qty{\Phi_2^{(n)}\in\mathcal{C}\qty(\mathcal{X}^n\to\mathcal{H}^{\otimes n})}_{n=1,2,\ldots}$, and any sequence $\qty{p_n}_{n=1,2,\ldots}$ of probability distributions over $\mathcal{X}^n$,
    if it holds that
    \begin{align}
        \underline{R}&>\liminf_{n\to\infty}-\frac{1}{n}\log\qty[\beta_{\epsilon}\left(\Phi_1^{(n)}\middle\|\Phi_2^{(n)}\right)],\\
        \overline{R}&>\limsup_{n\to\infty}-\frac{1}{n}\log\qty[\beta_{\epsilon}\left(\Phi_1^{(n)}\middle\|\Phi_2^{(n)}\right)],
    \end{align}
    then the families of projections
    \begin{align}
    \label{eq:projection_underline}
        &\qty{\underline{T}_{x^{(n)}}\coloneqq\qty{\Phi_1^{(n)}\qty(x^{(n)})\geq e^{\underline{R}n}\Phi_2^{(n)}\qty(x^{(n)})}}_{x^{(n)}\in\mathcal{X}^n},\\
    \label{eq:projection_overrline}
        &\qty{\overline{T}_{x^{(n)}}\coloneqq\qty{\Phi_1^{(n)}\qty(x^{(n)})\geq e^{\overline{R}n}\Phi_2^{(n)}\qty(x^{(n)})}}_{x^{(n)}\in\mathcal{X}^n}
    \end{align}
    satisfy, respectively,
    \begin{align}
    \label{eq:information_spectrum_underline}
        \limsup_{n\to\infty}\sum_{x^{(n)}\in\mathcal{X}^n}p_n\qty(x^{(n)})\Tr\qty[\qty(\mathds{1}-\underline{T}_{x^{(n)}})\Phi_1^{(n)}\qty(x^{(n)})]&>\epsilon,\\
    \label{eq:information_spectrum_overline}
        \liminf_{n\to\infty}\sum_{x^{(n)}\in\mathcal{X}^n}p_n\qty(x^{(n)})\Tr\qty[\qty(\mathds{1}-\overline{T}_{x^{(n)}})\Phi_1^{(n)}\qty(x^{(n)})]&>\epsilon,
    \end{align}
    where $\beta_\epsilon$ is defined as~\eqref{eq:beta_CQ_channels},
    and the notations on the projections~\eqref{eq:projection_underline} and~\eqref{eq:projection_overrline} follow~\eqref{eq:projection_A_B}.
\end{lemma}

\begin{proof}
    We will prove contraposition of~\eqref{eq:information_spectrum_underline} and~\eqref{eq:information_spectrum_overline}.

    \textbf{Proof of~\eqref{eq:information_spectrum_underline}.}
    To prove the contraposition of~\eqref{eq:information_spectrum_underline}, suppose that
    \begin{align}
    \label{eq:information_spectrum_underline_assumption}
        \limsup_{n\to\infty}\sum_{x^{(n)}\in\mathcal{X}^n}p_n\qty(x^{(n)})\Tr\qty[\qty(\mathds{1}-\underline{T}_{x^{(n)}})\Phi_1^{(n)}\qty(x^{(n)})]&\leq\epsilon,
    \end{align}
    so that we will prove
    \begin{align}
    \label{eq:information_spectrum_underline_conclusion}
        \liminf_{n\to\infty}-\frac{1}{n}\log\qty[\beta_{\epsilon}\left(\Phi_1^{(n)}\middle\|\Phi_2^{(n)}\right)]\geq\underline{R}.
    \end{align}
    
    Under~\eqref{eq:information_spectrum_underline_assumption}, there exists $n_0$ such that, for all $n\geq n_0$, we have
    \begin{align}
        \qty{\underline{T}_{x^{(n)}}}_{x^{(n)}}\in\mathcal{T}_{\epsilon,\Phi_1^{(n)},p_n},
    \end{align}
    where $\mathcal{T}_{\epsilon,\Phi_1^{(n)},p_n}$ is defined as~\eqref{eq:T_set}.
    Thus, for every $n\geq n_0$, it holds that
    \begin{align}
        &-\frac{1}{n}\log\qty[\beta_{\epsilon}\left(\Phi_1^{(n)}\middle\|\Phi_2^{(n)}\right)]\notag\\
        &\geq -\frac{1}{n}\log\qty[\sum_{x^{(n)}\in\mathcal{X}^n}p_n\qty(x^{(n)})\Tr\qty[\underline{T}_{x^{(n)}}\Phi_2^{(n)}\qty(x^{(n)})]].
    \end{align}
    
    On the other hand, for every $x^{(n)}\in\mathcal{X}^n$, by the definition~\eqref{eq:projection_underline} of $\underline{T}_{x^{(n)}}$, we have
    \begin{align}
        \Tr\qty[\underline{T}_{x^{(n)}}\qty(\Phi_1^{(n)}\qty(x^{(n)})- e^{\underline{R}n}\Phi_2^{(n)}\qty(x^{(n)}))]\geq 0,
    \end{align}
    and hence,
    \begin{align}
        \Tr\qty[\underline{T}_{x^{(n)}}\Phi_2^{(n)}\qty(x^{(n)})] &\leq e^{-\underline{R}n}\Tr\qty[\underline{T}_{x^{(n)}}\Phi_1^{(n)}\qty(x^{(n)})]\\
        &\leq e^{-\underline{R}n}.
    \end{align}
    Thus, it holds that
    \begin{align}
        &-\frac{1}{n}\log\qty[\sum_{x^{(n)}\in\mathcal{X}^n}p_n\qty(x^{(n)})\Tr\qty[\underline{T}_{x^{(n)}}\Phi_2^{(n)}\qty(x^{(n)})]]\notag\\
        &\geq\underline{R}.
    \end{align}

    Consequently, for every $n\geq n_0$, it holds that
    \begin{align}
        &-\frac{1}{n}\log\qty[\beta_{\epsilon}\left(\Phi_1^{(n)}\middle\|\Phi_2^{(n)}\right)]\geq\underline{R},
    \end{align}
    which yields~\eqref{eq:information_spectrum_underline_conclusion}.

    \textbf{Proof of~\eqref{eq:information_spectrum_overline}.}
    To prove the contraposition of~\eqref{eq:information_spectrum_overline}, suppose that
    \begin{align}
    \label{eq:information_spectrum_overline_assumption}
        \liminf_{n\to\infty}\sum_{x^{(n)}\in\mathcal{X}^n}p_n\qty(x^{(n)})\Tr\qty[\qty(\mathds{1}-\overline{T}_{x^{(n)}})\Phi_1^{(n)}\qty(x^{(n)})]&\leq\epsilon,
    \end{align}
    so that we will prove
    \begin{align}
    \label{eq:information_spectrum_overline_conclusion}
        \limsup_{n\to\infty}-\frac{1}{n}\log\qty[\beta_{\epsilon}\left(\Phi_1^{(n)}\middle\|\Phi_2^{(n)}\right)]\geq\overline{R}.
    \end{align}
    
    Under~\eqref{eq:information_spectrum_overline_assumption}, there exists a subsequence $\qty{n_l}_{l=1,2,\ldots}$ of $\{1,2,\ldots\}$ such that, for all $l$, we have
    \begin{align}
        \qty{\overline{T}_{x^{\qty(n_l)}}}_{x^{\qty(n_l)}}\in\mathcal{T}_{\epsilon,\Phi_1^{\qty(n_l)},p_{n_l}},
    \end{align}
    where $\mathcal{T}_{\epsilon,\Phi_1^{\qty(n_l)},p_{n_l}}$ is defined as~\eqref{eq:T_set}.
    Thus, for every $l$, it holds that
    \begin{align}
        &-\frac{1}{n_l}\log\qty[\beta_{\epsilon}\left(\Phi_1^{\qty(n_l)}\middle\|\Phi_2^{\qty(n_l)}\right)]\notag\\
        &\geq -\frac{1}{n_l}\log\qty[\sum_{x^{\qty(n_l)}\in\mathcal{X}^{n_l}}p_{n_l}\qty(x^{\qty(n_l)})\Tr\qty[\overline{T}_{x^{\qty(n_l)}}\Phi_2^{\qty(n_l)}\qty(x^{\qty(n_l)})]].
    \end{align}
    
    On the other hand, for every $x^{(n_l)}\in\mathcal{X}^{n_l}$, we have
    \begin{align}
        \Tr\qty[\overline{T}_{x^{\qty(n_l)}}\qty(\Phi_1^{\qty(n_l)}\qty(x^{\qty(n_l)})- e^{\overline{R}n_l}\Phi_2^{\qty(n_l)}\qty(x^{\qty(n_l)}))]\geq 0,
    \end{align}
    and hence,
    \begin{align}
        \Tr\qty[\overline{T}_{x^{\qty(n_l)}}\Phi_2^{\qty(n_l)}\qty(x^{\qty(n_l)})] &\leq e^{-\overline{R}n_l}\Tr\qty[\overline{T}_{x^{\qty(n_l)}}\Phi_1^{\qty(n_l)}\qty(x^{\qty(n_l)})]\\
        &\leq e^{-\overline{R}n_l}.
    \end{align}
    Thus, it holds that
    \begin{align}
        &-\frac{1}{n_l}\log\qty[\sum_{x^{\qty(n_l)}\in\mathcal{X}^{n_l}}p_{n_l}\qty(x^{\qty(n_l)})\Tr\qty[\overline{T}_{x^{\qty(n_l)}}\Phi_2^{\qty(n_l)}\qty(x^{\qty(n_l)})]]\notag\\
        &\geq\overline{R}.
    \end{align}

    Consequently, for every $n_l$ in the subsequence, it holds that
    \begin{align}
        &-\frac{1}{n_l}\log\qty[\beta_{\epsilon}\left(\Phi_1^{\qty(n_l)}\middle\|\Phi_2^{\qty(n_l)}\right)]\geq\overline{R},
    \end{align}
    which yields~\eqref{eq:information_spectrum_overline_conclusion}.
\end{proof}

To apply the information spectrum method, the following bound on type II errors in quantum hypothesis testing for CQ dynamical resources will be particularly useful.

\begin{lemma}[\label{lem:strong_converse_CQ_channel_copies}Upper bounds for the type II errors for CQ channels]
    For any parameter $\epsilon\in[0,1)$, any fixed $M$,
    any CQ channel $\Phi_1\in\mathcal{C}\qty(\mathcal{X}\to\mathcal{H})$,
    any CQ channel $\Phi_\mathrm{full}\in\mathcal{C}\qty(\mathcal{X}\to\mathcal{H})$ such that $\supp\qty(\Phi_1(x))\subseteq\supp\qty(\Phi_\mathrm{full}\qty(x))$ for all $x\in\mathcal{X}$,
    any CQ channel $\Phi_2^{(M)}\in\mathcal{F}\qty(\mathcal{X}^M\to\mathcal{H}^{\otimes M})$ such that $\supp\qty(\Phi^{\otimes M}(x^{(M)}))\subseteq\supp\qty(\Phi^{(M)}\qty(x^{(M)}))$ for all $x^{(M)}\in\mathcal{X}^M$,
    any sequence $\qty{p_n}_{n=1,2,\ldots}$ of probability distributions over $\mathcal{X}^n$,
    and any sequence $\qty{\qty{T_{x^{(n)}}}_{x^{(n)}\in\mathcal{X}^n}}_{n=1,2,\ldots}$ of families of elements of POVMs $\qty{T_{x^{(n)}},\mathds{1}-T_{x^{(n)}}}$ on $\mathcal{H}^{\otimes n}$ for input $x^{(n)}\in\mathcal{X}^n$,
    we choose $q_n$ and $r_n$ for every $n$ such that
    \begin{align}
    \label{eq:strong_converse_CQ_channel_copies_n}
        n&=q_n M+r_n,~0\leq r_n<M,
    \end{align}
    and define
    \begin{align}
        \Phi_2^{(n)}\coloneqq\Phi_2^{(M)\otimes q_n}\otimes\Phi_\mathrm{full}^{\otimes r_n}.
    \end{align}
    If we have
    \begin{align}
    \label{eq:strong_converse_CQ_channel_copies_limsup_assumption}
        \limsup_{n\to\infty}\sum_{x^{(n)}\in\mathcal{X}^n}p_n\qty(x^{(n)})\Tr\qty[\qty(\mathds{1}-T_{x^{(n)}})\Phi_1^{\otimes n}\qty(x^{(n)})]\leq\epsilon,
    \end{align}
    then it holds that
    \begin{align}
    \label{eq:strong_converse_CQ_channel_copies_limsup}
        &\limsup_{n\to\infty}-\frac{1}{n}\log\qty[\sum_{x^{(n)}\in\mathcal{X}^n}p_n\qty(x^{(n)})\Tr\qty[T_{x^{(n)}}\Phi_2^{(n)}\qty(x^{(n)})]]\notag\\
        &\leq\limsup_{n\to\infty}-\frac{1}{n}\log\qty[\beta_\epsilon\left(\Phi_1^{\otimes n}\middle\|\Phi_2^{\qty(n)}\right)]\notag\\
        &\leq\frac{1}{M}D\left(\Phi_1^{\otimes M}\middle\|\Phi_2^{(M)}\right),
    \end{align}
    where $\beta_\epsilon$ is defined as~\eqref{eq:beta_CQ_channels},
    and $D$ is defined as~\eqref{eq:D_channel}.
    If we have
    \begin{align}
    \label{eq:strong_converse_CQ_channel_copies_liminf_assumption}
        \liminf_{n\to\infty}\sum_{x^{(n)}\in\mathcal{X}^n}p_n\qty(x^{(n)})\Tr\qty[\qty(\mathds{1}-T_{x^{(n)}})\Phi_1^{\otimes n}\qty(x^{(n)})]\leq\epsilon,
    \end{align}
    then it holds that
    \begin{align}
    \label{eq:strong_converse_CQ_channel_copies_liminf}
        &\liminf_{n\to\infty}-\frac{1}{n}\log\qty[\sum_{x^{(n)}\in\mathcal{X}^n}p_n\qty(x^{(n)})\Tr\qty[T_{x^{(n)}}\Phi_2^{(n)}\qty(x^{(n)})]]\notag\\
        &\leq\frac{1}{M}D\left(\Phi_1^{\otimes M}\middle\|\Phi_2^{(M)}\right).
    \end{align}
\end{lemma}

\begin{proof}
    We will first prove~\eqref{eq:strong_converse_CQ_channel_copies_limsup}, and then~\eqref{eq:strong_converse_CQ_channel_copies_liminf}.

    \textbf{Proof of~\eqref{eq:strong_converse_CQ_channel_copies_limsup}.}
    Under the assumption~\eqref{eq:strong_converse_CQ_channel_copies_limsup_assumption}, there exists $n_0$ such that, for all $n\geq n_0$, we have
    \begin{align}
        \qty{T_{x^{\qty(n)}}}_{x^{\qty(n)}}\in\mathcal{T}_{\epsilon,\Phi_1^{\otimes n},p_{n}},
    \end{align}
    where $\mathcal{T}_{\epsilon,\Phi_1^{\otimes n},p_{n}}$ is defined as~\eqref{eq:T_set}.
    Then, for every $n\geq n_0$, due to Lemma~\ref{lem:bound_type_II_error_CQ_channels}, we have
    \begin{align}
        &-\frac{1}{n}\log\qty[\sum_{x^{\qty(n)}\in\mathcal{X}^n}p_n\qty(x^{\qty(n)})\Tr\qty[T_{x^{\qty(n)}}\Phi_2^{\qty(n)}\qty(x^{\qty(n)})]]\notag\\
        &\leq-\frac{1}{n}\log\qty[\beta_\epsilon\left(\Phi_1^{\otimes n}\middle\|\Phi_2^{\qty(n)}\right)]\\
        &\leq\frac{1}{n}\widetilde{D}_\alpha\left(\Phi_1^{\otimes n}\middle\|\Phi_2^{\qty(n)}\right)+\frac{\alpha}{n(\alpha-1)}\log\qty[\frac{1}{1-\epsilon}]\\
        \label{eq:strong_converse_CQ_channel_copies1_limsup}
        &=\frac{q_n}{n}\widetilde{D}_\alpha\left(\Phi_1^{\otimes M}\middle\|\Phi_2^{(M)}\right)+\frac{r_n}{n}\widetilde{D}_\alpha\left(\Phi_1\middle\|\Phi_\mathrm{full}\right)+\notag\\
        &\quad\frac{\alpha}{n(\alpha-1)}\log\qty[\frac{1}{1-\epsilon}],
    \end{align}
    where $\widetilde{D}_\alpha$ is defined as~\eqref{eq:D_alpha_channel}, and~\eqref{eq:strong_converse_CQ_channel_copies1_limsup} follows from the additivity in Lemma~\ref{lem:additivity_CQ_channels}.
    By taking the limit $n\to\infty$ with $q_{n}$ and $r_{n}$ in~\eqref{eq:strong_converse_CQ_channel_copies_n}, we have
    \begin{align}
        &\limsup_{n\to\infty}
        -\frac{1}{n}\log\qty[\sum_{x^{\qty(n)}\in\mathcal{X}^n}p_n\qty(x^{\qty(n)})\Tr\qty[T_{x^{\qty(n)}}\Phi_2^{\qty(n)}\qty(x^{\qty(n)})]]\notag\\
        &\leq\limsup_{n\to\infty}-\frac{1}{n}\log\qty[\beta_\epsilon\left(\Phi_1^{\otimes n}\middle\|\Phi_2^{\qty(n)}\right)]\\
        &\leq\frac{1}{M}\widetilde{D}_\alpha\left(\Phi_1^{\otimes M}\middle\|\Phi_2^{(M)}\right),
    \end{align}
    which holds for any $\alpha>1$.
    Taking the limit $\alpha\to\infty$ yields
    \begin{align}
        &\limsup_{n\to\infty}
        -\frac{1}{n}\log\qty[\sum_{x^{\qty(n)}\in\mathcal{X}^n}p_n\qty(x^{\qty(n)})\Tr\qty[T_{x^{\qty(n)}}\Phi_2^{\qty(n)}\qty(x^{\qty(n)})]]\notag\\
        &\leq\limsup_{n\to\infty}-\frac{1}{n}\log\qty[\beta_\epsilon\left(\Phi_1^{\otimes n}\middle\|\Phi_2^{\qty(n)}\right)]\\
        &\leq\frac{1}{M}D\left(\Phi_1^{\otimes M}\middle\|\Phi_2^{(M)}\right).
    \end{align}
    
    \textbf{Proof of~\eqref{eq:strong_converse_CQ_channel_copies_liminf}.}
    Under the assumption~\eqref{eq:strong_converse_CQ_channel_copies_liminf_assumption}, there exists a subsequence $\qty{n_l}_{l=1,2,\ldots}$ of $\{1,2,\ldots\}$ such that, for all $l$, we have
    \begin{align}
        \qty{T_{x^{\qty(n_l)}}}_{x^{\qty(n_l)}}\in\mathcal{T}_{\epsilon,\Phi_1^{\otimes n_l},p_{n_l}},
    \end{align}
    where $\mathcal{T}_{\epsilon,\Phi_1^{\otimes n_l},p_{n_l}}$ is defined as~\eqref{eq:T_set}.
    Then, for every $l$, due to Lemma~\ref{lem:bound_type_II_error_CQ_channels}, we have
    \begin{align}
        &-\frac{1}{n_l}\log\qty[\sum_{x^{\qty(n_l)}\in\mathcal{X}^{n_l}}p_{n_l}\qty(x^{\qty(n_l)})\Tr\qty[T_{x^{\qty(n_l)}}\Phi_2^{\qty(n_l)}\qty(x^{\qty(n_l)})]]\notag\\
        &\leq-\frac{1}{n_l}\log\qty[\beta_\epsilon\left(\Phi_1^{\otimes n_l}\middle\|\Phi_2^{\qty(n_l)}\right)]\\
        &\leq\frac{1}{n_l}\widetilde{D}_\alpha\left(\Phi_1^{\otimes n_l}\middle\|\Phi_2^{\qty(n_l)}\right)+\frac{\alpha}{n_l(\alpha-1)}\log\qty[\frac{1}{1-\epsilon}]\\
        \label{eq:strong_converse_CQ_channel_copies1_liminf}
        &=\frac{q_{n_l}}{n_l}\widetilde{D}_\alpha\left(\Phi_1^{\otimes M}\middle\|\Phi_2^{(M)}\right)+\frac{r_{n_l}}{n_l}\widetilde{D}_\alpha\left(\Phi_1\middle\|\Phi_\mathrm{full}\right)+\notag\\
        &\quad\frac{\alpha}{n_l(\alpha-1)}\log\qty[\frac{1}{1-\epsilon}],
    \end{align}
    where $\widetilde{D}_\alpha$ is defined as~\eqref{eq:D_alpha_channel}, and~\eqref{eq:strong_converse_CQ_channel_copies1_liminf} follows from the additivity in Lemma~\ref{lem:additivity_CQ_channels}.
    By taking the limit $l\to\infty$ with $n_l$, $q_{n_l}$, and $r_{n_l}$ in~\eqref{eq:strong_converse_CQ_channel_copies_n}, we have
    \begin{align}
        &\liminf_{n\to\infty}
        -\frac{1}{n}\log\qty[\sum_{x^{\qty(n)}\in\mathcal{X}^n}p_n\qty(x^{\qty(n)})\Tr\qty[T_{x^{\qty(n)}}\Phi_2^{\qty(n)}\qty(x^{\qty(n)})]]\notag\\
        &\leq\frac{1}{M}\widetilde{D}_\alpha\left(\Phi_1^{\otimes M}\middle\|\Phi_2^{(M)}\right),
    \end{align}
    which holds for any $\alpha>1$.
    Taking the limit $\alpha\to\infty$ yields
    \begin{align}
        &\liminf_{n\to\infty}
        -\frac{1}{n}\log\qty[\sum_{x^{\qty(n)}\in\mathcal{X}^n}p_n\qty(x^{\qty(n)})\Tr\qty[T_{x^{\qty(n)}}\Phi_2^{\qty(n)}\qty(x^{\qty(n)})]]\notag\\
        &\leq\frac{1}{M}D\left(\Phi_1^{\otimes M}\middle\|\Phi_2^{(M)}\right).
    \end{align}
\end{proof}

By combining all these techniques for CQ channels, we obtain the following lemma, which shows that a suboptimal sequence in minimizing $D$ can be improved by employing the optimal sequence in maximizing $\beta_\epsilon$. This result serves as a key ingredient in proving the direct part of the generalized quantum Stein's lemma.
While a closely related lemma appeared as a central step in the proof of the state version in Ref.~\cite{hayashi2025generalizedquantumsteinslemma}, our contribution is to extend this proof technique beyond the single-state setting to CQ channels with multiple possible inputs.
Again, further generalization to QQ channels remains an inherently challenging open problem, but our result firmly establishes a meaningful and tractable extension of this technique to CQ channels, thereby advancing the theory to analyze the generalized Stein's lemma from static to the fundamental class of dynamical resources.

\begin{lemma}[\label{lem:update}The update lemma for CQ channels]
    For any parameters $\epsilon\in(0,1)$, $\tilde{\epsilon}\in(0,\epsilon)$,
    any family $\mathcal{F}$ of sets of free CQ channels satisfying Axioms~\ref{p:full_rank},~\ref{p:compact},~\ref{p:tensor_product}, and~\ref{p:convex}, any CQ channel $\Phi\in\mathcal{C}\qty(\mathcal{X}\to\mathcal{H})$, and any sequence $\qty{\Phi_\mathrm{free}^{(n)}\in\mathcal{F}\qty(\mathcal{X}^n\to\mathcal{H}^{\otimes n})}_{n=1,2,\ldots}$ of free CQ channels,
    let $R_{1,\epsilon}$ and $R_2$ denote
    \begin{align}
        \label{eq:update_R_1}
        R_{1,\epsilon}&\coloneqq\liminf_{n\to\infty}-\frac{1}{n}\log\qty[\beta_{\epsilon}\left(\Phi^{\otimes n}\middle\|\mathcal{F}\right)],\\
        \label{eq:update_R_2}
        R_{2}&\coloneqq\liminf_{n\to\infty}\frac{1}{n}D\left(\Phi^{\otimes n}\middle\|\Phi_\mathrm{free}^{(n)}\right).
    \end{align}
    If it holds that
    \begin{align}
    \label{eq:R_2_greater_R_1}
        R_2> R_{1,\epsilon},
    \end{align}
    then there exists a sequence $\qty{\Phi_\mathrm{free}^{(n)\prime}\in\mathcal{F}\qty(\mathcal{X}^n\to\mathcal{H}^{\otimes n})}_{n=1,2,\ldots}$ of free CQ channels such that
    \begin{align}
    \label{eq:update_liminf}
        &\liminf_{n\to\infty}\frac{1}{n}D\left(\Phi^{\otimes n}\middle\|\Phi_\mathrm{free}^{(n)\prime}\right)-R_{1,\epsilon}\notag\\
        &\leq\qty(1-\tilde{\epsilon})\qty(R_2-R_{1,\epsilon}).
    \end{align}
\end{lemma}
\begin{proof}
    We first construct the updated sequence $\qty{\Phi_\mathrm{free}^{(n)\prime}}_{n=1,2,\ldots}$, followed by proving~\eqref{eq:update_liminf}.

    \textbf{Construction of $\qty{\Phi_\mathrm{free}^{(n)\prime}}_{n=1,2,\ldots}$.}
    Under the assumption~\eqref{eq:R_2_greater_R_1}, we fix a positive real parameter $\epsilon_0$ as
    \begin{align}
    \label{eq:epsilon_0}
        \epsilon_0\coloneqq\frac{\epsilon-\tilde{\epsilon}}{1-\epsilon}\qty(R_2-R_{1,\epsilon}).
    \end{align}
    With this $\epsilon_0$, due to~\eqref{eq:update_R_2}, there exists a sufficiently large integer $M$ and a free CQ channel $\Phi_\mathrm{free}^{(M)}\in\mathcal{F}\qty(\mathcal{X}^{M}\to\mathcal{H}^{\otimes M})$ such that
    \begin{align}
    \label{eq:Phi_free_n_2_definition_R_2}
        \frac{1}{M}D\left(\Phi^{\otimes M}\middle\|\Phi_\mathrm{free}^{(M)}\right)\leq R_2+\epsilon_0.
    \end{align}
    Axiom~\ref{p:full_rank} provides a free CQ channel
    \begin{align}
        \Phi_\mathrm{full}\in\mathcal{F}\qty(\mathcal{F}\to\mathcal{H})
    \end{align}
    satisfying the full-rank condition
    \begin{align}
    \label{eq:Lambda_min_update}
        \Phi_\mathrm{full}\qty(x)\geq\Lambda_{\min}\mathds{1},
    \end{align}
    where $\Lambda_{\min}\in(0,1]$ is defined in~\eqref{eq:Lambda_min}.
    For every $n$, we choose $q_n$ and $r_n$ such that
    \begin{align}
        n=q_n M+r_n,~0\leq r_n< M,
    \end{align}
    and, as in Lemma~\ref{lem:strong_converse_CQ_channel_copies}, define
    \begin{align}
    \label{eq:Phi_free_n_2_definition}
        \Phi_2^{(n)}\coloneqq\Phi_\mathrm{free}^{(M)\otimes q_n}\otimes\Phi_\mathrm{full}^{r_n}.
    \end{align}
    On the other hand, for every $n$, due to Axioms~\ref{p:compact} and~\ref{p:convex}, Proposition~\ref{prp:minimax} shows that
    \begin{align}
        \beta_{\epsilon}\left(\Phi^{\otimes n}\middle\|\mathcal{F}\right)=\max_{\Phi_\mathrm{free}\in\mathcal{F}}\beta_{\epsilon}\left(\Phi^{\otimes n}\middle\|\Phi_\mathrm{free}\right).
    \end{align}
    Let $\Phi_\mathrm{free}^{(n)\ast}\in\mathcal{F}\qty(\mathcal{X}^n\to\mathcal{H}^{\otimes n})$ denote a free CQ channel achieving this maximum, that is,
    \begin{align}
    \label{eq:Phi_free_n_ast_definition}
        \beta_{\epsilon}\left(\Phi^{\otimes n}\middle\|\Phi_\mathrm{free}^{(n)\ast}\right)=\max_{\Phi_\mathrm{free}\in\mathcal{F}}\beta_{\epsilon}\left(\Phi^{\otimes n}\middle\|\Phi_\mathrm{free}\right)=
        \beta_\epsilon\left(\Phi^{\otimes n}\middle\|\mathcal{F}\right).
    \end{align}
    Using these free CQ channels, we define
    \begin{align}
    \label{eq:Phi_free_n_prime_definition}
        \Phi_\mathrm{free}^{(n)\prime}\coloneqq\frac{\Phi_\mathrm{free}^{(n)\ast}+\Phi_2^{(n)}+\Phi_\mathrm{full}^{\otimes n}}{3}\in\mathcal{F},
    \end{align}
    where $\Phi_\mathrm{free}^{(n)\prime}$ is included in $\mathcal{F}$ due to Axiom~\ref{p:convex}.
    
    \textbf{Construction of approximations $\tilde{\Phi}^{(n)}$ and $\tilde{\Phi}_\mathrm{free}^{(n)\prime}$.}
    By the definition~\eqref{eq:Phi_free_n_prime_definition} of $\Phi_\mathrm{free}^{(n)\prime}$, we have, for every $x^{(n)}\in\mathcal{X}^n$,
    \begin{align}
    \label{eq:Phi_free_n_prime_operator_inequality_ast}
        \Phi_\mathrm{free}^{(n)\prime}\qty(x^{(n)})&\geq\frac{1}{3}\Phi_\mathrm{free}^{(n)\ast}\qty(x^{(n)}),\\
    \label{eq:Phi_free_n_prime_operator_inequality_2}
        \Phi_\mathrm{free}^{(n)\prime}\qty(x^{(n)})&\geq\frac{1}{3}\Phi_2^{(n)}\qty(x^{(n)}),\\
    \label{eq:Phi_free_n_prime_operator_inequality_full}
        \Phi_\mathrm{free}^{(n)\prime}\qty(x^{(n)})&\geq\frac{1}{3}\Phi_\mathrm{full}^{\otimes n}\qty(x^{(n)}).
    \end{align}
    With $\Lambda_{\min}$ in~\eqref{eq:Lambda_min_update}, we set
    \begin{align}
        C_n\coloneqq\log\qty[\frac{1}{\Lambda_{\min}}]+\frac{\log\qty[3]}{n},
    \end{align}
    so that we obtain from~\eqref{eq:Phi_free_n_prime_operator_inequality_full}
    \begin{align}
    \label{eq:Phi_free_n_prime_operator_inequality_full2}
        \Phi_\mathrm{free}^{(n)\prime}\qty(x^{(n)})&\geq\frac{\Phi_\mathrm{full}^{\otimes n}\qty(x^{(n)})}{3}\geq\frac{\Lambda_{\min}^n}{3}\mathds{1}=e^{-C_n n}\mathds{1}.
    \end{align}
    Then, Lemma~\ref{lem:approximation} yields CQ channels
    \begin{align}
        \tilde{\Phi}^{(n)},\tilde{\Phi}_\mathrm{free}^{(n)\prime}\in\mathcal{C}\qty(\mathcal{X}^n\to\mathcal{H}^{\otimes n})
    \end{align}
    such that
    \begin{align}
        \label{eq:update_commute}
        &\tilde{\Phi}^{(n)}\qty(x^{(n)})
        \tilde{\Phi}_\mathrm{free}^{(n)\prime}\qty(x^{(n)})=
        \tilde{\Phi}_\mathrm{free}^{(n)\prime}\qty(x^{(n)})
        \tilde{\Phi}^{(n)}\qty(x^{(n)}),\\
        \label{eq:update_operator_inequality}
        &e^{-C_n}\Phi_\mathrm{free}^{(n)\prime}\qty(x^{(n)})\leq\tilde{\Phi}_\mathrm{free}^{(n)\prime}\qty(x^{(n)})\leq e^{C_n}\Phi_\mathrm{free}^{(n)\prime}\qty(x^{(n)}),\\
        \label{eq:update_beta_liminf}
        &\liminf_{n\to\infty}-\frac{1}{n}\log\qty[\beta_\epsilon\left(\tilde{\Phi}^{(n)}\middle\|\tilde{\Phi}_\mathrm{free}^{(n)\prime}\right)]\notag\\
        &\leq\liminf_{n\to\infty}-\frac{1}{n}\log\qty[\beta_\epsilon\left(\Phi^{\otimes n}\middle\|\Phi_\mathrm{free}^{(n)\prime}\right)],\\
        \label{eq:update_beta_limsup}
        &\limsup_{n\to\infty}-\frac{1}{n}\log\qty[\beta_\epsilon\left(\tilde{\Phi}^{(n)}\middle\|\tilde{\Phi}_\mathrm{free}^{(n)\prime}\right)]\notag\\
        &\leq\limsup_{n\to\infty}-\frac{1}{n}\log\qty[\beta_\epsilon\left(\Phi^{\otimes n}\middle\|\Phi_\mathrm{free}^{(n)\prime}\right)],\\
        \label{eq:update_D}
        &\liminf_{n\to\infty}\frac{1}{n}D\left(\tilde{\Phi}^{(n)}\middle\|\tilde{\Phi}_\mathrm{free}^{(n)\prime}\right)=\liminf_{n\to\infty}\frac{1}{n}D\left(\Phi^{\otimes n}\middle\|\Phi_\mathrm{free}^{(n)\prime}\right).
    \end{align}
    
    \textbf{Definition of projections for the information spectrum method.}
    Due to~\eqref{eq:update_beta_liminf}, it holds that
    \begin{align}
        &\liminf_{n\to\infty}-\frac{1}{n}\log\qty[\beta_\epsilon\left(\tilde{\Phi}^{(n)}\middle\|\tilde{\Phi}_\mathrm{free}^{(n)\prime}\right)]\notag\\
        &\leq\liminf_{n\to\infty}-\frac{1}{n}\log\qty[\beta_\epsilon\left(\Phi^{\otimes n}\middle\|\Phi_\mathrm{free}^{(n)\prime}\right)],\\
        \label{eq:update_beta_liminf1}
        &\leq\liminf_{n\to\infty}-\frac{1}{n}\log\qty[\beta_\epsilon\left(\Phi^{\otimes n}\middle\|\Phi_\mathrm{free}^{(n)\ast}\right)],\\
        \label{eq:update_beta_liminf2}
        &=R_{1,\epsilon},
    \end{align}
    where~\eqref{eq:update_beta_liminf1} follows from~\eqref{eq:Phi_free_n_prime_operator_inequality_ast} due to Lemma~\ref{lem:bound_on_beta_from_operator_inequality}, and~\eqref{eq:update_beta_liminf2} is due to~\eqref{eq:update_R_1} and~\eqref{eq:Phi_free_n_ast_definition}.
    Due to~\eqref{eq:update_beta_limsup}, for any
    \begin{align}
    \label{eq:epsilon_1}
        \epsilon_1\in(0,1],
    \end{align}
    it holds that
    \begin{align}
        &\limsup_{n\to\infty}-\frac{1}{n}\log\qty[\beta_{1-\epsilon_1}\left(\tilde{\Phi}^{(n)}\middle\|\tilde{\Phi}_\mathrm{free}^{(n)\prime}\right)]\notag\\
        &\leq\limsup_{n\to\infty}-\frac{1}{n}\log\qty[\beta_{1-\epsilon_1}\left(\Phi^{\otimes n}\middle\|\Phi_\mathrm{free}^{(n)\prime}\right)],\\
        \label{eq:update_beta_limsup1}
        &\leq\limsup_{n\to\infty}-\frac{1}{n}\log\qty[\beta_{1-\epsilon_1}\left(\Phi^{\otimes n}\middle\|\Phi_2^{(n)}\right)],\\
        \label{eq:update_beta_limsup2}
        &\leq R_{2}+\epsilon_0,
    \end{align}
    where~\eqref{eq:update_beta_limsup1} follows from~\eqref{eq:Phi_free_n_prime_operator_inequality_2} due to Lemma~\ref{lem:bound_on_beta_from_operator_inequality}, and~\eqref{eq:update_beta_limsup2} is obtained from~\eqref{eq:Phi_free_n_2_definition_R_2} and~\eqref{eq:Phi_free_n_2_definition} due to Lemma~\ref{lem:strong_converse_CQ_channel_copies}.

    For any
    \begin{align}
    \label{eq:epsilon_2}
        \epsilon_2>0
    \end{align}
    and any sequence $\qty{p_n}_{n=1,2,\ldots}$ probability distributions over $\mathcal{X}^n$, due to~\eqref{eq:update_beta_liminf2} and~\eqref{eq:update_beta_limsup2}, Lemma~\ref{lem:information_spectrum} shows that the families $\qty{T_{x^{(n)},1}}_{x^{(n)}\in\mathcal{X}^n}$ and $\qty{T_{x^{(n)},2}}_{x^{(n)}\in\mathcal{X}^n}$ of projections given by
    \begin{align}
        &T_{x^{(n)},1}\coloneqq\qty{\tilde{\Phi}^{(n)}\qty(x^{(n)})\geq e^{\qty(R_{1,\epsilon}+\epsilon_2)n}\tilde{\Phi}_\mathrm{free}^{(n)\prime}\qty(x^{(n)})},\\
        &T_{x^{(n)},2}\coloneqq\qty{\tilde{\Phi}^{(n)}\qty(x^{(n)})\geq e^{\qty(R_2+\epsilon_0+\epsilon_2)n}\tilde{\Phi}_\mathrm{free}^{(n)\prime}\qty(x^{(n)})}
    \end{align}
    satisfy
    \begin{align}
        &\limsup_{n\to\infty}\sum_{x^{(n)}\in\mathcal{X}^n}p_n\qty(x^{(n)})\Tr\qty[\qty(\mathds{1}-T_{x^{(n)},1})\tilde{\Phi}^{(n)}\qty(x^{(n)})]\notag\\
        &>\epsilon,\\
        &\liminf_{n\to\infty}\sum_{x^{(n)}\in\mathcal{X}^n}p_n\qty(x^{(n)})\Tr\qty[\qty(\mathds{1}-T_{x^{(n)},2})\tilde{\Phi}^{(n)}\qty(x^{(n)})]\notag\\
        &>1-\epsilon_1.
    \end{align}
    That is, for an arbitrary sequence $\qty{p_n}_n$ of probability distributions, these families of projections satisfy
    \begin{align}
    \label{eq:update_T_liminf}
        &\liminf_{n\to\infty}\sum_{x^{(n)}\in\mathcal{X}^n}p_n\qty(x^{(n)})\Tr\qty[T_{x^{(n)},1}\tilde{\Phi}^{(n)}\qty(x^{(n)})]\notag\\
        &<1-\epsilon,\\
    \label{eq:update_T_limsup}
        &\limsup_{n\to\infty}\sum_{x^{(n)}\in\mathcal{X}^n}p_n\qty(x^{(n)})\Tr\qty[T_{x^{(n)},2}\tilde{\Phi}^{(n)}\qty(x^{(n)})]\notag\\
        &<\epsilon_1.
    \end{align}
    Then, we define
    \begin{align}
    \label{eq:update_pi_1}
        \Pi_{x^{(n)},1}&\coloneqq \mathds{1}-T_{x^{(n)},1},\\
    \label{eq:update_pi_2}
        \Pi_{x^{(n)},2}&\coloneqq T_{x^{(n)},1}-T_{x^{(n)},2},\\
    \label{eq:update_pi_3}
        \Pi_{x^{(n)},3}&\coloneqq T_{x^{(n)},2},
    \end{align}
    which satisfy
    \begin{align}
    \label{eq:update_pi_all}
        \Pi_{x^{(n)},1}+\Pi_{x^{(n)},2}+\Pi_{x^{(n)},3}=\mathds{1}.
    \end{align}
    
    \textbf{Derivation of operator inequalities using the information spectrum method.}
    Due to~\eqref{eq:update_commute}, the operators $\tilde{\Phi}^{(n)}\qty(x^{(n)})$, $\tilde{\Phi}_\mathrm{free}^{(n)\prime}\qty(x^{(n)})$, $T_{x^{(n)},1}$, and $T_{x^{(n)},2}$ commute with each other; hence, by the definitions~\eqref{eq:update_pi_1} and~\eqref{eq:update_pi_2} of $\Pi_{x^{(n)},1}$ and $\Pi_{x^{(n)},2}$ using these mutually commuting operators, for every $x^{(n)}\in\mathcal{X}^n$, we have operator inequalities
    \begin{align}
        \label{eq:update_pi_1_operator_inequality}
        &\frac{1}{n}\Pi_{x^{(n)},1}\qty(\log\qty[\tilde{\Phi}^{(n)}\qty(x^{(n)})]-\log\qty[\tilde{\Phi}_\mathrm{free}^{(n)\prime}\qty(x^{(n)})])\notag\\
        &\leq\qty(R_{1,\epsilon}+\epsilon_2)\Pi_{x^{(n)},1},\\
        \label{eq:update_pi_2_operator_inequality}
        &\frac{1}{n}\Pi_{x^{(n)},2}\qty(\log\qty[\tilde{\Phi}^{(n)}\qty(x^{(n)})]-\log\qty[\tilde{\Phi}_\mathrm{free}^{(n)\prime}\qty(x^{(n)})])\notag\\
        &\leq\qty(R_2+\epsilon_0+\epsilon_2)\Pi_{x^{(n)},2}.
    \end{align}
    To derive an operator inequality for $\Pi_{x^{(n)},2}$ in~\eqref{eq:update_pi_3},
    we write
    \begin{align}
        C_n'\coloneqq\max\qty{C_n+\frac{C_n}{n},R_2+\epsilon_0+\epsilon_2},
    \end{align}
    so that~\eqref{eq:Phi_free_n_prime_operator_inequality_full2} and~\eqref{eq:update_operator_inequality} yield, for every $x^{(n)}\in\mathcal{X}^n$,
    \begin{align}
        &\tilde{\Phi}_\mathrm{free}^{(n)\prime}\qty(x^{(n)})\notag\\
        &\geq e^{-C_n}\Phi_\mathrm{free}^{(n)\prime}\qty(x^{(n)})\\
        &\geq e^{-\qty(C_n+\frac{C_n}{n})n}\mathds{1}\\
        &\geq e^{-C_n'n}\mathds{1}.
    \end{align}
    Thus, we have
    \begin{align}
        &\frac{1}{n}\Pi_{x^{(n)},3}\qty(\log\qty[\tilde{\Phi}^{(n)}\qty(x^{(n)})]-\log\qty[\tilde{\Phi}_\mathrm{free}^{(n)\prime}\qty(x^{(n)})])\notag\\
        &\leq\frac{1}{n}\Pi_{x^{(n)},3}\qty(-\log\qty[\tilde{\Phi}_\mathrm{free}^{(n)\prime}\qty(x^{(n)})])\\
        \label{eq:update_pi_3_operator_inequality}
        &\leq C_n'\Pi_3^{(n)}
    \end{align}
    
    \textbf{Evaluation of quantum relative entropy using the operator inequalities.}
    Using these projections $\Pi_{x^{(n)},1}$, $\Pi_{x^{(n)},2}$, and $\Pi_{x^{(n)},3}$, we have, for every $x^{(n)}\in\mathcal{X}^n$,
    \begin{widetext}
    \begin{align}
        &\frac{1}{n}D\left(\tilde{\Phi}^{(n)}\qty(x^{(n)})\middle\|\tilde{\Phi}_\mathrm{free}^{(n)\prime}\qty(x^{(n)})\right)\notag\\
        &=\frac{1}{n}\Tr\qty[\tilde{\Phi}^{(n)}\qty(x^{(n)})\qty(\log\qty[\tilde{\Phi}^{(n)}\qty(x^{(n)})]-\log\qty[\tilde{\Phi}_\mathrm{free}^{(n)\prime}\qty(x^{(n)})])]\\
        \label{eq:update_relative_entropy_1}
        &=\frac{1}{n}\Tr\qty[\tilde{\Phi}^{(n)}\qty(x^{(n)})\Pi_{x^{(n)},1}\qty(\log\qty[\tilde{\Phi}^{(n)}\qty(x^{(n)})]-\log\qty[\tilde{\Phi}_\mathrm{free}^{(n)\prime}\qty(x^{(n)})])]+\notag\\
        &\quad\frac{1}{n}\Tr\qty[\tilde{\Phi}^{(n)}\qty(x^{(n)})\Pi_{x^{(n)},2}\qty(\log\qty[\tilde{\Phi}^{(n)}\qty(x^{(n)})]-\log\qty[\tilde{\Phi}_\mathrm{free}^{(n)\prime}\qty(x^{(n)})])]+\notag\\
        &\quad\frac{1}{n}\Tr\qty[\tilde{\Phi}^{(n)}\qty(x^{(n)})\Pi_{x^{(n)},3}\qty(\log\qty[\tilde{\Phi}^{(n)}\qty(x^{(n)})]-\log\qty[\tilde{\Phi}_\mathrm{free}^{(n)\prime}\qty(x^{(n)})])]\\
        \label{eq:update_relative_entropy_2}
        &\leq\qty(R_{1,\epsilon}+\epsilon_2)\Tr\qty[\Pi_{x^{(n)},1}\tilde{\Phi}^{(n)}\qty(x^{(n)})]+\qty(R_2+\epsilon_0+\epsilon_2)+\Tr\qty[\Pi_{x^{(n)},2}\tilde{\Phi}^{(n)}\qty(x^{(n)})]+C_n'\Tr\qty[\Pi_{x^{(n)},3}\tilde{\Phi}^{(n)}\qty(x^{(n)})]\\
        \label{eq:update_relative_entropy_3}
        &=\qty(R_{1,\epsilon}+\epsilon_2)+\qty(R_2-R_{1,\epsilon}+\epsilon_0)\Tr\qty[T_{x^{(n)},1}\tilde{\Phi}^{(n)}\qty(x^{(n)})]+\qty(C_n'-\qty(R_2+\epsilon_0+\epsilon_2))\Tr\qty[T_{x^{(n)},2}\tilde{\Phi}^{(n)}\qty(x^{(n)})],
    \end{align}
    where~\eqref{eq:update_relative_entropy_1} follows from~\eqref{eq:update_pi_all}, we have~\eqref{eq:update_relative_entropy_2} due to~\eqref{eq:update_pi_1_operator_inequality},~\eqref{eq:update_pi_2_operator_inequality}, and~\eqref{eq:update_pi_3_operator_inequality}, and~\eqref{eq:update_relative_entropy_3} is obtained from~\eqref{eq:update_pi_1},~\eqref{eq:update_pi_2}, and~\eqref{eq:update_pi_3}.
    By taking the maximum over all possible input probability distributions, we obtain
    \begin{align}
        &\frac{1}{n}D\left(\tilde{\Phi}^{(n)}\middle\|\tilde{\Phi}_\mathrm{free}^{(n)\prime}\right)\notag\\
        &\leq\max_{p_n}\left\{\qty(R_{1,\epsilon}+\epsilon_2)+\right.\notag\\
        &\quad \qty(R_2-R_{1,\epsilon}+\epsilon_0) \sum_{x^{(n)}\in\mathcal{X}^n}p_n\qty(x^{(n)})\Tr\qty[T_{x^{(n)},1}\tilde{\Phi}^{(n)}\qty(x^{(n)})]+\notag\\
        &\quad\left.\qty(C_n'-\qty(R_2+\epsilon_0+
        \epsilon_2))\sum_{x^{(n)}\in\mathcal{X}^n}p_n\qty(x^{(n)})\Tr\qty[T_{x^{(n)},2}\tilde{\Phi}^{(n)}\qty(x^{(n)})]\right\},
    \end{align}
    \end{widetext}
    which holds for all $n$.
    By taking the limit $n\to\infty$, due to~\eqref{eq:update_T_liminf} and~\eqref{eq:update_T_limsup}, we have
    \begin{align}
        &\liminf_{n\to\infty}\frac{1}{n}D\left(\tilde{\Phi}^{(n)}\middle\|\tilde{\Phi}_\mathrm{free}^{(n)\prime}\right)\notag\\
        &\leq\qty(R_{1,\epsilon}+\epsilon_2)+\notag\\
        &\quad \qty(R_2-R_{1,\epsilon}+\epsilon_0) \qty(1-\epsilon)+\notag\\
        &\quad\qty(C_n'-\qty(R_2+\epsilon_0+
        \epsilon_2))\epsilon_1,
    \end{align}
    which holds for any choices of $\epsilon_1$ in~\eqref{eq:epsilon_1} and $\epsilon_2$ in~\eqref{eq:epsilon_2}.
    By taking the limit $\epsilon_1,\epsilon_2\to 0$, it holds that
    \begin{align}
        &\liminf_{n\to\infty}\frac{1}{n}D\left(\tilde{\Phi}^{(n)}\middle\|\tilde{\Phi}_\mathrm{free}^{(n)\prime}\right)\notag\\
        &\leq R_{1,\epsilon}+\qty(R_2-R_{1,\epsilon}+\epsilon_0) \qty(1-\epsilon)\\
        &= R_{1,\epsilon}+\qty(1-\tilde{\epsilon})\qty(R_2-R_{1,\epsilon}),
    \end{align}
    where the last line follows from the definition~\eqref{eq:epsilon_0} of $\epsilon_0$.
    Consequently, due to~\eqref{eq:update_D}, we obtain
    \begin{align}
        &\liminf_{n\to\infty}\frac{1}{n}D\left(\Phi^{(n)}\middle\|\Phi_\mathrm{free}^{(n)\prime}\right)\notag\\
        &=\liminf_{n\to\infty}\frac{1}{n}D\left(\tilde{\Phi}^{(n)}\middle\|\tilde{\Phi}_\mathrm{free}^{(n)\prime}\right)\notag\\
        &\leq R_{1,\epsilon}+\qty(R_2-R_{1,\epsilon})\qty(1-\tilde{\epsilon}),
    \end{align}
    which shows the conclusion.
\end{proof}

\paragraph{Proof of the direct part of the generalized quantum Stein's lemma for CQ channels}
\label{sec:proof_direct}

By employing the proof techniques for CQ channels established above, we now complete the direct part of the generalized quantum Stein's lemma.
This proof is enabled by extending the toolkit for analyzing the generalized Stein's lemma from static resources of quantum states to the fundamental class of dynamical resources represented by CQ channels.
As highlighted throughout, a full extension to QQ channels may be inherently challenging, but our contribution lies in making this extension feasible for CQ channels by addressing the challenge of handling multiple possible inputs.

\begin{proposition}[\label{prp:direct}The direct part of the generalized quantum Stein's lemma for CQ channels]
    For any parameter $\epsilon\in(0,1)$, 
    any family $\mathcal{F}$ of sets of free CQ channels satisfying Axioms~\ref{p:full_rank},~\ref{p:compact},~\ref{p:tensor_product}, and~\ref{p:convex},
    and any CQ channel $\Phi\in\mathcal{C}\qty(\mathcal{X}\to\mathcal{H})$,
    it holds that
    \begin{align}
        &\liminf_{n\to\infty}-\frac{1}{n}\log\qty[\beta_\epsilon\left(\Phi^{\otimes n}\middle\|\mathcal{F}\right)]\geq\lim_{n\to\infty}\frac{1}{n}D\left(\Phi^{\otimes n}\middle\|\mathcal{F}\right),
    \end{align}
    where $\beta_\epsilon$ is defined as~\eqref{eq:beta}, and $D$ is defined as~\eqref{eq:D_F}.
\end{proposition}

\begin{proof}
    We provide proof by contradiction.
    We write
    \begin{align}
        R_{1,\epsilon}&\coloneqq\liminf_{n\to\infty}-\frac{1}{n}\log\qty[\beta_\epsilon\left(\Phi^{\otimes n}\middle\|\mathcal{F}\right)],\\
        R_2&\coloneqq\lim_{n\to\infty}\frac{1}{n}D\left(\Phi^{\otimes n}\middle\|\mathcal{F}\right).
    \end{align}
    Let $\qty{\Phi_\mathrm{free}^{(n)}\in\mathcal{C}\qty(\mathcal{X}^n\to\mathcal{H}^{\otimes n})}_{n=1,2,\ldots}$ be an optimal sequence of free CQ channels achieving the minimum in the definition~\eqref{eq:D_F} of $D\left(\Phi^{\otimes n}\middle\|\mathcal{F}\right)$
    \begin{align}
        D\left(\Phi^{\otimes n}\middle\|\Phi_\mathrm{free}^{(n)}\right)=D\left(\Phi^{\otimes n}\middle\|\mathcal{F}\right)=\min_{\Phi_\mathrm{free}\in\mathcal{F}}D\left(\Phi^{\otimes n}\middle\|\Phi_\mathrm{free}\right),
    \end{align}
    and to derive a contradiction, suppose that
    \begin{align}
        R_{1,\epsilon}<R_2.
    \end{align}
    Then, for any $\tilde{\epsilon}\in(0,\epsilon)$, Lemma~\ref{lem:update} provides an updated sequence $\qty{\Phi_\mathrm{free}^{(n)\prime}\in\mathcal{C}\qty(\mathcal{X}^n\to\mathcal{H}^{\otimes n})}_{n=1,2,\ldots}$ achieving
    \begin{align}
        \liminf_{n\to\infty}D\left(\Phi^{\otimes n}\middle\|\Phi_\mathrm{free}^{(n)\prime}\right)&\leq\qty(1-\tilde{\epsilon})R_2+\tilde{\epsilon}R_{1,\epsilon}<R_2,
    \end{align}
    which contradicts the optimality of the choice of $\qty{\Phi_\mathrm{free}^{(n)}}_n$.
\end{proof}

\section{Analysis of reversible QRT framework for CQ channel conversion}
\label{sec:reversible}

In this section, we formulate and analyze a reversible QRT framework for CQ channel conversion, based on the generalized quantum Stein's lemma for CQ channels.
In Sec.~\ref{sec:formulation_qrt}, we introduce the formulation of this framework, where the conversion rate between resource CQ channels is determined by a single quantity: the regularized relative entropy of resource.
For its analysis, Sec.~\ref{sec:relative_entropy_robustness} provides a characterization of the regularized relative entropy of resource in terms of the logarithmic generalized robustness for CQ channels.
Building on this characterization, Sec.~\ref{sec:conversion_rate} analyses the conversion rate in this framework.

\subsection{Formulation of reversible QRT framework for CQ channel conversion}
\label{sec:formulation_qrt}

In this section, we formulate a reversible QRT framework for CQ channel conversion in such a way that it has a smaller set of assumptions than the previous work~\cite{hayashi2025generalizedquantumsteinslemma}.
As discussed in Sec.~\ref{sec:QRTs_for_CQ_channels}, QRTs are specified by a family $\mathcal{O}$ of free operations in~\eqref{eq:O} as a subset of superchannels converting CQ channels to CQ channels.
The choice of $\mathcal{O}$ leads to the family $\mathcal{F}$ of free CQ channels as in~\eqref{eq:F}.
However, to make QRTs reversible, it is generally insufficient to consider $\mathcal{O}$, but we may need to consider its relaxation.
An essential feature of free operations $\mathcal{O}$ is that the free operations should not generate resource states from free states; however, in the context of asymptotic conversion, it is possible to axiomatically define a relaxed class of operations, $\tilde{O}$, which captures this requirement only in an asymptotic sense.

To introduce an appropriate relaxation $\tilde{\mathcal{O}}$, in analogy to the state case~\cite{Brand_o_2008,brandao2010reversible,Brandao2010,Brandao2015}, we define a family
\begin{widetext}
\begin{align}
    \label{eq:tilde_O}
&\tilde{\mathcal{O}}\Big(\qty(\mathcal{X}_\mathrm{in}\to\mathcal{H}_\mathrm{in})\to \qty(\mathcal{X}_\mathrm{out}\to\mathcal{H}_\mathrm{out})\Big)\notag\\
&\coloneqq
\Big\{\qty{\Theta^{(n)}\in\mathcal{C}\qty(\qty(\mathcal{X}_\mathrm{in}^{n}\to\mathcal{H}_\mathrm{in}^{\otimes n})\to\qty(\mathcal{X}_\mathrm{out}^{f(n)}\to\mathcal{H}_\mathrm{out}^{\otimes f(n)}))}_{n=1,2,\ldots}:\text{$\qty{\Theta^{(n)}}_{n}$ satisfies the axiom shown below in~\ref{p:arng}}\Big\}
\end{align}
\end{widetext}
of sequences of operations, i.e., superchannels mapping $n$-fold CQ channels to $f(n)$-fold CQ channels as in~\eqref{eq:C}, satisfying the following property, where $f(n)$ is an increasing function of $n$.
\begin{enumerate}[label={SC\arabic*}]
    \item\label{p:arng}\textbf{Asymptotically resource-non-generating property}:  
    For any sequence $\qty{\Theta^{(n)}}_{n=1,2,\ldots} \in \tilde{\mathcal{O}}\qty(\qty(\mathcal{X}_\mathrm{in}\to\mathcal{H}_\mathrm{in})\to\qty(\mathcal{X}_\mathrm{out}\to\mathcal{H}_\mathrm{out}))$ of superchannels and any sequence $\qty{\Phi_\mathrm{free}^{(n)}\in\mathcal{F}\qty(\mathcal{X}_\mathrm{in}^n\to\mathcal{H}_\mathrm{in}^{\otimes n})}_{n=1,2,\ldots}$ of free CQ channels, the sequence $\{\Theta^{(n)}\}_{n}$ is asymptotically resource-non-generating in terms of the generalized robustness, i.e.,
    \begin{align}
        \lim_{n\to\infty} R_\mathrm{G}\qty(\Theta^{(n)}\qty[\Phi_\mathrm{free}^{(n)}])=0,
    \end{align}
    where $R_\mathrm{G}$ is defined as~\eqref{eq:R_G}.
\end{enumerate}
We may write this family as $\tilde{\mathcal{O}}$ if the argument is obvious from the context.

As in the conversion rate~\eqref{eq:conversion_rate_O} under $\mathcal{O}$, the asymptotic conversion rate under the relaxed class $\tilde{\mathcal{O}}$ of operations, from a CQ channel $\Phi_\mathrm{in}$ to a CQ channel $\Phi_\mathrm{out}$, is defined as
\begin{align}
\label{eq:conversion_rate}
    &r_{\tilde{\mathcal{O}}}\qty(\Phi_\mathrm{in}\to\Phi_\mathrm{out})\coloneqq\sup\left\{r\geq 0:\right.\notag\\
    &\quad\left.\exists\qty{\Theta^{(n)}}_n\in\tilde{\mathcal{O}},~\liminf_{n\to\infty}d_\diamond\qty(\Theta^{(n)}\qty[\Phi_\mathrm{in}^{\otimes n}],\Phi_\mathrm{out}^{\otimes \lceil rn\rceil})=0\right\}.
\end{align}
The main result in this section is that, under $\tilde{\mathcal{O}}$, QRTs are shown to be reversible as in the following theorem, which is considered the second law of QRTs as originally proposed in Ref.~\cite{Brand_o_2008,brandao2010reversible,Brandao2010,Brandao2015,hayashi2025generalizedquantumsteinslemma}.

\begin{theorem}[\label{thm:second_law}The reversible QRT framework for CQ channel conversion]
For any family $\mathcal{F}$ of sets of free CQ channels satisfying Axioms~\ref{p:full_rank},~\ref{p:compact},~\ref{p:tensor_product}, and~\ref{p:convex},
any family $\tilde{\mathcal{O}}$ of sequences of superchannels satisfying Axiom~\ref{p:arng},
and any CQ channels $\Phi_\mathrm{in}\in\mathcal{C}\qty(\mathcal{X}_\mathrm{in}\to\mathcal{H}_\mathrm{in})$ and $\Phi_\mathrm{out}\in\mathcal{C}\qty(\mathcal{X}_\mathrm{out}\to\mathcal{H}_\mathrm{out})$ satisfying
\begin{align}
\label{eq:direct_second_law_assumption1}
    R_\mathrm{R}^\infty\qty(\Phi_\mathrm{in})&>0,\\
\label{eq:direct_second_law_assumption2}
    R_\mathrm{R}^\infty\qty(\Phi_\mathrm{out})&>0,
\end{align}
it holds that
\begin{align}
r_{\tilde{\mathcal{O}}}\qty(\Phi_{\mathrm{in}}\to\Phi_{\mathrm{out}})=\frac{R_\mathrm{R}^\infty\qty(\Phi_\mathrm{in})}{R_\mathrm{R}^\infty\qty(\Phi_\mathrm{out})}.
\end{align}
where $R_\mathrm{R}^\infty$ is defined as~\eqref{eq:regularized_relative_entropy_of_resource}, and $r_{\tilde{\mathcal{O}}}$ is defined as~\eqref{eq:conversion_rate}.
\end{theorem}

In the remainder of this section, we will prove this theorem.
In Sec.~\ref{sec:relative_entropy_robustness}, we will provide a characterization of the regularized relative entropy of resource in logarithmic generalized robustness for CQ channels.
Then, in Sec.~\ref{sec:conversion_rate}, using this characterization, we will provide techniques to prove the theorem.
Together with these results, our proof is summarized as follows.

\begin{proof}[Proof of Theorem~\ref{thm:second_law}]
    In Sec.~\ref{sec:converse_second_law}, we prove Proposition~\ref{prp:converse_second_law}, which shows
    \begin{align}
    r_{\tilde{\mathcal{O}}}\qty(\Phi_{\mathrm{in}}\to\Phi_{\mathrm{out}}) \leq\frac{R_\mathrm{R}^\infty\qty(\Phi_\mathrm{in})}{R_\mathrm{R}^\infty\qty(\Phi_\mathrm{out})}.
    \end{align}
    In Sec.~\ref{sec:direct_second_law}, we prove Proposition~\ref{prp:direct_second_law}, which shows
    \begin{align}
    r_{\tilde{\mathcal{O}}}\qty(\Phi_{\mathrm{in}}\to\Phi_{\mathrm{out}})\geq\frac{R_\mathrm{R}^\infty\qty(\Phi_\mathrm{in})}{R_\mathrm{R}^\infty\qty(\Phi_\mathrm{out})}.
    \end{align}
    Altogether, we obtain the conclusion.
\end{proof}

We remark that our essential contribution beyond the previous work~\cite{hayashi2025generalizedquantumsteinslemma} lies in eliminating the need to impose an additional ``asymptotic continuity'' requirement when introducing the relaxed class $\tilde{\mathcal{O}}$ of operations, thereby broadening the applicability of the framework.  
In particular, Ref.~\cite{hayashi2025generalizedquantumsteinslemma} also sought to construct a reversible QRT framework for CQ channel conversion by introducing another relaxation $\tilde{\mathcal{O}}'$, which requires Axiom~\ref{p:arng} but in addition imposes an asymptotic continuity condition: for any two sequences $\qty{\Phi_1^{(n)}}_n$ and $\qty{\Phi_2^{(n)}}_n$ of CQ channels satisfying
\begin{align}
    \lim_{n\to\infty}\frac{1}{X_\mathrm{in}^n}\sum_{x_\mathrm{in}\in\mathcal{X}_\mathrm{in}^n}
    d_{\mathrm{T}}\left(\Phi_1^{(n)}(x_\mathrm{in}),\Phi_2^{(n)}(x_\mathrm{in})\right)=0,
\end{align}
any sequence $\qty{\Theta^{(n)}}_n\in\tilde{\mathcal{O}}'$ is required to satisfy
\begin{widetext}
\begin{align}
    &\lim_{n\to\infty}\frac{1}{X_\mathrm{out}^n}\sum_{x_\mathrm{out}\in\mathcal{X}_\mathrm{out}^n}
    d_{\mathrm{T}}\left(\qty(\Theta^{(n)}\qty[\Phi_1^{(n)}])(x_\mathrm{out}),\qty(\Theta^{(n)}\qty[\Phi_2^{(n)}])(x_\mathrm{out})\right)=0,
\end{align}
\end{widetext}
where $X_\mathrm{in}$ and $X_\mathrm{out}$ are given by~\eqref{eq:dim}, and $d_{\mathrm{T}}$ denotes the trace distance defined in~\eqref{eq:d_trace}.  
However, this condition is restrictive since a general superchannel does not necessarily satisfy asymptotic continuity.  
By contrast, in our definition of $\tilde{\mathcal{O}}$, we dispense with this requirement and instead consider only the asymptotically resource-non-generating property in Axiom~\ref{p:arng}, fully in line with the state-based frameworks of Refs.~\cite{Brand_o_2008,brandao2010reversible,Brandao2010,Brandao2015}.  
Consequently, our framework encompasses a strictly more general class of operations, which subsumes the operations in the previous reversible framework~\cite{hayashi2025generalizedquantumsteinslemma} as a special case, i.e., $\tilde{\mathcal{O}}\supset\tilde{\mathcal{O}}'$.

Despite this stronger class of operations in our setting, whether the conversion rate $r_{\tilde{O}}$ in~\eqref{eq:conversion_rate} in our framework is larger or smaller than that in Ref.~\cite{hayashi2025generalizedquantumsteinslemma} is also not a priori obvious since the definition of the conversion rate is also different.
In Ref.~\cite{hayashi2025generalizedquantumsteinslemma}, the conversion rate is defined as
\begin{widetext}
\begin{align}
\label{eq:conversion_rate_previous}
    &r_{\tilde{\mathcal{O}}}'\qty(\Phi_\mathrm{in}\to\Phi_\mathrm{out})\coloneqq\sup\left\{r\geq 0:\exists\qty{\Theta^{(n)}}_n\in\tilde{\mathcal{O}},\liminf_{n\to\infty}\frac{1}{X_\mathrm{out}^{\lceil rn\rceil}}\sum_{x_\mathrm{out}\in\mathcal{X}_\mathrm{out}^{\lceil rn\rceil}}d_{\mathrm{T}}\qty(\qty(\Theta^{(n)}\qty[\Phi_\mathrm{in}^{\otimes n}])\qty(x_\mathrm{out}),\Phi_\mathrm{out}^{\otimes \lceil rn\rceil}\qty(x_\mathrm{out}))=0\right\}.
\end{align}
\end{widetext}
Compared to our definition~\eqref{eq:conversion_rate} of the conversion rate $r_{\tilde{\mathcal{O}}}$ in terms of the diamond distance $d_\diamond$, the definition~\eqref{eq:conversion_rate_previous} of $r_{\tilde{\mathcal{O}}}'$ in Ref.~\cite{hayashi2025generalizedquantumsteinslemma} is defined for a particular Choi-state input rather than the worst-case input in the definition~\eqref{eq:d_diamond} of $d_\diamond$, making the asymptotic conversion task easier.
As a whole, our setting uses a stronger class of operations to achieve a harder approximation in the asymptotic CQ channel conversion.
In place of $R_\mathrm{R}^\infty$ in Theorem~\ref{thm:second_law}, Ref.~\cite{hayashi2025generalizedquantumsteinslemma} used another function
\begin{align}
\label{eq:R_R_infty_prime}
    &R_\mathrm{R}^{\infty\prime}\qty(\Phi)\notag\\
    &\coloneqq\lim_{n\to\infty}\frac{1}{n}\min_{\Phi_\mathrm{free}^{(n)}\in\mathcal{F}}\frac{1}{X^n}\sum_{x\in\mathcal{X}^n}D\left(\Phi^{\otimes n}\qty(x)\middle\|\Phi_\mathrm{free}^{(n)}\qty(x)\right),
\end{align}
to show that
\begin{align}
    r_{\tilde{\mathcal{O}}}'\qty(\Phi_\mathrm{in}\to\Phi_\mathrm{out})=\frac{R_\mathrm{R}^{\infty\prime}\qty(\Phi_\mathrm{in})}{R_\mathrm{R}^{\infty\prime}\qty(\Phi_\mathrm{out})}.
\end{align}
By contrast, Theorem~\ref{thm:second_law} shows that $R_\mathrm{R}^\infty$, defined as~\eqref{eq:regularized_relative_entropy_of_resource} via the channel divergence, characterizes the asymptotic conversion rate in our reversible QRT framework for CQ channel conversion.

\subsection{Characterization of regularized relative entropy of resource in logarithmic generalized robustness for CQ channels}
\label{sec:relative_entropy_robustness}

In this section, we characterize the regularized relative entropy of resource $R_\mathrm{R}^\infty$, defined as~\eqref{eq:regularized_relative_entropy_of_resource}, in terms of the logarithm of the generalized robustness $R_\mathrm{G}$ defined as~\eqref{eq:R_G} for CQ channels.
In the static QRT setting, an analogous characterization of $R_\mathrm{R}^\infty$ for states in terms of the logarithm of $R_\mathrm{G}$ for states was established in Refs.~\cite{Brandao2010,hayashi2025generalizedquantumsteinslemma}.
Our result extends this to CQ channels, in the sense that when the input dimension is one, our statement reduces to the known state case~\cite{Brandao2010,hayashi2025generalizedquantumsteinslemma}.
We note that Ref.~\cite{hayashi2025generalizedquantumsteinslemma} also generalized the state result to CQ channels by characterizing $R_\mathrm{R}^{\infty\prime}$ in~\eqref{eq:R_R_infty_prime}, which is defined using Choi-state inputs.
By contrast, our analysis characterizes $R_\mathrm{R}^\infty$ in~\eqref{eq:regularized_relative_entropy_of_resource}, where the channel divergence is taken to capture the worst-case input.

To this end, we first introduce the method of types for analyzing CQ channels, with the goal of applying it to establish an operator inequality based on the pinching inequality~\eqref{eq:pinching_inequality}.
For a quantum state $\rho$ of a $D$-dimensional system $\mathcal{H}$, written in its spectral decomposition as $\rho=\sum_{x=0}^{D-1}p(x)\ket{x}\bra{x}\in\mathcal{D}(\mathcal{H})$, there are at most $D$ distinct eigenvalues.
However, the $n$-fold tensor product $\rho^{\otimes n}\in\mathcal{D}(\mathcal{H}^{\otimes n})$ on the $D^n$-dimensional space $\mathcal{H}^{\otimes n}$ exhibits large degeneracies in its eigenvalues, which can be systematically understood via the method of types~\cite{hayashi2016quantum}.
Specifically, given a sequence $x^{(n)}=(x_1,\ldots,x_n)\in\mathcal{X}^n$, let $n(x)$ denote the number of occurrences of $x\in\mathcal{X}$ in $x^{(n)}$. The type $t_{x^{(n)}}$ of $x^{(n)}$ is the probability distribution $t_{x^{(n)}}(x)\coloneqq n(x)/n$ for all $x\in\mathcal{X}$. The set of all types of sequences of length $n$ is denoted by $\mathcal{P}^{(n)}$. For each type $t\in\mathcal{P}^{(n)}$, the type class $\mathcal{T}_t^{(n)}$ consists of all sequences $x^{(n)}$ of length $n$ having type $t$, i.e., $\mathcal{T}_t^{(n)}\coloneqq\qty{x^{(n)}\in\mathcal{X}^n:t_{x^{(n)}}=t}$.
In the spectral decomposition of $\rho^{\otimes n}$, all basis vectors $\ket{x^{(n)}}$ belonging to the same type $t$ correspond to the same eigenvalue $p(t)\coloneqq p(x_1)\cdots p(x_n)$.
For each type $t$, define the projection operator $\Pi_t\coloneqq\sum_{x^{(n)}\in\mathcal{T}_t^{(n)}}\ket{x^{(n)}}\bra{x^{(n)}}$, which projects onto the subspace spanned by eigenvectors associated with type $t$. Then,
\begin{align}
\rho^{\otimes n}=\sum_{t\in\mathcal{P}^{(n)}}p(t)\Pi_t,
\end{align}
so that $\rho^{\otimes n}$ has at most $\lvert\mathcal{P}^{(n)}\rvert$ distinct eigenvalues. Since $\dim(\mathcal{H})=|\mathcal{X}|=D$, the number $J_n$ of distinct eigenvalues is bounded by~\cite[Theorem~2.5]{hayashi2016quantum}
\begin{align}
\label{eq:number_of_distinct_eigenvalues_state}
J_n\leq \qty|\mathcal{P}^{(n)}|=\binom{n+D-1}{D-1}\leq (n+1)^{D-1}.
\end{align}
This bound gives a tight control on the number of distinct eigenvalues of $\rho^{\otimes n}$.
In the case of a CQ channel $\Phi\in\mathcal{C}(\mathcal{X}\to\mathcal{H})$, however, the output states $\Phi^{\otimes n}(x^{(n)})$ of the $n$-fold channel $\Phi^{\otimes n}$ need not be IID states of the form $\rho^{\otimes n}$, since they depend on the particular input sequence $x^{(n)}$. We therefore generalize the bound~\eqref{eq:number_of_distinct_eigenvalues_state} to the case of $n$-fold copies of a CQ channel as follows.

\begin{lemma}[\label{lem:number_of_distinct_eigenvalues_CQ_channel}The number of distinct eigenvalues of multiple copies of CQ channels]
    For any positive integer $n$, any CQ channel $\Phi\in\mathcal{C}\qty(\mathcal{X}\to\mathcal{H})$ with finite $X=|\mathcal{X}|$ and $D=\dim\qty(\mathcal{H})$, and any classical input $x^{(n)}\in\mathcal{X}^n$,
    the number of distinct eigenvalues of $\Phi^{\otimes n}\qty(x^{(n)})$ is at most
    \begin{align}
        \qty(n+1)^{X+D-1}.
    \end{align}
\end{lemma}
\begin{proof}
    The Choi operator $J\qty(\Phi^{\otimes n})$ of $\Phi^{n}$ is given by
    \begin{align}
        &J\qty(\Phi^{\otimes n})\notag\\
        &=\sum_{x^{(n)}\in\mathcal{X}^{n}}\ket{x^{(n)}}\bra{x^{(n)}}\otimes\Phi^{\otimes n}\qty(x^{(n)})\\
        &\cong\qty(\sum_{x\in\mathcal{X}}\ket{x}\bra{x}\otimes\Phi^{\otimes n}\qty(x))^{\otimes n}\\
        \label{eq:J_n_iid}
        &=\qty(J\qty(\Phi))^{\otimes n},
    \end{align}
    where $\cong$ means that the equality holds up to permutation of subsystems in the tensor product.
    The spectral decomposition of $\Phi\qty(x)$ for each $x\in\mathcal{X}$ is denoted by
    \begin{align}
        \Phi\qty(x)=\sum_{j=0}^{D-1}\lambda_{x,j}\ket{x,j}\bra{x,j},
    \end{align}
    where $\qty{\lambda_{x,j}}_j$ is the (multi)set of $D$ eigenvalues of $\Phi\qty(x)$ on $\mathcal{H}$ with $D=\dim\qty(\mathcal{H})$, which may include degenerate ones, and $\qty{\ket{x,j}}_{j}$ is the set of eigenvectors.
    Using this notation, the set of eigenvectors of $J\qty(\Phi^{\otimes n})$ is given by
    \begin{align}
    \label{eq:eigenvectors_J_Phi_free_M_l}
        &\left\{\ket{x^{(n)}}\otimes\ket{x_1,j_1}\otimes\cdots\otimes\ket{x_n,j_n}:\right.\notag\\
        &\left.x^{(n)}=\qty(x_1,\ldots,x_n)\in\mathcal{X}^n,\right.\notag\\
        &\left.j_1,\ldots,j_n\in\{0,\ldots,D-1\}\right\}.
    \end{align}
    The spectral decomposition of $J\qty(\Phi^{\otimes n})$ is written as
    \begin{align}
    \label{eq:J_phi_otimes_n}
        J\qty(\Phi^{\otimes n})=\sum_{j=0}^{J_n-1}\lambda_j\Pi_j,
    \end{align}
    where $J_n$ is the number of distinct eigenvalues of $J\qty(\Phi^{\otimes n})$, $\qty{\lambda_j}_j$ denotes the set of distinct eigenvalues, and $\qty{\Pi_j}_j$ denotes the set of projection operators onto the corresponding eigenspaces associated with each eigenvalue.
    With this $J\qty(\Phi^{\otimes n})$, for every input $x^{(n)}\in\mathcal{X}^{n}$, we represent $\Phi^{\otimes n}\qty(x^{(n)})$ as
    \begin{align}
        &\ket{x^{(n)}}\bra{x^{(n)}}\otimes\Phi^{\otimes n}\qty(x^{(n)})\notag\\
        &=J\qty(\Phi^{\otimes n})\qty(\ket{x^{(n)}}\bra{x^{(n)}}\otimes\mathds{1})\\
        \label{eq:Phi_free_representation}
        &=\sum_{j=0}^{J_n-1}\lambda_j\Pi_j\qty(\ket{x^{(n)}}\bra{x^{(n)}}\otimes\mathds{1}),
    \end{align}
    which has the same number of distinct eigenvalues of $\Phi^{\otimes n}\qty(x^{(n)})$ since $\ket{x^{(n)}}\bra{x^{(n)}}$ in the first line has rank one.
    Since $J\qty(\Phi)$ acts on an $(X+D)$-dimensional Hilbert space, due to~\eqref{eq:number_of_distinct_eigenvalues_state},~\eqref{eq:J_n_iid}, and~\eqref{eq:J_phi_otimes_n},
    \begin{align}
    \label{eq:J_l_bound}
        J_n\leq\qty(n+1)^{X+D-1}.
    \end{align}
    Since the support of each $\Pi_j$ in~\eqref{eq:Phi_free_representation} is spanned by a subset of the eigenvectors given in~\eqref{eq:eigenvectors_J_Phi_free_M_l},
    the number of distinct eigenvalues of $\Phi^{\otimes n}\qty(x^{(n)})$ is also upper-bounded by $J_n$ in~\eqref{eq:J_l_bound}.
\end{proof}

To establish the relation between the relative entropy of resource and the generalized robustness for CQ channels, we make use of the following characterization of the generalized robustness in terms of operator inequalities.
This extends the characterization for states originally established in Ref.~\cite{doi:10.1142/S0219749909005298} to CQ channels.

\begin{lemma}[\label{lem:generalized_robustness_operator_inequalities}Characterization of generalized robustness by operator inequalities for CQ channels]
    For any family $\mathcal{F}$ of sets of free CQ channels satisfying Axioms~\ref{p:full_rank} and~\ref{p:compact},
    and any CQ channel $\Phi\in\mathcal{C}\qty(\mathcal{X}\to\mathcal{H})$,
    it holds that
    \begin{align}
    \label{eq:generalized_robustness_operator_inequalities}
        R_\mathrm{G}\qty(\Phi)&=\min\left\{s\geq 0:\right.\notag\\
        &\quad\left.\exists\Phi_\mathrm{free}\in\mathcal{F}, \forall x\in\mathcal{X}, \Phi(x)\leq \qty(1+s)\Phi_\mathrm{free}(x)\right\},
    \end{align}
    where $R_\mathrm{G}$ is defined as~\eqref{eq:R_G}.
\end{lemma}

\begin{proof}
    We write the right-hand side of~\eqref{eq:generalized_robustness_operator_inequalities} as
    \begin{align}
    \label{eq:R_G_prime}
        R_\mathrm{G}'\qty(\Phi)&\coloneqq\min\left\{s\geq 0:\right.\notag\\
        &\quad\left.\exists\Phi_\mathrm{free}\in\mathcal{F}, \forall x\in\mathcal{X}, \Phi(x)\leq \qty(1+s)\Phi_\mathrm{free}(x)\right\}.
    \end{align}
    Axioms~\ref{p:full_rank} and~\ref{p:compact} guarantee the finiteness and the existence of minima in the definition~\eqref{eq:R_G} of $R_\mathrm{G}$ and~\eqref{eq:R_G_prime} of $R_\mathrm{G}'$.
    Our proof will show
    \begin{align}
    \label{eq:generalized_robustness_operator_inequality_leq}
    R_\mathrm{G}\qty(\Phi)&\geq R_\mathrm{G}'\qty(\Phi),~\text{and}\\
    \label{eq:generalized_robustness_operator_inequality_geq}
    R_\mathrm{G}\qty(\Phi)&\leq R_\mathrm{G}'\qty(\Phi).
    \end{align}

    \textbf{Proof of~\eqref{eq:generalized_robustness_operator_inequality_leq}.}
    By the definition~\eqref{eq:R_G} of $R_\mathrm{G}$, we have a CQ channel $\Phi'\in\mathcal{C}\qty(\mathcal{X}\to\mathcal{H})$ and a free CQ channel $\Phi_\mathrm{free}\in\mathcal{F}\qty(\mathcal{X}\to\mathcal{H})$ such that
    \begin{align}
    \frac{\Phi+R_\mathrm{G}\qty(\Phi)\Phi'}{1+R_\mathrm{G}\qty(\Phi)}=\Phi_\mathrm{free}.
    \end{align}
    Thus, for every input $x\in\mathcal{X}$, we have
    \begin{align}
        \Phi(x)&\leq\Phi(x)+R_\mathrm{G}\qty(\Phi)\Phi'(x)\\
        &=\qty(1+R_\mathrm{G}\qty(\Phi))\Phi_\mathrm{free}(x),
    \end{align}
    which yields~\eqref{eq:generalized_robustness_operator_inequality_leq} by the definition~\eqref{eq:R_G_prime} of $R_\mathrm{G}'$.

    \textbf{Proof of~\eqref{eq:generalized_robustness_operator_inequality_geq}.}
    By the definition~\eqref{eq:R_G_prime} of $R_\mathrm{G}'$, we have a free CQ channel $\Phi_\mathrm{free}\in\mathcal{F}\qty(\mathcal{X}\to\mathcal{H})$ such that, for every input $x\in\mathcal{X}$,
    \begin{align}
    \label{eq:generalized_robustness_operator_inequality_geq_assumption}
        \Phi(x)\leq\qty(1+R_\mathrm{G}'\qty(\Phi))\Phi_\mathrm{free}(x).
    \end{align}
    We then define a CQ channel $\Phi'$ for every input $x\in\mathcal{X}$ as
    \begin{align}
        \Phi'\qty(x)\coloneqq\frac{\qty(1+R_\mathrm{G}'\qty(\Phi))\Phi_\mathrm{free}(x)-\Phi(x)}{R_\mathrm{G}'\qty(\Phi)},
    \end{align}
    where $\Phi'\qty(x)\geq 0$ follows from~\eqref{eq:generalized_robustness_operator_inequality_geq_assumption}, and $\Tr\qty[\Phi'\qty(x)]=1$ follows from $\Tr\qty[\Phi_\mathrm{free}(x)]=1$ and $\Tr\qty[\Phi(x)]=1$.
    This CQ channel $\Phi'$ satisfies
    \begin{align}
    \frac{\Phi+R_\mathrm{G}'\qty(\Phi)\Phi'}{1+R_\mathrm{G}'\qty(\Phi)}=\Phi_\mathrm{free},
    \end{align}
    which yields~\eqref{eq:generalized_robustness_operator_inequality_geq} by the definition~\eqref{eq:R_G} of $R_\mathrm{G}$.
\end{proof}

Using the techniques developed above, we derive a lower bound on the regularized relative entropy of resource in terms of the logarithmic generalized robustness for CQ channels.
In contrast to Ref.~\cite{hayashi2025generalizedquantumsteinslemma}, where the corresponding result for $R_\mathrm{R}^{\infty\prime}$ in~\eqref{eq:R_R_infty_prime} was shown using an additional assumption of Axiom~\ref{p:convex}, our analysis provides a simpler construction that eliminates the need for this axiom in this bound.

\begin{lemma}[\label{lem:lower_bound_log_robustness}Lower bound on the regularized relative entropy of resource in terms of logarithmic generalized robustness for CQ channels]
    For any family $\mathcal{F}$ of sets of free CQ channels satisfying Axioms~\ref{p:full_rank},~\ref{p:compact}, and~\ref{p:tensor_product},
    and any CQ channel $\Phi\in\mathcal{C}\qty(\mathcal{X}\to\mathcal{H})$,
    there exists a sequence $\qty{\Phi^{(n)}}_{n=1,2,\ldots}$ of CQ channels such that
    \begin{align}
    \label{eq:lower_bound_reguralized_relative_entropy_of_resource}
        &R_\mathrm{R}^\infty\qty(\Phi)\geq\limsup_{n\to\infty}\frac{1}{n}\log\qty[1+R_\mathrm{G}\qty(\Phi^{(n)})],\\
    \label{eq:diamond_norm_convergence}
        &\lim_{n\to\infty}d_\diamond\qty(\Phi^{(n)},\Phi^{\otimes n})=0.
    \end{align}
    where $R_\mathrm{R}^\infty$ is defined as~\eqref{eq:regularized_relative_entropy_of_resource}, $R_\mathrm{G}$ is defines as~\eqref{eq:R_G}, and $d_\diamond$ is defined as~\eqref{eq:d_diamond}.
\end{lemma}

\begin{proof}
    Due to Lemma~\ref{prp:existence_limit} with
    Axioms~\ref{p:full_rank},~\ref{p:compact}, and~\ref{p:tensor_product}, $R_\mathrm{R}^\infty\qty(\Phi)$ is well-defined.
    Below, we first construct the sequence $\qty{\Phi^{(n)}}_{n=1,2,\ldots}$, followed by proving that this sequence satisfies~\eqref{eq:lower_bound_reguralized_relative_entropy_of_resource} and~\eqref{eq:diamond_norm_convergence}.

    \textbf{Construction of $\qty{\Phi^{(n)}}_{n=1,2,\ldots}$.}
    We choose any $R$ satisfying
    \begin{align}
    \label{eq:R_lower_bound_log_robustness}
        R>R_\mathrm{R}^\infty\qty(\Phi).
    \end{align}
    By definition of $R_\mathrm{R}^\infty$ in~\eqref{eq:regularized_relative_entropy_of_resource}, there exists a positive integer $M$ and a free CQ channel $\Phi_\mathrm{free}^{(M)}\in\mathcal{F}\qty(\mathcal{X}^M\to\mathcal{H}^{\otimes M})$ such that
    \begin{align}
    \label{eq:R_lower_bound_M}
        \frac{1}{M} D\left(\Phi^{\otimes M}\middle\|\Phi_\mathrm{free}^{(M)}\right)<R.
    \end{align}
    With this fixed $M$, we choose $q_n$ and $r_n$ for every $n$ such that
    \begin{align}
    \label{eq:lower_bound_log_robustness_n}
        n&=q_n M+r_n,~0\leq r_n<M,
    \end{align}
    and define
    \begin{align}
        \Phi_\mathrm{free}^{(n)}\coloneqq\Phi_\mathrm{free}^{(M)\otimes q_n}\otimes\Phi_\mathrm{full}^{\otimes r_n},
    \end{align}
    where $\Phi_\mathrm{full}\in\mathcal{F}\qty(\mathcal{X}\to\mathcal{H})$ is the free CQ channel satisfying the full-rank condition in Axiom~\ref{p:full_rank}.
    Due to Axiom~\ref{p:tensor_product}, $\Phi_\mathrm{free}^{(n)}$ is a free CQ channel
    \begin{align}
        \Phi_\mathrm{free}^{(n)}\in\mathcal{F}\qty(\mathcal{X}^n\to\mathcal{H}^{\otimes n}).
    \end{align}

    To define a sequence $\qty{\Phi^{(n)}}_{n=1,2,\ldots}$ of CQ channels,
    let $\mathcal{P}_{\Phi_\mathrm{free}^{(n)}}$ denote the pinching superchannel with respect to $\Phi_\mathrm{free}^{(n)}$, as in~\eqref{eq:pinching_superchannel}.
    We define a projection
    \begin{align}
    \label{eq:T_x_l_M}
        &T_{x^{(n)}}\coloneqq\qty{\qty(\mathcal{P}_{\Phi_\mathrm{free}^{(n)}}\qty[\Phi^{\otimes n}])\qty(x^{(n)})\geq e^{Rn}\Phi_\mathrm{free}^{(n)}\qty(x^{(n)})}.
    \end{align}
    Using the POVM $\qty{T_{x^{(n)}},\mathds{1}-T_{x^{(n)}}}$ for every $n$, we define a CQ channel
    \begin{align}
    \label{eq:Phi_n_robustness}
        &\Phi^{(n)}\qty(x^{(n)})\notag\\
        &\coloneqq \qty(\mathds{1}-T_{x^{(n)}})\Phi^{\otimes n}\qty(x^{(n)})\qty(\mathds{1}-T_{x^{(n)}})+\notag\\
        &\quad\Tr\qty[T_{x^{\qty(n)}}\Phi^{\otimes n}\qty(x^{(n)})]\Phi_\mathrm{free}^{(n)}\qty(x^{(n)}).
    \end{align}

    \textbf{Proof of~\eqref{eq:lower_bound_reguralized_relative_entropy_of_resource}.}
    To show~\eqref{eq:lower_bound_reguralized_relative_entropy_of_resource}, we will bound $R_\mathrm{G}\qty(\Phi^{(n)})$ by showing operator inequalities and converting them to the bounds on $R_\mathrm{G}$.
    In the definition~\eqref{eq:Phi_n_robustness} of $\Phi^{(n)}$, the second term on the right-hand side can be bounded for every input $x^{(n)}\in\mathcal{X}^{n}$ by
    \begin{align}
    \label{eq:operator_inequality_second_term}
        &\Tr\qty[T_{x^{\qty(n)}}\Phi^{\otimes n}\qty(x^{(n)})]\Phi_\mathrm{free}^{(n)}\qty(x^{(n)})\leq\Phi_\mathrm{free}^{(n)}\qty(x^{(n)}).
    \end{align}

    To obtain a similar operator inequality for the first term on the right-hand side of~\eqref{eq:Phi_n_robustness},
    let $J_{n}$ denote the maximum number of distinct eigenvalues of $\Phi_\mathrm{free}^{(n)}\qty(x^{(n)})$ over all $x^{(n)}\in\mathcal{X}^n$; then, for every $x^{(n)}\in\mathcal{X}^n$, we obtain from the pinching inequality in~\eqref{eq:pinching_inequality}
    \begin{align}
    \label{eq:Phi_pinching}
        \Phi^{\otimes n}\qty(x^{(n)})\leq J_n\qty(\mathcal{P}_{\Phi_\mathrm{free}^{(n)}}\qty[\Phi^{\otimes n}])\qty(x^{(n)}).
    \end{align}
    Moreover, we obtain from~\eqref{eq:T_x_l_M} that
    \begin{align}
    \label{eq:operator_inequality_T}
    &\qty(\mathds{1}-T_{x^{(n)}})\qty(\mathcal{P}_{\Phi_\mathrm{free}^{(n)}}\qty[\Phi^{\otimes n}])\qty(x^{(n)})\qty(\mathds{1}-T_{x^{(n)}})\notag\\
    &\leq e^{Rn}\Phi_\mathrm{free}^{(n)}\qty(x^{(n)}),
    \end{align}
    which follows from the fact that $\qty(\mathcal{P}_{\Phi_\mathrm{free}^{(n)}}\qty[\Phi^{\otimes n}])\qty(x^{(n)})$, $\Phi_\mathrm{free}^{(n)}\qty(x^{(n)})$, and $T_{x^{(n)}}$ all commute with each other as a result of the pinching.
    Due to~\eqref{eq:Phi_pinching} and~\eqref{eq:operator_inequality_T}, we have, for every input $x^{(n)}\in\mathcal{X}^{n}$,
    \begin{align}
    &\qty(\mathds{1}-T_{x^{(n)}})\Phi^{\otimes n}\qty(x^{(n)})\qty(\mathds{1}-T_{x^{(n)}})\notag\\
    &\leq J_n\qty(\mathds{1}-T_{x^{(n)}})\qty(\mathcal{P}_{\Phi_\mathrm{free}^{(n)}}\qty[\Phi^{\otimes n}])\qty(x^{(n)})\qty(\mathds{1}-T_{x^{(n)}})\\
    \label{eq:operator_inequality_first_term}
    &\leq J_n e^{Rn}\Phi_\mathrm{free}^{(n)}\qty(x^{(n)}).
    \end{align}

    As a whole, for the CQ channel $\Phi^{(n)}$ in~\eqref{eq:Phi_n_robustness}, it follows from from~\eqref{eq:operator_inequality_second_term} and~\eqref{eq:operator_inequality_first_term} that, for every input $x^{(n)}\in\mathcal{X}^{n}$,
    \begin{align}
        \Phi^{(n)}\qty(x^{(n)})\leq\qty(1+J_n e^{Rn})\Phi_\mathrm{free}^{(n)}\qty(x^{(n)}).
    \end{align}
    Therefore, Lemma~\ref{lem:generalized_robustness_operator_inequalities} shows
    \begin{align}
    \label{eq:robustness_Phi_lM}
        R_\mathrm{G}\qty(\Phi^{(n)})\leq J_n e^{Rn},
    \end{align}

    We will bound $J_n$ in~\eqref{eq:robustness_Phi_lM}, which has appeared in~\eqref{eq:Phi_pinching} as the maximum number of distinct eigenvalues of
    \begin{align}
        \Phi_\mathrm{free}^{(n)}\qty(x^{(n)})=\Phi_\mathrm{free}^{(M)\otimes q_{n}}\qty(x^{(q_{n} M)})\otimes\Phi_\mathrm{full}^{\otimes r_{n}}\qty(x^{(r_{n})})
    \end{align}
    with $q_{n}$ and $r_{n}$ in~\eqref{eq:lower_bound_log_robustness_n}.
    Lemma~\ref{lem:number_of_distinct_eigenvalues_CQ_channel} shows that the number $J_n^{(1)}$ of distinct eigenvalues of $\Phi_\mathrm{free}^{(M)\otimes q_{n}}\qty(x^{(q_{n} M)})$ is bounded by
    \begin{align}
        \label{eq:J_n_bound_log_robustness_1}
        J_n^{(1)}\leq(q_{n}+1)^{X^M+D^M-1},
    \end{align}
    and the number $J_n^{(2)}$ of distinct eigenvalues of $\Phi_\mathrm{full}^{\otimes r_{n}}\qty(x^{(r_{n})})$ is bounded by
    \begin{align}
        \label{eq:J_n_bound_log_robustness_2}
        J_n^{(2)}\leq(r_{n}+1)^{X+D-1},
    \end{align}
    where $X=|\mathcal{X}|$, $D=\dim\qty(\mathcal{H})$, and $M$ are constants.
    Hence, we have
    \begin{align}
        J_n&\leq J_n^{(1)}J_n^{(2)}\\
        &\leq(q_{n}+1)^{X^M+D^M-1}(r_{n}+1)^{X+D-1},
    \end{align}
    and, due to~\eqref{eq:lower_bound_log_robustness_n},
    \begin{align}
    \label{eq:J_n_limit}
        \lim_{n\to\infty}\frac{\log\qty[J_n]}{n}=0.
    \end{align}
    
    Therefore, due to~\eqref{eq:robustness_Phi_lM} and~\eqref{eq:J_n_limit}, we obtain
    \begin{align}
        &\limsup_{n\to\infty}\frac{1}{n}\log\qty[1+R_\mathrm{G}\qty(\Phi^{(n)})]\notag\\
        &\leq\limsup_{n\to\infty}\frac{1}{n}\log\qty[1+J_n e^{Rn}]\\
        &= R,
    \end{align}
    which holds for any $R$ satisfying~\eqref{eq:R_lower_bound_log_robustness}.
    In the limit $R\to R_\mathrm{R}^\infty\qty(\Phi)$, we have
    \begin{align}
        \limsup_{n\to\infty}\frac{1}{n}\log\qty[1+R_\mathrm{G}\qty(\Phi^{(n)})]\leq R_\mathrm{R}^\infty\qty(\Phi).
    \end{align}

    \textbf{Proof of~\eqref{eq:diamond_norm_convergence}}.
    By the definition~\eqref{eq:T_x_l_M} of $\qty{T_{x^{(n)}}}_{x^{(n)}\in\mathcal{X}^n}$, it holds that
    \begin{align}
        &\Tr\qty[T_{x^{(n)}}\qty(e^{-Rn}\mathcal{P}_{\Phi_\mathrm{free}^{(n)}}\qty[\Phi^{\otimes n}]\qty(x^{(n)})-\Phi_\mathrm{free}^{(n)}\qty(x^{(n)}))]\geq 0.
    \end{align}
    Hence, it holds for every $n$ that
    \begin{align}
        &\Tr\qty[T_{x^{(n)}}\Phi_\mathrm{free}^{(n)}\qty(x^{(n)})]\notag\\
        &\leq e^{-Rn}\Tr\qty[T_{x^{(n)}}\mathcal{P}_{x^{(n)}}\qty(\Phi^{\otimes n}\qty(x^{(n)}))]\\
        \label{eq:T_Phi_free_limit_zero}
        &\leq e^{-Rn}.
    \end{align}
    By taking the limit $n\to\infty$, we obtain
    \begin{align}
        &\liminf_{n\to\infty}-\frac{1}{n}\log\qty[\Tr\qty[T_{x^{(n)}}\Phi_\mathrm{free}^{(n)}\qty(x^{(n)})]]\notag\\
        &\geq R\\
        &>\frac{1}{M}D\left(\Phi^{\otimes M}\middle\|\Phi_\mathrm{free}^{(M)}\right),
    \end{align}
    where the last line follows from~\eqref{eq:R_lower_bound_M}.
    Therefore, contraposition of Lemma~\ref{lem:strong_converse_CQ_channel_copies} shows that, for any $\epsilon\in[0,1)$ and any sequence $\{p_n\}_{n=1,2,\ldots}$ of probability distributions, we have
    \begin{align}
        \liminf_{n\to\infty}\sum_{x^{(n)}\in\mathcal{X}^n}p_n\qty(x^{(n)})\Tr\qty[\qty(\mathds{1}-T_{x^{(n)}})\Phi^{\otimes n}\qty(x^{(n)})]>\epsilon,
    \end{align}
    meaning that, for every $\qty{x^{(n)}\in\mathcal{X}^n}_{n=1,2,\ldots}$,
    \begin{align}
        \lim_{n\to\infty}\Tr\qty[\qty(\mathds{1}-T_{x^{(n)}})\Phi^{\otimes n}\qty(x^{(n)})]=1.
    \end{align}
    Thus, for any $x^{(n_l)}\in\mathcal{X}^{n_l}$, we have
    \begin{align}
    \label{eq:T_Phi_lM_limit_zero}
        \lim_{n\to\infty}\Tr\qty[T_{x^{(n)}}\Phi^{\otimes n}\qty(x^{(n)})]=0.
    \end{align}
    Consequently, for $\Phi^{(n)}$ in~\eqref{eq:Phi_n_robustness}, it holds for all $x^{(n)}\in\mathcal{X}^n$ that
    \begin{widetext}
    \begin{align}
        &\frac{1}{2}\left\|\Phi^{(n)}\qty(x^{(n)})-\Phi^{\otimes n}\qty(x^{(n)})\right\|_1\notag\\
        &=\frac{1}{2}\left\|\qty(\mathds{1}-T_{x^{(n)}})\Phi^{\otimes n}\qty(x^{(n)})\qty(\mathds{1}-T_{x^{(n)}})+\Tr\qty[T_{x^{\qty(n)}}\Phi^{\otimes n}\qty(x^{(n)})]\Phi_\mathrm{free}^{(M)\otimes l}\qty(x^{(n)})-\Phi^{\otimes n}\qty(x^{(n)})\right\|_1\\
        \label{eq:diamond_norm_evaluation_1}
        &\leq\frac{1}{2}\left(\left\|T_{x^{(n)}}\Phi^{\otimes n}\qty(x^{(n)})\right\|_1+\left\|\Phi^{\otimes n}\qty(x^{(n)})T_{x^{(n)}}\right\|_1+\left\|T_{x^{(n)}}\Phi^{\otimes n}\qty(x^{(n)})T_{x^{(n)}}\right\|_1+
        \Tr\qty[T_{x^{\qty(n)}}\Phi^{\otimes n}\qty(x^{(n)})]\left\|\Phi_\mathrm{free}^{(n)}\qty(x^{(n)})\right\|_1\right)\\
        \label{eq:diamond_norm_evaluation_2}
        &\leq\frac{1}{2}\qty(4\Tr\qty[T_{x^{(n)}}\Phi^{\otimes n}\qty(x^{(n)})])\\
        \label{eq:diamond_norm_evaluation_3}
        &\to 0~\text{as $n\to\infty$},
    \end{align}
    \end{widetext}
    where~\eqref{eq:diamond_norm_evaluation_1} follows from the subadditivity of the norm, and~\eqref{eq:diamond_norm_evaluation_2} uses the fact that $T_{x^{(n)}}\geq 0$, $\Phi^{\otimes n}\qty(x^{(n)})\geq 0$, and $\Phi_\mathrm{free}^{(n)}\qty(x^{(n)})\in\mathcal{D}\qty(\mathcal{H}^{\otimes n})$.
    Therefore, $d_\diamond$ in~\eqref{eq:d_diamond} is bounded by
    \begin{align}
        \lim_{n\to\infty}d_\diamond\qty(\Phi^{(n)},\Phi^{\otimes n})=0.
    \end{align}
    
\end{proof}

On the other hand, we show upper bounds on the regularized relative entropy of resource in terms of the logarithmic generalized robustness for CQ channels.

\begin{lemma}[\label{lem:upper_bound_log_robustness}Upper bound on the regularized relative entropy of resource in terms of logarithmic generalized robustness for CQ channels]
    For any family $\mathcal{F}$ of sets of free CQ channels satisfying Axioms~\ref{p:full_rank},~\ref{p:compact},~\ref{p:tensor_product}, and~\ref{p:convex},
    and any sequences $\qty{\Phi^{(n)}\in\mathcal{C}\qty(\mathcal{X}^n\to\mathcal{H}^{\otimes n})}_{n=1,2,\ldots}$ and $\qty{\tilde{\Phi}^{(n)}\in\mathcal{C}\qty(\mathcal{X}^n\to\mathcal{H}^{\otimes n})}_{n=1,2,\ldots}$ of CQ channels,
    if it holds that
    \begin{align}
    \label{eq:upper_bound_log_robustness_diamond_distance}
        \lim_{n\to\infty}d_\diamond\qty(\tilde{\Phi}^{(n)},\Phi^{(n)})=0,
    \end{align}
    then we have
    \begin{align}
        \liminf_{n\to\infty}\frac{1}{n}R_\mathrm{R}\qty(\Phi^{(n)})&\leq\liminf_{n\to\infty}\frac{1}{n}\log\qty[1+R_\mathrm{G}\qty(\tilde{\Phi}^{(n)})],\\
        \limsup_{n\to\infty}\frac{1}{n}R_\mathrm{R}\qty(\Phi^{(n)})&\leq\limsup_{n\to\infty}\frac{1}{n}\log\qty[1+R_\mathrm{G}\qty(\tilde{\Phi}^{(n)})],
    \end{align}
    where $R_\mathrm{R}$ is defined as~\eqref{eq:relative_entropy_of_resource}, $R_\mathrm{G}$ is defines as~\eqref{eq:R_G}, and $d_\diamond$ is defined as~\eqref{eq:d_diamond}.
\end{lemma}

\begin{proof}
    Due to Lemma~\ref{lem:generalized_robustness_operator_inequalities}, we have a sequence $\qty{\Phi_\mathrm{free}^{(n)}\in\mathcal{F}\qty(\mathcal{X}^n\to\mathcal{H}^{\otimes n})}_{n=1,2,\ldots}$ of free CQ channels such that, for every $x^{(n)}\in\mathcal{X}^n$,
    \begin{align}
        \tilde{\Phi}^{(n)}\qty(x^{(n)})\leq\qty(1+R_\mathrm{G}\qty(\tilde{\Phi}^{(n)}))\Phi_\mathrm{free}^{(n)}\qty(x^{(n)}).
    \end{align}
    Then, Lemma~\ref{lem:bound_channel_divergence_operator_inequality} shows that
    \begin{align}
        D\left(\tilde{\Phi}^{(n)}\middle\|\Phi_\mathrm{free}^{(n)}\right)\leq\log\qty[1+R_\mathrm{G}\qty(\tilde{\Phi}^{(n)})].
    \end{align}
    Therefore, by the definition~\eqref{eq:relative_entropy_of_resource} of $R_\mathrm{R}$, we have, for all $n$,
    \begin{align}
        R_\mathrm{R}\qty(\tilde{\Phi}^{(n)})
        &\leq D\left(\tilde{\Phi}^{(n)}\middle\|\Phi_\mathrm{free}^{(n)}\right)\\
        \label{eq:upper_bound_log_robustness1}
        &\leq\log\qty[1+R_\mathrm{G}\qty(\tilde{\Phi}^{(n)})].
    \end{align}
    Due to the asymptotic continuity~\eqref{eq:D_asymptotic_continuity} of $R_\mathrm{R}$ under Axioms~\ref{p:full_rank},~\ref{p:compact},~\ref{p:tensor_product}, and~\ref{p:convex}, we obtain from~\eqref{eq:upper_bound_log_robustness_diamond_distance}
    \begin{align}
        \liminf_{n\to\infty}\frac{1}{n}R_\mathrm{R}\qty(\Phi^{(n)})&=\liminf_{n\to\infty}\frac{1}{n}R_\mathrm{R}\qty(\tilde{\Phi}^{(n)}),\\
        \limsup_{n\to\infty}\frac{1}{n}R_\mathrm{R}\qty(\Phi^{(n)})&=\limsup_{n\to\infty}\frac{1}{n}R_\mathrm{R}\qty(\tilde{\Phi}^{(n)}).
    \end{align}
    Therefore, we obtain from~\eqref{eq:upper_bound_log_robustness1}
    \begin{align}
        \liminf_{n\to\infty}\frac{1}{n}R_\mathrm{R}\qty(\Phi^{(n)})&\leq\liminf_{n\to\infty}\frac{1}{n}\log\qty[1+R_\mathrm{G}\qty(\tilde{\Phi}^{(n)})],\\
        \limsup_{n\to\infty}\frac{1}{n}R_\mathrm{R}\qty(\Phi^{(n)})&\leq\limsup_{n\to\infty}\frac{1}{n}\log\qty[1+R_\mathrm{G}\qty(\tilde{\Phi}^{(n)})].
    \end{align}
\end{proof}

By combining the lower and upper bounds established above, we obtain the following characterization of the regularized relative entropy of resource in terms of the logarithmic generalized robustness for CQ channels.

\begin{proposition}
[\label{prp:robustness_characterization}Characterization of regularized relative entropy of resource in terms of logarithmic generalized robustness for CQ channels]
For any family $\mathcal{F}$ of sets of free CQ channels satisfying Axioms~\ref{p:full_rank},~\ref{p:compact},~\ref{p:tensor_product}, and~\ref{p:convex},
and any CQ channel $\Phi\in\mathcal{C}\qty(\mathcal{X}\to\mathcal{H})$,
it holds that
\begin{align}
&R_\mathrm{R}^\infty\qty(\Phi)\notag\\
&= 
\min_{\qty{\tilde{\Phi}^{(n)}}_{n=1,2,\ldots}}
\left\{
\lim_{n\to\infty}
\frac{1}{n}
\log\qty(1+R_\mathrm{G}\qty(\tilde{\Phi}^{(n)})):\right.\notag\\
&\quad
\lim_{n\to \infty}d_\diamond\qty(\tilde{\Phi}^{(n)},\Phi^{\otimes n})
= 0,\notag\\
&\quad\left.\text{the limit $\lim_{n\to\infty}\frac{1}{n}
\log\qty(1+R_\mathrm{G}\qty(\tilde{\Phi}^{(n)}))$ exists}
\right\},
\end{align}
    where $R_\mathrm{R}$ is defined as~\eqref{eq:relative_entropy_of_resource}, $R_\mathrm{G}$ is defines as~\eqref{eq:R_G}, $d_\diamond$ is defined as~\eqref{eq:d_diamond}, and the minimum on the right-hand side exists.
\end{proposition}

\begin{proof}
Due to Axioms~\ref{p:full_rank},~\ref{p:compact},~\ref{p:tensor_product}, and~\ref{p:convex}, Lemma~\ref{lem:upper_bound_log_robustness} shows that
\begin{align}
&R_\mathrm{R}^\infty\qty(\Phi)\notag\\
&=\liminf_{n\to\infty}\frac{1}{n}R_\mathrm{R}(\Phi^{\otimes n})\\
&\leq\inf_{\qty{\tilde{\Phi}^{(n)}}_n}\left\{
\liminf_{n\to\infty}
\frac{1}{n} \log\qty(1+R_\mathrm{G}\qty(\tilde{\Phi}^{(n)})):\right.\notag\\
&\quad\left.\lim_{n\to \infty}
d_\diamond\qty(\tilde{\Phi}^{(n)},\Phi^{\otimes n})
= 0 \right\}\\
&\leq\inf_{\qty{\tilde{\Phi}^{(n)}}_n}\left\{
\liminf_{n\to\infty}
\frac{1}{n} \log\qty(1+R_\mathrm{G}\qty(\tilde{\Phi}^{(n)})):\right.\notag\\
&\quad
\lim_{n\to \infty} d_\diamond\qty(\tilde{\Phi}^{(n)},\Phi^{\otimes n}) = 0,
\notag\\
&\quad\left.
\text{the limit $\lim_{n\to\infty}\frac{1}{n}
\log\qty(1+R_\mathrm{G}\qty(\tilde{\Phi}^{(n)}))$ exists}
\right\},
\label{eq:limit_liminf_R_R}
\end{align}
where the last inequality is due to adding the constraint that the limit should exist.
Due to Axioms~\ref{p:full_rank},~\ref{p:compact}, and~\ref{p:tensor_product}, it follows from Lemma~\ref{lem:lower_bound_log_robustness} that there exists a sequence $\qty{\tilde{\Phi}^{(n)}}_n$ achieving
\begin{align}
\limsup_{n\to\infty}
\frac{1}{n}
\log\qty(1+R_\mathrm{G}\qty(\tilde{\Phi}^{(n)}))&\leq R_\mathrm{R}^\infty(\Phi),\\
\lim_{n\to \infty} d_\diamond\qty(\tilde{\Phi}^{(n)},\Phi^{\otimes n}) &= 0.
\end{align}
Thus, for this sequence, the limit
\begin{align}
\label{eq:limit_existence}
&\lim_{n\to\infty} \frac{1}{n} \log\qty(1+R_\mathrm{G}\qty(\tilde{\Phi}^{(n)}))\notag\\
&=\liminf_{n\to\infty} \frac{1}{n} \log\qty(1+R_\mathrm{G}\qty(\tilde{\Phi}^{(n)}))\notag\\
&=\limsup_{n\to\infty} \frac{1}{n} \log\qty(1+R_\mathrm{G}\qty(\tilde{\Phi}^{(n)}))\notag\\
&=R_\mathrm{R}^\infty(\mathcal{N})
\end{align}
exists.
Due to the existence of $\qty{\tilde{\Phi}^{(n)}}_n$ achieving~\eqref{eq:limit_existence}, it follows from~\eqref{eq:limit_liminf_R_R} that
\begin{align}
&R_\mathrm{R}^\infty\qty(\Phi)\notag\\
&= 
\min_{\qty{\tilde{\Phi}^{(n)}}_{n}}
\left\{
\liminf_{n\to\infty}
\frac{1}{n}
\log\qty(1+R_\mathrm{G}\qty(\tilde{\Phi}^{(n)})):\right.\notag\\
&\quad\left.
\lim_{n\to \infty}d_\diamond\qty(\tilde{\Phi}^{(n)},\Phi^{\otimes n})
= 0
\right\}\\
&= 
\min_{\qty{\tilde{\Phi}^{(n)}}_{n}}
\left\{
\lim_{n\to\infty}
\frac{1}{n}
\log\qty(1+R_\mathrm{G}\qty(\tilde{\Phi}^{(n)})):\right.\notag\\
&\quad
\lim_{n\to \infty}d_\diamond\qty(\tilde{\Phi}^{(n)},\Phi^{\otimes n})
= 0,\notag\\
&\quad\left.\text{the limit $\lim_{n\to\infty}\frac{1}{n}
\log\qty(1+R_\mathrm{G}\qty(\tilde{\Phi}^{(n)}))$ exists}
\right\},
\end{align}
where the minima exist.
\end{proof}

\subsection{Main parts of proof for reversible QRT framework for CQ channel conversion}
\label{sec:conversion_rate}

In this section, we present the techniques used to establish the reversibility stated in Theorem~\ref{thm:second_law}.
The proof consists of two parts: the converse part, which establishes optimality, and the direct part, which demonstrates achievability.
Section~\ref{sec:converse_second_law} is devoted to the converse part, while Section~\ref{sec:direct_second_law} addresses the direct part.

\subsubsection{Converse part}
\label{sec:converse_second_law}

In this section, we show the converse part of Theorem~\ref{thm:second_law}.
To this end, we first show the following asymptotic version of the monotonicity of the regularized relative entropy of resource under asymptotically resource-non-generating operations.

\begin{lemma}[\label{lem:monotonicity_regularized_relative_entropy_of_resource}Monotonicity of regularized relative entropy of resource under asymptotically resource-non-generating operations]
For any family $\mathcal{F}$ of sets of free CQ channels satisfying Axioms~\ref{p:full_rank},~\ref{p:compact},~\ref{p:tensor_product}, and~\ref{p:convex},
any sequence $\qty{\Theta^{(n)}}_{n=1,2,\ldots}\in\tilde{\mathcal{O}}\qty(\qty(\mathcal{X}\to\mathcal{H})\to\qty(\mathcal{X}'\to\mathcal{H}'))$ of superchannels satisfying Axiom~\ref{p:arng},
and any CQ channel $\Phi\in\mathcal{C}\qty(\mathcal{X}\to\mathcal{H})$,
it holds that
\begin{align}
R_\mathrm{R}^\infty\qty(\Phi)\geq \limsup_{n\to\infty}\frac{1}{n}R_\mathrm{R}\qty(\Theta^{(n)}\qty(\Phi^{\otimes n})),
\end{align}
where $R_\mathrm{R}$ and $R_\mathrm{R}^\infty$ are defined as~\eqref{eq:relative_entropy_of_resource} and~\eqref{eq:regularized_relative_entropy_of_resource}, respectively.
\end{lemma}
\begin{proof}
    Under Axioms~\ref{p:full_rank},~\ref{p:compact},~\ref{p:tensor_product}, and~\ref{p:convex}, due to Proposition~\ref{prp:robustness_characterization}, we have a sequence $\qty{\tilde{\Phi}^{(n)}}_{n=1,2,\ldots}$ of CQ channels satisfying
    \begin{align}
    \label{eq:monotonicity_regularized_relative_entropy_of_resource1}
        &R_\mathrm{R}^\infty\qty(\Phi)=\lim_{n\to\infty}\frac{1}{n}\log\qty[1+R_\mathrm{G}\qty(\tilde{\Phi}^{(n)})],\\
        \label{eq:diamond_tilde_phi_phi}
        &\lim_{n\to\infty}d_\diamond\qty(\tilde{\Phi}^{(n)},\Phi^{\otimes n})=0.
    \end{align}
    Due to the monotonicity~\eqref{eq:monotonicity_diamond} of the diamond-norm distance, we obtain from~\eqref{eq:diamond_tilde_phi_phi}
    \begin{align}
        &\lim_{n\to\infty}d_\diamond\qty(\Theta^{(n)}\qty[\tilde{\Phi}^{(n)}],\Theta^{(n)}\qty[\Phi^{\otimes n}])=0.
    \end{align}
    Then, Lemma~\ref{lem:upper_bound_log_robustness} shows that
    \begin{align}
    \label{eq:monotonicity_regularized_relative_entropy_of_resource2}
        &\limsup_{n\to\infty}\frac{1}{n}\log\qty[1+R_\mathrm{G}\qty(\Theta^{(n)}\qty[\tilde{\Phi}^{(n)}])]\notag\\
        &\geq \limsup_{n\to\infty}\frac{1}{n}R_\mathrm{R}\qty(\Theta^{(n)}\qty[\Phi^{(n)}]).
    \end{align}

    To show a relation between~\eqref{eq:monotonicity_regularized_relative_entropy_of_resource1} and~\eqref{eq:monotonicity_regularized_relative_entropy_of_resource2}, we will bound $R_\mathrm{G}\qty(\Theta^{(n)}\qty[\tilde{\Phi}^{(n)}])$ in terms of   \begin{align}
    \label{eq:R_n}
        R_n\coloneqq R_\mathrm{G}\qty(\tilde{\Phi}^{(n)}).
    \end{align}
    By the definition~\eqref{eq:R_G} of $R_\mathrm{G}$, there exists $\tilde{\Phi}^{(n)\prime}$ such that
    \begin{align}
        \frac{\tilde{\Phi}^{(n)}+R_n\tilde{\Phi}^{(n)\prime}}{1+R_n}\in\mathcal{F}.     
    \end{align}
    Then, Axiom~\ref{p:arng} implies that
    \begin{align}
        \epsilon_n\coloneqq R_\mathrm{G}\qty(\Theta^{(n)}\qty[\frac{\tilde{\Phi}^{(n)}+R_n\tilde{\Phi}^{(n)\prime}}{1+R_n}])
    \end{align}
    should satisfy
    \begin{align}
    \label{eq:epsilon_n}
        \lim_{n\to\infty}\epsilon_n=0
    \end{align}
    For this $\epsilon_n$, by the definition~\eqref{eq:R_G} of $R_\mathrm{G}$, there exists $\tilde{\Phi}^{(n)\prime\prime}$ such that
    \begin{align}
        \frac{\Theta^{(n)}\qty[\frac{\tilde{\Phi}^{(n)}+R_n\tilde{\Phi}^{(n)\prime}}{1+R_n}]+\epsilon_n\tilde{\Phi}^{(n)\prime\prime}}{1+\epsilon_n}\in\mathcal{F}.     
    \end{align}
    Due to the linearity of $\Theta^{(n)}$, we have
    \begin{align}
        \frac{\Theta^{(n)}\qty[\tilde{\Phi}^{(n)}]+R_n\Theta^{(n)}\qty[\tilde{\Phi}^{(n)\prime}]+\epsilon_n\qty(1+R_n)\tilde{\Phi}^{(n)\prime\prime}}{1+R_n+\epsilon_n\qty(1+R_n)}\in\mathcal{F}.     
    \end{align}
    Therefore, it holds that
    \begin{align}
        R_\mathrm{G}\qty(\Theta^{(n)}\qty[\tilde{\Phi}^{(n)}])&\leq R_n+\epsilon_n\qty(1+R_n),
    \end{align}
    which holds for all $n$.
    
    Consequently, we obtain
    \begin{align}
        &\limsup_{n\to\infty}\frac{1}{n}\log\qty(1+R_\mathrm{G}\qty(\Theta^{(n)}\qty[\tilde{\Phi}^{(n)}]))\notag\\
        &\leq\limsup_{n\to\infty}\frac{1}{n}\log\qty[1+R_n+\epsilon_n\qty(1+R_n)]\\
        \label{eq:monotonicity_regularized_relative_entropy_of_resource3}
        &=\lim_{n\to\infty}\frac{1}{n}\log\qty[1+R_\mathrm{G}\qty(\tilde{\Phi}^{(n)})],
    \end{align}
    where we use~\eqref{eq:R_n} and~\eqref{eq:epsilon_n}.
    Therefore, it follows from~\eqref{eq:monotonicity_regularized_relative_entropy_of_resource1},~\eqref{eq:monotonicity_regularized_relative_entropy_of_resource2}, and~\eqref{eq:monotonicity_regularized_relative_entropy_of_resource3} that
    \begin{align}
        R_\mathrm{R}^\infty\qty(\Phi)\geq \limsup_{n\to\infty}\frac{1}{n}R_\mathrm{R}\qty(\Theta^{(n)}\qty(\Phi^{\otimes n})).
    \end{align}
\end{proof}

Using this, we show the converse part of Theorem~\ref{thm:second_law} as follows.

\begin{proposition}[\label{prp:converse_second_law}The converse part of the second law for CQ channels]
For any family $\mathcal{F}$ of sets of free CQ channels satisfying Axioms~\ref{p:full_rank},~\ref{p:compact},~\ref{p:tensor_product}, and~\ref{p:convex},
any family $\tilde{\mathcal{O}}$ sequences of superchannels satisfying Axiom~\ref{p:arng},
and any CQ channels $\Phi_\mathrm{in}\in\mathcal{C}\qty(\mathcal{X}_\mathrm{in}\to\mathcal{H}_\mathrm{in})$ and $\Phi_\mathrm{out}\in\mathcal{C}\qty(\mathcal{X}_\mathrm{out}\to\mathcal{H}_\mathrm{out})$ ,
it holds that
\begin{align}
r_{\tilde{\mathcal{O}}}\qty(\Phi_{\mathrm{in}}\to\Phi_{\mathrm{out}}) R_\mathrm{R}^\infty\qty(\Phi_\mathrm{out})\leq R_\mathrm{R}^\infty\qty(\Phi_\mathrm{in}).
\end{align}
where $R_\mathrm{R}^\infty$ is defined as~\eqref{eq:regularized_relative_entropy_of_resource}, and $r_{\tilde{\mathcal{O}}}$ is defined as~\eqref{eq:conversion_rate}.
\end{proposition}

\begin{proof}
    By the definition~\eqref{eq:conversion_rate} of $r_{\tilde{\mathcal{O}}}$, 
    we take an arbitrary achievable rate
    \begin{align}
    \label{eq:converse_second_law_r}
        r<r_{\tilde{\mathcal{O}}}\qty(\Phi_{\mathrm{in}}\to\Phi_{\mathrm{out}})
    \end{align}
    such that we have a sequence $\qty{\Theta^{(n)}}_{n=1,2,\ldots}$ achieving
    \begin{align}
        \liminf_{n\to\infty}d_\diamond\qty(\Theta^{(n)}\qty[\Phi_\mathrm{in}^{\otimes n}],\Phi_\mathrm{out}^{\otimes\lceil rn\rceil})=0.
    \end{align}
    We have a subsequence $\qty{n_l}_{l=1,2\ldots}$ of $\qty{1,2,\ldots}$ such that
    \begin{align}
    \label{eq:converse_second_law_lim}
        \lim_{l\to\infty}d_\diamond\qty(\Theta^{\qty(n_l)}\qty[\Phi_\mathrm{in}^{\otimes n_l}],\Phi_\mathrm{out}^{\otimes\lceil rn_l\rceil})=0.
    \end{align}
    Then, under Axioms~\ref{p:full_rank},~\ref{p:compact},~\ref{p:tensor_product}, and~\ref{p:convex}, Lemma~\ref{lem:monotonicity_regularized_relative_entropy_of_resource} shows that
    \begin{align}
        R_\mathrm{R}^\infty\qty(\Phi_\mathrm{in})
        &\geq\limsup_{n\to\infty}\frac{1}{n}R_\mathrm{R}\qty(\Theta^{(n)}\qty[\Phi_\mathrm{in}^{\otimes n}])\\
    \label{eq:converse_second_law_1}
        &\geq\limsup_{l\to\infty}\frac{1}{n_l}R_\mathrm{R}\qty(\Theta^{\qty(n_l)}\qty[\Phi_\mathrm{in}^{\otimes n_l}]).
    \end{align}
    Within the subsequence, due to~\eqref{eq:converse_second_law_lim}, the asymptotic continuity~\eqref{eq:D_asymptotic_continuity} of $R_\mathrm{R}$ yields
    \begin{align}
    \label{eq:converse_second_law_2}
        \limsup_{l\to\infty}\frac{1}{n_l}R_\mathrm{R}\qty(\Theta^{\qty(n_l)}\qty[\Phi_\mathrm{in}^{\otimes n_l}])
        &=\lim_{n\to\infty}\frac{1}{n}R_\mathrm{R}\qty(\Phi_\mathrm{out}^{\otimes \lceil rn\rceil}),
    \end{align}
    where the limit on the right-hand side exists due to Lemma~\ref{prp:existence_limit}.
    From~\eqref{eq:converse_second_law_1} and~\eqref{eq:converse_second_law_2}, it follows that
    \begin{align}
        R_\mathrm{R}^\infty\qty(\Phi_\mathrm{in})
        &\geq\lim_{n\to\infty}\frac{1}{n}R_\mathrm{R}\qty(\Phi_\mathrm{out}^{\otimes \lceil rn\rceil})\\
        &=rR_\mathrm{R}^\infty\qty(\Phi_\mathrm{out}),
    \end{align}
    which holds for any achievable rate $r$.
    By taking the supremum of $r$, we obtain
    \begin{align}
        R_\mathrm{R}^\infty\qty(\Phi_\mathrm{in})\geq r_{\tilde{\mathcal{O}}}\qty(\Phi_{\mathrm{in}}\to\Phi_{\mathrm{out}}) R_\mathrm{R}^\infty\qty(\Phi_\mathrm{out}).
    \end{align}
\end{proof}

\subsubsection{Direct part}
\label{sec:direct_second_law}

In this section, we present the proof of the direct part of Theorem~\ref{thm:second_law} as follows.

\begin{proposition}[\label{prp:direct_second_law}The direct part of the second law for CQ channels]
For any family $\mathcal{F}$ of sets of free CQ channels satisfying Axioms~\ref{p:full_rank},~\ref{p:compact},~\ref{p:tensor_product}, and~\ref{p:convex},
any family $\tilde{\mathcal{O}}$ of sequences of superchannels satisfying Axiom~\ref{p:arng},
and any CQ channels $\Phi_\mathrm{in}\in\mathcal{C}\qty(\mathcal{X}_\mathrm{in}\to\mathcal{H}_\mathrm{in})$ and $\Phi_\mathrm{out}\in\mathcal{C}\qty(\mathcal{X}_\mathrm{out}\to\mathcal{H}_\mathrm{out})$ satisfying
\begin{align}
\label{eq:direct_second_law_assumption1}
    R_\mathrm{R}^\infty\qty(\Phi_\mathrm{in})&>0,\\
\label{eq:direct_second_law_assumption2}
    R_\mathrm{R}^\infty\qty(\Phi_\mathrm{out})&>0,
\end{align}
it holds that
\begin{align}
r_{\tilde{\mathcal{O}}}\qty(\Phi_{\mathrm{in}}\to\Phi_{\mathrm{out}})\geq\frac{R_\mathrm{R}^\infty\qty(\Phi_\mathrm{in})}{R_\mathrm{R}^\infty\qty(\Phi_\mathrm{out})}.
\end{align}
where $R_\mathrm{R}^\infty$ is defined as~\eqref{eq:regularized_relative_entropy_of_resource}, and $r_{\tilde{\mathcal{O}}}$ is defined as~\eqref{eq:conversion_rate}.
\end{proposition}

\begin{proof}
    Under the assumptions~\eqref{eq:direct_second_law_assumption1} and~\eqref{eq:direct_second_law_assumption1}, we choose any paramter
    \begin{align}
    \label{eq:direct_second_law_delta}
        \delta\in\qty(0,\min\qty{R_\mathrm{R}^\infty\qty(\Phi_\mathrm{in}),R_\mathrm{R}^\infty\qty(\Phi_\mathrm{out})}),
    \end{align}
    and set
    \begin{align}
    \label{eq:direct_second_law_r}
        r\coloneqq\frac{R_\mathrm{R}^\infty\qty(\Phi_\mathrm{in})-\delta}{R_\mathrm{R}^\infty\qty(\Phi_\mathrm{out})}.
    \end{align}
    We will construct a sequence $\qty{\Theta^{(n)}\in\mathcal{C}\qty(\qty(\mathcal{X}_\mathrm{in}^n\to\mathcal{H}_\mathrm{in}^{\otimes n})\to\qty(\mathcal{X}_\mathrm{out}^{n}\to\mathcal{H}_\mathrm{out}^{\otimes n}))}_{n=1,2,\ldots}\in\tilde{\mathcal{O}}$ achieving the rate $r$, i.e.,
    \begin{align}
    \label{eq:direct_second_law_conclusion}
        \liminf_{n\to\infty}d_\diamond\qty(\Theta^{(n)}\qty[\Phi_\mathrm{in}^{\otimes n}],\Phi_\mathrm{out}^{\otimes \lceil rn\rceil})=0,
    \end{align}
    where $d_\diamond$ is defined as~\eqref{eq:d_diamond}.
    In the following, we will first present a construction of $\qty{\Theta^{(n)}}_{n}$, then prove that $\qty{\Theta^{(n)}}_{n}$ satisfies Axiom~\ref{p:arng}, and finally show~\eqref{eq:direct_second_law_conclusion}.

    \textbf{Construction of $\qty{\Theta^{(n)}}_{n}$.}
    Applying the generalized quantum Stein's lemma for CQ channels in Theorem~\ref{thm:main}, we have a sequence $\qty{\epsilon_n\in\qty(0,1)}_{n=1,2,\ldots}$ of parameters satisfying
    \begin{align}
    \label{eq:direct_second_law_delta}
        &\lim_{n\to\infty}\epsilon_n=0,
    \end{align}
    a sequence $\qty{p_n}_{n=1,2,\ldots}$ of probability distributions over $\mathcal{X}_\mathrm{in}^n$, a sequence $\qty{\qty{T_{x_\mathrm{in}^{(n)}}}_{x_\mathrm{in}^{(n)}\in\mathcal{X}_\mathrm{in}^n}}_{n=1,2,\ldots}$ of families of POVM elements, and $n_0$ such that they achieve, for every $n\geq n_0$,
    \begin{align}
    \label{eq:direct_second_law_T_1}
        \alpha_n&\coloneqq\sum_{x_\mathrm{in}^{(n)}\in\mathcal{X}_\mathrm{in}^n}p_n\qty(x_\mathrm{in}^{(n)})\Tr\qty[\qty(\mathds{1}-T_{x_\mathrm{in}^{(n)}})\Phi_\mathrm{in}^{\otimes n}\qty(x_\mathrm{in}^{(n)})]\notag\\
        &\leq\epsilon_n,\\
    \label{eq:direct_second_law_T_2}
        \beta_n&\coloneqq\max_{\Phi_\mathrm{free}\in\mathcal{F}}\sum_{x_\mathrm{in}^{(n)}\in\mathcal{X}_\mathrm{in}^n}p_n\qty(x_\mathrm{in}^{(n)})\Tr\qty[T_{x_\mathrm{in}^{(n)}}\Phi_\mathrm{free}\qty(x_\mathrm{in}^{(n)})]\notag\\
        &\leq\exp[-n\qty(R_\mathrm{R}^\infty\qty(\Phi_\mathrm{in})-\frac{\delta}{3})],
    \end{align}
    where $\delta$ is the constant given by~\eqref{eq:direct_second_law_delta}.
    Due to Proposition~\ref{prp:robustness_characterization}, we have a sequence $\qty{\Phi_\mathrm{out}^{(rn)}}_{n=1,2,\ldots}$ of CQ channels satisfying
    \begin{align}
    \label{eq:direct_second_law_1}
        &R_\mathrm{R}^\infty\qty(\Phi_\mathrm{out})=\lim_{n\to\infty}\frac{\log\qty[1+R_\mathrm{G}\qty(\Phi_\mathrm{out}^{(rn)})]}{\lceil rn\rceil},\\
    \label{eq:direct_second_law_diamond_distance_limit}
        &\lim_{n\to\infty}d_\diamond\qty(\Phi_\mathrm{out}^{(rn)},\Phi_\mathrm{out}^{\otimes \lceil rn\rceil})=0.
    \end{align}
    By the definition~\eqref{eq:R_G} of $R_\mathrm{G}$, there exists a CQ channel $\Phi_\mathrm{out}^{(rn)\prime}$ achieving
    \begin{align}
    \label{eq:direct_second_law_phi_out_prime}
        \frac{\Phi_\mathrm{out}^{(rn)}+R_\mathrm{G}\qty(\Phi_\mathrm{out}^{(rn)})\Phi_\mathrm{out}^{(rn)\prime}}{1+R_\mathrm{G}\qty(\Phi_\mathrm{out}^{(rn)})}\in\mathcal{F}.
    \end{align}
    For every $n$, we define $\Theta^{(n)}$ as
    \begin{widetext}
    \begin{align}
    \label{eq:direct_second_law_theta}
        \Theta^{(n)}\qty[\Phi]\qty(x_\mathrm{out}^{(n)})
        &\coloneqq\qty(\sum_{x_\mathrm{in}^{(n)}\in\mathcal{X}_\mathrm{in}^n}p_n\qty(x_\mathrm{in}^{(n)})\Tr\qty[T_{x_\mathrm{in}^{(n)}}\Phi\qty(x_\mathrm{in}^{(n)})])\Phi_\mathrm{out}^{(rn)}\qty(x_\mathrm{out}^{(n)})+\notag\\
        &\quad\qty(\sum_{x_\mathrm{in}^{(n)}\in\mathcal{X}_\mathrm{in}^n}p_n\qty(x_\mathrm{in}^{(n)})\Tr\qty[\qty(\mathds{1}-T_{x_\mathrm{in}^{(n)}})\Phi\qty(x_\mathrm{in}^{(n)})])\Phi_\mathrm{out}^{(rn)\prime}\qty(x_\mathrm{out}^{(n)}),
    \end{align}
    \end{widetext}
    For any given input $x_\mathrm{out}^{(n)}\in\mathcal{X}_\mathrm{out}^n$, the CQ channel $\Theta^{(n)}\qty[\Phi]\in\mathcal{C}\qty(\mathcal{X}_\mathrm{out}^n\to\mathcal{H}_\mathrm{out}^{\otimes n})$ can be implemented by sampling an input $x_\mathrm{in}^{(n)}\in\mathcal{X}_\mathrm{in}^n$ according to the probability distribution $p_n$, inputting $x_\mathrm{in}^{(n)}$ to $\Phi$, measuring the corresponding output $\Phi\qty(x_\mathrm{in}^{(n)})$ by the POVM $\qty{T_{x^{(n)}},\mathds{1}-T_{x^{(n)}}}$, and, conditioned on each measurement outcome, outputting $\Phi_\mathrm{out}^{(rn)}\qty(x_\mathrm{out}^{(n)})$ and $\Phi_\mathrm{out}^{(rn)\prime}\qty(x_\mathrm{out}^{(n)})$, respectively, which shows that $\Theta^{(n)}$ is a superchannel as in~\eqref{eq:theta}.

    \textbf{Proof of Axiom~\ref{p:arng} for $\qty{\Theta^{(n)}}_{n}$.}
    For any sequence $\qty{\Phi_\mathrm{free}^{(n)}\in\mathcal{F}\qty(\mathcal{X}^n\to\mathcal{H}^{\otimes n})}_{n=1,2,\ldots}$, we will bound $R_\mathrm{G}\qty(\Theta^{(n)}\qty[\Phi_\mathrm{free}^{(n)}])$.
    For simplicity of notation, we write
    \begin{align}
        s_n&\coloneqq R_\mathrm{G}\qty(\Phi_\mathrm{out}^{(rn)}),\\
        t_n&\coloneqq \sum_{x_\mathrm{in}^{(n)}\in\mathcal{X}_\mathrm{in}^n}p_n\qty(x_\mathrm{in}^{(n)})\Tr\qty[T_{x_\mathrm{in}^{(n)}}\Phi_\mathrm{free}^{(n)}\qty(x_\mathrm{in}^{(n)})].
    \end{align}
    Due to~\eqref{eq:direct_second_law_1}, there exists $n_s$ such that, for every $n\geq n_s$,
    \begin{widetext}
    \begin{align}
    \label{eq:direct_second_law_limit_s}
        \exp[rn\qty(R_\mathrm{R}^\infty\qty(\Phi_\mathrm{out})+\frac{\delta}{3r})]-1\geq R_\mathrm{G}\qty(\Phi_\mathrm{out}^{(rn)})\geq\exp[rn\qty(R_\mathrm{R}^\infty\qty(\Phi_\mathrm{out})-\frac{\delta}{3r})]-1\to\infty~\text{as $n\to\infty$}.
    \end{align}
    \end{widetext}
    Due to~\eqref{eq:direct_second_law_T_2}, there exists $n_t$ such that for every $n\geq n_t$,
    \begin{align}
    \label{eq:direct_second_law_limit_t}
        t_n\leq\beta_n\leq\exp[-n\qty(R_\mathrm{R}^\infty\qty(\Phi_\mathrm{in})-\frac{\delta}{3})]\to 0~\text{as $n\to\infty$}.
    \end{align}
    Therefore, for every $n\geq\max\qty{n_s,n_t}$, it holds that
    \begin{align}
        &\frac{1}{1+s_n}-t_n\notag\\
        &\geq\exp[-rn\qty(R_\mathrm{R}^\infty\qty(\Phi_\mathrm{out})+\frac{\delta}{3r})]\notag\\
        &\quad-\exp[-n\qty(R_\mathrm{R}^\infty\qty(\Phi_\mathrm{in})-\frac{\delta}{3})]\\
        &\geq\exp[-n\qty(R_\mathrm{R}^\infty\qty(\Phi_\mathrm{in})-\delta+\frac{\delta}{3})]\notag\\
        &\quad-\exp[-n\qty(R_\mathrm{R}^\infty\qty(\Phi_\mathrm{in})-\frac{\delta}{3})]\\
        &\geq 0,
    \end{align}
    where we use the definition~\eqref{eq:direct_second_law_r} of $r$.
    Then, for every $n\geq\max\qty{n_s,n_t}$, we introduce a nonnegative parameter
    \begin{align}
    \label{eq:direct_second_law_r}
        r_n\coloneqq\frac{\frac{1}{1+s_n}-t_n}{\frac{s_n}{1+s_n}}\geq 0.
    \end{align}

    Using this notation, by the definition~\eqref{eq:direct_second_law_theta} of $\Theta^{(n)}$, we have
    \begin{align}
    \label{eq:direct_second_law_t}
        &\Theta^{(n)}\qty[\Phi_\mathrm{free}^{(n)}]=t_n\Phi_\mathrm{out}^{(rn)}+\qty(1-t_n)\Phi_\mathrm{out}^{(rn)\prime}.
    \end{align}
    On the other hand, it follows from~\eqref{eq:direct_second_law_phi_out_prime} that
    \begin{align}
    \label{eq:direct_second_law_s}
        \frac{\Phi_\mathrm{out}^{(rn)}+s_n\Phi_\mathrm{out}^{(rn)\prime}}{1+s_n}\in\mathcal{F}.
    \end{align}
    From~\eqref{eq:direct_second_law_r},~\eqref{eq:direct_second_law_t} and~\eqref{eq:direct_second_law_s}, we obtain, for every $n\geq\max\qty{n_s,n_t}$,
    \begin{align}
        &\frac{\Theta^{(n)}\qty[\Phi_\mathrm{free}^{(n)}]+r_n \Phi_\mathrm{out}^{(rn)}}{1+r_n}\notag\\
        &=\frac{t_n\Phi_\mathrm{out}^{(rn)}+\qty(1-t_n)\Phi_\mathrm{out}^{(rn)\prime}+\frac{\frac{1}{1+s_n}-t_n}{\frac{s_n}{1+s_n}}\Phi_\mathrm{out}^{(rn)}}{1+\frac{\frac{1}{1+s_n}-t_n}{\frac{s_n}{1+s_n}}}\\
        &=\frac{\Phi_\mathrm{out}^{(rn)}+s_n\Phi_\mathrm{out}^{(rn)\prime}}{1+s_n}\in\mathcal{F};
    \end{align}
    thus, by the definition~\eqref{eq:R_G} of $R_\mathrm{G}$, it holds that
    \begin{align}
        R_\mathrm{G}\qty(\Theta^{(n)}\qty[\Phi_\mathrm{free}^{(n)}])\leq r_n\to 0~\text{as $n\to\infty$},
    \end{align}
    where the limit vanishes due to~\eqref{eq:direct_second_law_limit_s},~\eqref{eq:direct_second_law_limit_t}, and~\eqref{eq:direct_second_law_r}, satisfying Axiom~\ref{p:arng}.

    \textbf{Proof of~\eqref{eq:direct_second_law_conclusion}.}
    By the definition~\eqref{eq:direct_second_law_theta} of $\Theta^{(n)}$, with $\alpha_n$ in~\eqref{eq:direct_second_law_T_1}, we have
    \begin{align}
        \Theta^{(n)}\qty[\Phi_\mathrm{in}^{\otimes n}]=\qty(1-\alpha_n)\Phi_\mathrm{out}^{(rn)}+\alpha_n\Phi_\mathrm{out}^{(rn)\prime}.
    \end{align}
    Thus, it holds that
    \begin{align}
        &d_\diamond\qty(\Theta^{(n)}\qty[\Phi_\mathrm{in}^{\otimes n}],\Phi_\mathrm{out}^{\otimes \lceil rn\rceil})\notag\\
        \label{eq:direct_second_law_conclusion_1}
        &\leq\qty(1-\alpha_n)d_\diamond\qty(\Phi_\mathrm{out}^{(rn)},\Phi_\mathrm{out}^{\otimes \lceil rn\rceil})+\alpha_n d_\diamond\qty(\Phi_\mathrm{out}^{(rn)\prime},\Phi_\mathrm{out}^{\otimes \lceil rn\rceil})\\
        \label{eq:direct_second_law_conclusion_2}
        &\leq d_\diamond\qty(\Phi_\mathrm{out}^{(rn)},\Phi_\mathrm{out}^{\otimes \lceil rn\rceil})+\epsilon_n d_\diamond\qty(\Phi_\mathrm{out}^{(rn)\prime},\Phi_\mathrm{out}^{\otimes \lceil rn\rceil}),
    \end{align}
    where~\eqref{eq:direct_second_law_conclusion_1} follows from the convexity~\eqref{eq:d_diamond_convexity} of $d_\diamond$, and~\eqref{eq:direct_second_law_conclusion_2} is due to~\eqref{eq:direct_second_law_T_1}.
    Therefore, due to~\eqref{eq:epsilon_n} and~\eqref{eq:direct_second_law_diamond_distance_limit}, we have
    \begin{align}
        &\lim_{n\to\infty}d_\diamond\qty(\Theta^{(n)}\qty[\Phi_\mathrm{in}^{\otimes n}],\Phi_\mathrm{out}^{\otimes \lceil rn\rceil})\notag\\
        &\leq\lim_{n\to\infty}d_\diamond\qty(\Phi_\mathrm{out}^{(rn)},\Phi_\mathrm{out}^{\otimes \lceil rn\rceil})+\lim_{n\to\infty}\epsilon_n d_\diamond\qty(\Phi_\mathrm{out}^{(rn)\prime},\Phi_\mathrm{out}^{\otimes \lceil rn\rceil})\\
        &=0,
    \end{align}
    which yields~\eqref{eq:direct_second_law_conclusion}.
\end{proof}

\section{Application to channel capacities}
\label{sec:application}

In this section, we apply the QRT framework for CQ channel conversion introduced above to the analysis of CQ channel capacities.
A previous work~\cite{hayashi2025generalizedquantumsteinslemma} also proposed a reversible QRT framework for CQ channel conversion and demonstrated its application to certain communication scenarios; however, their framework had limited applicability because it required asymptotic continuity of the allowed operations for CQ channel conversion. Due to this restriction, the reversible QRT framework in Ref.~\cite{hayashi2025generalizedquantumsteinslemma} was applicable only to communication scenarios with a fixed encoding scheme and could not be fully applied to the conventional channel-capacity settings that allow optimization over encoding schemes.
By contrast, our contribution lies in removing the requirement of asymptotic continuity, so that the only condition on the allowed operations is that they be asymptotically resource-non-generating. This refinement enables the resulting reversible QRT framework to be applied consistently to conventional channel capacity problems, as demonstrated here.
In Sec.~\ref{sec:operations}, we introduce a hierarchy of sets of operations for CQ channel conversion that underpins this analysis. In Sec.~\ref{sec:capacity}, we establish capacity bounds for CQ channels derived from the QRT framework using these sets of operations.

\subsection{Hierarchical relation of sets of operations for CQ channel converison}
\label{sec:operations}

In the spirit of Refs.~\cite{Takagi2020,hayashi2025generalizedquantumsteinslemma}, we here formulate the framework of QRTs for CQ dynamical resources to analyze the capacity of CQ channels.
To this end, we take the family $\mathcal{F}_\mathrm{R}$ of free CQ channels to be the set of replacers:
\begin{align}
\label{eq:F_R}
\mathcal{F}_\mathrm{R}\qty(\mathcal{X}\to\mathcal{H})\coloneqq\qty{\Phi_\rho:\rho\in\mathcal{D}\qty(\mathcal{H})},
\end{align}
where each $\Phi_\rho\in\mathcal{C}\qty(\mathcal{X}\to\mathcal{H})$ is defined by
\begin{align}
\Phi_\rho(x)=\rho \quad \text{for all $x\in\mathcal{X}$}.
\end{align}
By construction, every channel in $\mathcal{F}_\mathrm{R}$ has zero capacity.
A natural requirement for the operations allowed in channel coding, such as encoding and decoding, is that they should convert these free CQ channels to free CQ channels.
In what follows, we introduce a hierarchy of superchannels that represent such admissible classes of operations.

As in the definition~\eqref{eq:tilde_O} of $\tilde{\mathcal{O}}$, we define a family of asymptotically resource-non-generating operations for $\mathcal{F}_\mathrm{R}$ in~\eqref{eq:F_R} as
\begin{widetext}
\begin{align}
&\tilde{\mathcal{O}}_\mathrm{R}
\coloneqq
\qty{\qty{\Theta^{(n)}}_{n=1,2,\ldots}:\forall\qty{\Phi_\mathrm{free}^{(n)}\in\mathcal{F}}_{n=1,2,\ldots},~\lim_{n\to\infty} R_\mathrm{G}\qty(\Theta^{(n)}\qty[\Phi_\mathrm{free}^{(n)}])=0},
\end{align}    
\end{widetext}
where $R_\mathrm{G}$ is defined as~\eqref{eq:R_G}.
The family $\tilde{\mathcal{O}}_\mathrm{R}$ can be considered as a relaxation of a family $\mathcal{O}_\mathrm{R}$ of resource-non-generating operations for $\mathcal{F}_\mathrm{R}$ given by
\begin{align}
&\mathcal{O}_\mathrm{R}
\coloneqq
\qty{\Theta:\forall\Phi_\mathrm{free}\in\mathcal{F},~R_\mathrm{G}\qty(\Theta\qty[\Phi_\mathrm{free}])=0},
\end{align}
These operations are superchannels, and as in~\eqref{eq:theta}, we henceforth represent them as $\Theta\in\mathcal{C}\qty(\qty(\mathcal{X}_\mathrm{in}\to\mathcal{H}_\mathrm{in})\to\qty(\mathcal{X}_\mathrm{out}\to\mathcal{H}_\mathrm{out}))$ acting on a CQ channel $\Phi_\mathrm{in}\in\mathcal{C}\qty(\mathrm{X}_\mathrm{in}\to\mathcal{H}_\mathrm{in})$ as
\begin{align}
    \label{eq:Theta_application}
    &\qty(\Theta\qty[\Phi_\mathrm{in}])\qty(x_\mathrm{out})\notag\\
    &=\sum_{x_\mathrm{in}\in\mathcal{X}_\mathrm{in}}p_\Theta\qty(x_\mathrm{in}|x_\mathrm{out})\qty(\mathcal{N}_{\Theta,x_\mathrm{in},x_\mathrm{out}}\circ\Phi_\mathrm{in})\qty(x_\mathrm{in}).
\end{align}

Then, the following lemma characterize the property of $\mathcal{O}_\mathrm{R}$.
\begin{lemma}[Characterization of resource-non-generating operations by non-signaling condition]
    \label{lem:NS}
    For any superchannel $\Theta$ represented as~\eqref{eq:Theta_application} with $p_\Theta$ and $\mathcal{N}_{\Theta,x_\mathrm{in},x_\mathrm{out}}$, we consider the following conditions:
    \begin{enumerate}
    \item it holds that
    \begin{align}
        \label{eq:NS1}
        \Theta\in\mathcal{O}_\mathrm{R};
    \end{align}
    \item there exists a CPTP linear map $\mathcal{N}$ from $\mathcal{L}\qty(\mathcal{H}_\mathrm{in})$ to $\mathcal{L}\qty(\mathcal{H}_\mathrm{out})$ such that, for any $x_\mathrm{out}\in\mathcal{X}_\mathrm{out}$,
    \begin{align}
        \label{eq:NS}
        \mathcal{N}=\sum_{x_\mathrm{in}\in\mathcal{X}_\mathrm{in}}p_\Theta\qty(x_\mathrm{in}|x_\mathrm{out})\mathcal{N}_{\Theta,x_\mathrm{in},x_\mathrm{out}};
    \end{align}
    \item there exist two reference systems $\mathcal{H}_{\mathcal{X},R}$ and $\mathcal{H}_{R}$, an entangle state $\tilde{\rho}\in\mathcal{D}\qty(\mathcal{H}_{\mathcal{X},R}\otimes\mathcal{H}_{R})$, a POVM $\qty{\tilde{\Lambda}_{x_\mathrm{in}|x_\mathrm{out}}}_{x_\mathrm{in}}$ on $\mathcal{H}_{\mathcal{X},R}$, and a CPTP linear map $\tilde{\mathcal{N}}$ from $\mathcal{L}\qty(\mathcal{H}_{\mathrm{in}}\otimes\mathcal{H}_{R})$ to $\mathcal{L}\qty(\mathcal{H}_{\mathrm{out}})$ such that, for any $x_\mathrm{in}\in\mathcal{X}_\mathrm{in}$, $x_\mathrm{out}\in\mathcal{X}_\mathrm{out}$, and $\rho\in\mathcal{D}\qty(\mathcal{H}_{\mathrm{in}})$, we have relations
    \begin{align}
        \label{eq:ENS}
        p_\Theta\qty(x_\mathrm{in}|x_\mathrm{out})&=\Tr\qty[\qty(\tilde{\Lambda}_{x_\mathrm{in}|x_\mathrm{out}}\otimes\mathds{1})\tilde{\rho}],\notag\\
        \mathcal{N}_{\Theta,x_\mathrm{in},x_\mathrm{out}}\qty(\rho)&=\frac{\tilde{\mathcal{N}}\qty(\rho\otimes\Tr_{\mathcal{X},R}\qty[\qty(\tilde{\Lambda}_{x_\mathrm{in}|x_\mathrm{out}},\mathds{1})\tilde{\rho}])}{p_\Theta\qty(x_\mathrm{in}|x_\mathrm{out})},
    \end{align}
    where $\Tr_{\mathcal{X},R}$ is the partial trace over $\mathcal{H}_{\mathcal{X},R}$.
    \end{enumerate}
    The first and second conditions are equivalent. The third condition implies the first and second conditions.
\end{lemma}

The condition shown in~\eqref{eq:NS} corresponds to the non-signaling condition; hence, we let $\mathcal{O}_\mathrm{NS}$ denote the set of superchannels satisfying~\eqref{eq:NS}.
Lemma~\ref{lem:NS} guarantees the relation
\begin{align}
    \mathcal{O}_\mathrm{NS}=\mathcal{O}_\mathrm{R}.
\end{align}
Note that Ref.~\cite{IEEE-IT-8850073} also characterizes the non-signaling condition by using a semidefinite programming (SDP), where the condition~\eqref{eq:NS} corresponds to (6f) in Ref.~\cite{IEEE-IT-8850073}.
When the condition in~\eqref{eq:ENS} holds, the superchannel $\Theta$ can be implemented by the combination of shared entangled state $\tilde{\rho}$, a POVM measurement $\qty{\tilde{\Lambda}_{x_\mathrm{in}|x_\mathrm{out}}}_{x_\mathrm{in}}$, and the CPTP linear map $\tilde{\Phi}$; thus, we write the set of superchannels satisfying~\eqref{eq:ENS} as $\mathcal{O}_\mathrm{ENS}$.
Lemma~\ref{lem:NS} implies
\begin{align}
    \mathcal{O}_\mathrm{ENS}\subseteq\mathcal{O}_\mathrm{NS}.
\end{align}

\begin{proof}[Proof of Lemma~\ref{lem:NS}]
    Assume~\eqref{eq:NS1}.
    For an arbitrary state $\rho\in\mathcal{D}\qty(\mathcal{H}_\mathrm{out})$, the output of the CQ channel $\Theta\qty[\Phi_\rho]$, i.e.,
    \begin{align}
        \qty(\Theta\qty[\Phi_\rho])\qty(x_\mathrm{out})=\sum_{x_\mathrm{in}\in\mathcal{X}_\mathrm{in}}p_\Theta\qty(x_\mathrm{in}|x_\mathrm{out})\mathcal{N}_{\Theta,x_\mathrm{in},x_\mathrm{out}}\qty(\rho),
    \end{align}
    does not depend on the input $x_\mathrm{out}$.
    Thus, the map $\sum_{x_\mathrm{in}\in\mathcal{X}_\mathrm{in}}p_\Theta\qty(x_\mathrm{in}|x_\mathrm{out})\mathcal{N}_{\Theta,x_\mathrm{in},x_\mathrm{out}}$ does not depend on $x_\mathrm{out}$, which implies~\eqref{eq:NS}.

    Assume~\eqref{eq:NS}.
    For an arbitrary state $\rho\in\mathcal{D}\qty(\mathcal{H}_\mathrm{out})$, the output of the channel $\Theta\qty[\Phi_\rho]$ is given by
    \begin{align}
        \qty(\Theta\qty[\Phi_\rho])\qty(x_\mathrm{out})&=\sum_{x_\mathrm{in}\in\mathcal{X}_\mathrm{in}}p_\Theta\qty(x_\mathrm{in}|x_\mathrm{out})\mathcal{N}_{\Theta,x_\mathrm{in},x_\mathrm{out}}\qty(\rho)\\
        &=\mathcal{N}\qty(\rho),
    \end{align}
    which implies~\eqref{eq:NS1}.

    Finally, assume~\eqref{eq:ENS}.
    Then, we have
    \begin{align}
        &\sum_{x_\mathrm{in}\in\mathcal{X}_\mathrm{in}}p_\Theta\qty(x_\mathrm{in}|x_\mathrm{out})\mathcal{N}_{\Theta,x_\mathrm{in},x_\mathrm{out}}\notag\\
        &=\sum_{x_\mathrm{in}\in\mathcal{X}_\mathrm{in}}\tilde{\mathcal{N}}\qty(\rho\otimes\Tr_{\mathcal{X},R}\qty[\qty(\tilde{\Lambda}_{x_\mathrm{in}|x_\mathrm{out}}\otimes\mathds{1})\tilde{\rho}])\\
        &=\tilde{\mathcal{N}}\qty(\rho\otimes\Tr_{\mathcal{X},R}\qty[\qty(\mathds{1}\otimes\mathds{1})\tilde{\rho}])\\
        &=\tilde{\mathcal{N}}\qty(\rho\otimes\Tr_{\mathcal{X},R}\qty[\tilde{\rho}]),
    \end{align}
    which implies~\eqref{eq:NS}.
\end{proof}

Now, we consider a subset of $\mathcal{O}_\mathrm{ENS}$\@.
In particular, a superchannel $\Theta\in\mathcal{O}_\mathrm{ENS}$ is called classically correlated when there exists a probability distribution $\tilde{p}(s)$ representing shared randomness satisfying the following condition: there exists a pair of a conditional distribution $p_\Theta(x_\mathrm{in}|x_\mathrm{out}, s)$ and a CPTP linear map $\tilde{\mathcal{N}}_{x_\mathrm{in},x_\mathrm{out},s}$ such that
\begin{align}
    &\qty(\Theta\qty[\Phi_\mathrm{in}])\qty(x_\mathrm{out})\notag\\
    &=\sum_{x_\mathrm{in}\in\mathcal{X}_\mathrm{in}}\sum_s\tilde{p}\qty(s)p_\Theta(x_\mathrm{in}|x_\mathrm{out}, s)\qty(\tilde{\mathcal{N}}_{x_\mathrm{in},x_\mathrm{out},s}\circ\Phi_\mathrm{in})\qty(x_\mathrm{in}).
\end{align}
We let $\mathcal{O}_\mathrm{CNS}$ denote the set of these superchannels.
Since $\mathcal{O}_\mathrm{CNS}$ is a special case of $\mathcal{O}_\mathrm{ENS}$ restricted to using the classical shared randomness obtained by measuring the shared entanglement in the standard basis, we have
\begin{align}
    \mathcal{O}_\mathrm{CNS}\subset\mathcal{O}_\mathrm{ENS}.
\end{align}

Next, we consider a subset of $\mathcal{O}_\mathrm{CNS}$\@.
In particular, a superchannel $\Theta\in\mathcal{O}_\mathrm{CNS}$ is called deterministic if there exists a pair of a conditional distribution $p_\Theta(x_\mathrm{in}|x_\mathrm{out})$ and a CPTP linear map $\tilde{\mathcal{N}}_{x_\mathrm{in},x_\mathrm{out}}$ such that
\begin{align}
    &\qty(\Theta\qty[\Phi_\mathrm{in}])\qty(x_\mathrm{out})\notag\\
    &=\sum_{x_\mathrm{in}\in\mathcal{X}_\mathrm{in}}p_\Theta(x_\mathrm{in}|x_\mathrm{out})\qty(\tilde{\mathcal{N}}_{x_\mathrm{in},x_\mathrm{out}}\circ\Phi_\mathrm{in})\qty(x_\mathrm{in}).
\end{align}
We let $\mathcal{O}_\mathrm{DNS}$ denote the set of these superchannels.
Since $\mathcal{O}_\mathrm{DNS}$ is a special case of $\mathcal{O}_\mathrm{CNS}$ restricted to using $\tilde{p}(s)=1$ over a single-element set, we have
\begin{align}
    \mathcal{O}_\mathrm{DNS}\subset\mathcal{O}_\mathrm{CNS}.
\end{align}

Overall, we have the inclusion relation
\begin{align}
    \label{eq:inclusion}
    \mathcal{O}_\mathrm{DNS}\subset\mathcal{O}_\mathrm{CNS}\subset\mathcal{O}_\mathrm{ENS}\subset\mathcal{O}_\mathrm{NS}=\mathcal{O}_\mathrm{R},
\end{align}
and $\tilde{O}_\mathrm{R}$ is a relaxation of these sets of superchannels.

\subsection{Channel capacity from CQ channel conversion}
\label{sec:capacity}

In this section, we analyze bounds of the capacity and the conversion rate of CQ channels using the sets of superchannels introduced above.
Given a CQ channel $\Phi\in\mathcal{C}\qty(\mathcal{X}\to\mathcal{H})$, the relative entropy of resource with respect to $\mathcal{F}_\mathrm{R}$ in~\eqref{eq:F_R} is calculated as
\begin{align}
    R_\mathrm{R}\qty(\Phi)
    &=\min_{\Phi_\mathrm{free}\in\mathcal{F}_\mathrm{R}}D\left(\Phi\middle\|\Phi_\mathrm{free}\right)\\
    &=\min_{\rho\in\mathcal{D}\qty(\mathcal{H})}D\left(\Phi\middle\|\Phi_\rho\right)\\
    &=\min_{\rho\in\mathcal{D}\qty(\mathcal{H})}\max_{x\in\mathcal{X}}D\left(\Phi\qty(x)\middle\|\rho\right)\\
    &=\max_{p}\sum_{x\in\mathcal{X}}p(x)D\left(\Phi\qty(x)\middle\|\sum_{x'\in\mathcal{X}}p\qty(x')\Phi\qty(x')\right)
\end{align}
which is the same as the channel capacity of the CQ channel $\Phi$, denoted by $C[\Phi]$.
This quantity satisfies the additivity
\begin{align}
    R_\mathrm{R}\qty(\Phi\otimes\Phi')=R_\mathrm{R}\qty(\Phi)+R_\mathrm{R}\qty(\Phi'),
\end{align}
and hence, the regularized relative entropy of resource coincides with
\begin{align}
    R_\mathrm{R}^\infty\qty(\Phi)=R_\mathrm{R}\qty(\Phi).
\end{align}
For the conversion rates defined in~\eqref{eq:conversion_rate_O} and~\eqref{eq:conversion_rate}, due to the inclusion relation shown in~\eqref{eq:inclusion}, we have
\begin{align}
    \label{eq:r_relation}
    &r_{\mathcal{O}_\mathrm{DNS}}\qty(\Phi_\mathrm{in}\to\Phi_\mathrm{out})\notag\\
    &\leq r_{\mathcal{O}_\mathrm{CNS}}\qty(\Phi_\mathrm{in}\to\Phi_\mathrm{out})\notag\\
    &\leq r_{\mathcal{O}_\mathrm{ENS}}\qty(\Phi_\mathrm{in}\to\Phi_\mathrm{out})\notag\\
    &\leq r_{\mathcal{O}_\mathrm{NS}}\qty(\Phi_\mathrm{in}\to\Phi_\mathrm{out})= r_{\mathcal{O}_\mathrm{R}}\qty(\Phi_\mathrm{in}\to\Phi_\mathrm{out})\notag\\
    &\leq r_{\tilde{\mathcal{O}}_\mathrm{R}}\qty(\Phi_\mathrm{in}\to\Phi_\mathrm{out})=\frac{C[\Phi_\mathrm{in}]}{C[\Phi_\mathrm{out}]},
\end{align}
where the last line follows from Theorem~\ref{thm:second_law} when $\Phi_\mathrm{in},\Phi_\mathrm{out}\not\in\mathcal{F}_\mathrm{R}$.
In the following, we will check that the equality holds in some of the above inequalities in various scenarions of channel coding.

First, to analyze the channel coding, we consider the simulation of a noiseless CQ channel
\begin{align}
\Phi_\mathrm{noiseless}(x)=\ket{x}\bra{x}~\text{with input $x\in\qty{0,1}$}
\end{align}
by using a given CQ channel $\Phi\in\mathcal{C}\qty(\mathcal{X}\to\mathcal{H})$.
In this scenario, the quantity $d_\diamond\qty(\Theta[\Phi^{\otimes n}],\Phi_\mathrm{noiseless}^{\lceil rn\rceil})$ appearing in the definitions~\eqref{eq:conversion_rate_O} and~\eqref{eq:conversion_rate} of the conversion rates represents the maximum decoding error probability.
Hence, the CQ channel coding theorem~\cite{IEEE-IT-651037,PhysRevA.56.131} means the relation
\begin{align}
    r_{\mathcal{O}_\mathrm{DNS}}\qty(\Phi\to\Phi_\mathrm{noiseless})=C[\Phi].
\end{align}
More recently, Ref.~\cite{OTB2024} showed the relation
\begin{align}
    r_{\mathcal{O}_\mathrm{NS}}\qty(\Phi\to\Phi_\mathrm{noiseless})=C[\Phi].
\end{align}

On the other hand, Refs.~\cite[Theorem~3]{6757002} and~\cite[Theorem~4.3]{berta2011quantum} consider the problem of converting from $\Phi_\mathrm{noiseless}$ to $\Phi$ with a shared entangled state between the sender and the receiver; this type of result is called a quantum reverse Shannon theorem.
As explained in Appendix~\ref{sec:entanglement_assist}, they essentially showed that
\begin{align}
    \label{eq:bound_ENS}
    \frac{1}{r_{\mathcal{O}_\mathrm{ENS}}\qty(\Phi_\mathrm{noiseless}\to\Phi)}=C[\Phi].
\end{align}
Therefore, any two CQ channels $\Phi_\mathrm{in}$ and $\Phi_\mathrm{out}$ satisfy
\begin{align}
    r_{\mathcal{O}_\mathrm{ENS}}\qty(\Phi_\mathrm{in}\to\Phi_\mathrm{out})=\frac{C[\Phi_\mathrm{in}]}{C[\Phi_\mathrm{out}]},
\end{align}
which implies the equality in the third inequality of~\eqref{eq:r_relation}.

In addition, when the states in the set $\qty{\Phi(x)}_{x\in\mathcal{X}}$ are commutative with each other, Ref.~\cite{1035117} showed the relation
\begin{align}
    \frac{1}{r_{\mathcal{O}_\mathrm{CNS}}\qty(\Phi_\mathrm{noiseless}\to\Phi)}=C[\Phi];
\end{align}
therefore, such two CQ channels $\Phi_\mathrm{in}$ and $\Phi_\mathrm{out}$ satisfy
\begin{align}
    r_{\mathcal{O}_\mathrm{CNS}}\qty(\Phi_\mathrm{in}\to\Phi_\mathrm{out})=\frac{C[\Phi_\mathrm{in}]}{C[\Phi_\mathrm{out}]}.
\end{align}
In this case, the equality holds in the second inequality of~\eqref{eq:r_relation}.

\section{Conclusion}
\label{sec:conclusion}

In this work, we have formulated and proved a generalized quantum Stein's lemma for CQ channels, characterizing the optimal error exponent in hypothesis testing for distinguishing IID copies of a CQ channel from a non-IID set of free CQ channels.
This result extends the generalized quantum Stein's lemma from the state setting~\cite{Brand_o_2008,brandao2010reversible,Brandao2010,hayashi2025generalizedquantumsteinslemma,10898013} to the fundamental class of dynamical resources represented by CQ channels.
A key technical contribution was the development of CQ-channel counterparts of proof techniques originally devised for the state version of the lemma in Ref.~\cite{hayashi2025generalizedquantumsteinslemma}, including the pinching technique~\cite{hayashi2002optimal} and the information spectrum method~\cite{4069150}, as well as error-exponent bounds based on R\'{e}nyi relative entropies~\cite{887855,cooney2016strong}.
These tools address the nontrivial challenge arising from the presence of multiple possible inputs in CQ channels, and they enable channel discrimination tasks to be analyzed directly using quantities defined for channels, rather than relying on reductions to the state case.

Furthermore, using this CQ-channel version of the generalized quantum Stein's lemma, we construct a reversible QRT framework for CQ channel conversion.
Conceptually, our key advance is the removal of the asymptotic continuity requirement imposed in the earlier framework of Ref.~\cite{hayashi2025generalizedquantumsteinslemma}, demonstrating that the asymptotic resource-non-generating property of free operations alone suffices, as in the reversible QRT framework for static resources originally proposed in Refs.\cite{Brand_o_2008,brandao2010reversible,Brandao2010,Brandao2015}.
This refinement is significant because it broadens the applicability of the reversible framework to conventional channel coding scenarios, where input optimization is essential but typically violates asymptotic continuity.
As we have shown, our framework can now be applied to the analysis of channel capacities, thereby covering a key application domain of QRTs for dynamical resources~\cite{Takagi2020} that was largely inaccessible in the previous approach~\cite{hayashi2025generalizedquantumsteinslemma}.

While extending these results to QQ channels remains a challenging open problem, our findings highlight the importance of focusing on CQ channels to establish a tractable framework.
CQ channels provide a natural bridge between static resources and the full generality of dynamical resources, encompassing classical channels as special cases and reducing to states when the channels have a single input.
Looking forward, our framework establishes a tractable and conceptually robust foundation for analyzing quantum information processing tasks based on dynamical resources.
Beyond its theoretical contributions, the techniques developed here are also expected to serve as practical analytical tools, from the design of efficient coding strategies in quantum communication to principled resource-theoretic benchmarks for operations in quantum devices, where understanding and manipulating dynamical resources are of central importance.

\begin{acknowledgments}
M.H.\ was supported in part by the National Natural Science Foundation of China under Grant 62171212, and the General R\&D Projects of 1+1+1 CUHK-CUHK(SZ)-GDST Joint Collaboration Fund  (Grant No. GRDP2025-022)\@.
H.Y.\ was supported by JST PRESTO Grant Number JPMJPR23FC\@.
\end{acknowledgments}

\newpage
\appendix

\section{CQ-channel conversion rate under entanglement-assisted non-signaling operations}
\label{sec:entanglement_assist}

Here, we explain how to derive~\eqref{eq:bound_ENS} from Refs.~\cite[Theorem~3]{6757002} and~\cite[Theorem~3.10]{berta2011quantum}.
For this aim, we consider the conversion for QQ channels given in these references.

We consider the following conversion from a QQ channel
from $\mathcal{L}\qty(\mathcal{H}_{\mathrm{in},\mathcal{X}})$ to $\mathcal{L}\qty(\mathcal{H}_{\mathrm{in}})$ to a QQ channel from $\mathcal{L}\qty(\mathcal{H}_{\mathrm{out},\mathcal{X}})$ to $\mathcal{L}\qty(\mathcal{H}_{\mathrm{out}})$.
We choose two reference systems $\mathcal{H}_{\mathcal{X},R}$ and
$\mathcal{H}_{R}$, an entangled state $\tilde{\rho}\in\mathcal{D}\qty(\mathcal{H}_{\mathcal{X},R}\otimes\mathcal{H}_{R})$, a CPTP linear map $\tilde{\mathcal{N}}_\mathcal{X}$ from $\mathcal{L}\qty(\mathcal{H}_{\mathrm{out},\mathcal{X}}\otimes\mathcal{H}_{\mathcal{X},R})$ to $\mathcal{L}\qty(\mathcal{H}_{\mathrm{in},\mathcal{X}})$, and a CPTP linear map $\tilde{\mathcal{N}}$ from $\mathcal{L}\qty(\mathcal{H}_{\mathrm{in}}\otimes\mathcal{H}_{R})$ to $\mathcal{L}\qty(\mathcal{H}_{\mathrm{out}})$.
Then, we define a superchannel $\Theta_{\tilde{\mathcal{N}}_\mathcal{X},\tilde{\mathcal{N}},\tilde{\rho}}$ for QQ channel conversion as follows: for any QQ channel $\mathcal{N}$ from $\mathcal{L}\qty(\mathcal{H}_{\mathrm{in},\mathcal{X}})$ to $\mathcal{L}\qty(\mathcal{H}_\mathrm{in})$ and a state $\rho\in\mathcal{D}\qty(\mathcal{H}_{\mathrm{out},\mathcal{X}})$, the superchannel $\Theta_{\tilde{\mathcal{N}}_\mathcal{X},\tilde{\mathcal{N}},\tilde{\rho}}$ acts as
\begin{align}
    &\qty(\Theta_{\tilde{\mathcal{N}}_\mathcal{X},\tilde{\mathcal{N}},\tilde{\rho}}\qty[\mathcal{N}])\qty(\rho)\notag\\
    &\coloneqq\tilde{\mathcal{N}}\circ\qty(\mathcal{N}\otimes\id_{\mathrm{out},R})\circ\qty(\tilde{\mathcal{N}}_\mathcal{X}\otimes\id_{\mathrm{out},R})\qty(\rho\otimes\tilde{\rho}).
\end{align}
We also define a POVM $\tilde{\mathcal{M}}=\qty{\tilde{\Lambda}_{x_\mathrm{in}|x_\mathrm{out}}}_{x_\mathrm{in}}$ acting on $\mathcal{H}_{\mathcal{X},R}$ as
\begin{align}
    \tilde{\Lambda}_{x_\mathrm{in}|x_\mathrm{out}}\coloneqq\Tr_{\mathrm{out},\mathcal{X}}\qty[\tilde{\mathcal{N}}_\mathcal{X}^\ast\qty(\ket{x_\mathrm{in}}\bra{x_\mathrm{in}})\qty(\ket{x_\mathrm{out}}\bra{x_\mathrm{out}}\otimes\mathds{1})],
\end{align}
where $\tilde{\mathcal{N}}_\mathcal{X}^\ast$ is the dual map of $\tilde{\mathcal{N}}_\mathcal{X}$, and $\Tr_{\mathrm{out},\mathcal{X}}$ is the partial trace over $\mathcal{H}_{\mathrm{out},\mathcal{X}}$.
Then, we define a superchannel $\Theta_{\tilde{\mathcal{M}},\tilde{\mathcal{N}},\tilde{\rho}}$ for CQ channel conversion as
\begin{align}
     &\qty(\Theta_{\tilde{\mathcal{M}},\tilde{\mathcal{N}},\tilde{\rho}}\qty[\Phi])\qty(x_\mathrm{out})\notag\\
    &\coloneqq\tilde{\mathcal{N}}\qty(\sum_{x_\mathrm{in}\in\mathcal{X}_\mathrm{in}}\Phi\qty(x_\mathrm{in})\otimes\Tr_{\mathcal{X},R}\qty[\qty(\tilde{\Lambda}_{x_\mathrm{in}|x_\mathrm{out}}\otimes\mathds{1})\tilde{\rho}]),
\end{align}
where $\Tr_{\mathcal{X},R}$ is the partial trace over $\mathcal{H}_{\mathcal{X},R}$.

Given a CQ channel $\Phi$ from $\mathcal{X}$ to $\mathcal{D}\qty(\mathcal{H})$, with $\mathcal{H}_\mathcal{X}$ denoting a space satisfying $\dim\qty(\mathcal{H}_\mathcal{X})=|\mathcal{X}|$, we define the QQ channel $\mathcal{N}_\Phi$ from $\mathcal{L}\qty(\mathcal{H}_\mathcal{X})$ to $\mathcal{L}\qty(\mathcal{H})$ as
\begin{align}
    \label{eq:QQ_channel}
    \mathcal{N}_\Phi\qty(\rho)\coloneqq\sum_{x\in\mathcal{X}}\bra{x}\rho\ket{x}\ket{x}\bra{x}\otimes\Phi\qty(x).
\end{align}
Then, we have
\begin{align}
    \label{eq:equivalence_relation}
\qty(\Theta_{\tilde{\mathcal{N}}_\mathcal{X},\tilde{\mathcal{N}},\tilde{\rho}}\qty[\mathcal{N}_\Phi])\qty(\ket{x_\mathrm{out}}\bra{x_\mathrm{out}})=\qty(\Theta_{\tilde{\mathcal{M}},\tilde{\mathcal{N}},\tilde{\rho}}\qty[\Phi])\qty(x_\mathrm{out}).
\end{align}

We consider a CQ channel $\Phi_\mathrm{in}\in\mathcal{C}\qty(\mathcal{X}_\mathrm{in}\to\mathcal{H}_\mathrm{in})$ and $\Phi_\mathrm{out}\in\mathcal{C}\qty(\mathcal{X}_\mathrm{out}\to\mathcal{H}_\mathrm{out})$.
As in~\eqref{eq:QQ_channel}, we correspondingly have QQ channels $\mathcal{N}_{\Phi_\mathrm{in}}$ from $\mathcal{L}\qty(\mathcal{H}_{\mathrm{in},\mathcal{X}})$ to $\mathcal{L}\qty(\mathcal{H}_\mathrm{in})$ and $\mathcal{N}_{\Phi_\mathrm{out}}$ from $\mathcal{L}\qty(\mathcal{H}_{\mathrm{out},\mathcal{X}})$ to $\mathcal{L}\qty(\mathcal{H}_\mathrm{out})$.
Let $\mathcal{H}_{\mathrm{out},\mathcal{X},R}$ denote a reference system satisfying $\dim\qty(\mathcal{H}_{\mathrm{out},\mathcal{X},R})=\dim\qty(\mathcal{H}_{\mathrm{out},\mathcal{X}})$.
Then, the relation~\eqref{eq:equivalence_relation} yields
\begin{widetext}
\begin{align}
    \label{eq:error_measure}
    \max_{\rho\in\mathcal{D}\qty(\mathcal{H}_{\mathrm{out},\mathcal{X}}\otimes\mathcal{H}_{\mathrm{out},\mathcal{X},R})}d_{\mathrm{T}}\qty(\qty(\Theta_{\tilde{\mathcal{N}}_\mathcal{X},\tilde{\mathcal{N}},\tilde{\rho}}\qty[\mathcal{N}_{\Phi_\mathrm{in}}]\otimes\id)\qty(\rho),\qty(\mathcal{N}_{\Phi_\mathrm{out}}\otimes\id)\qty(\rho))
    &=\max_{x_\mathrm{out}\in\mathcal{X}_\mathrm{out}}d_{\mathrm{T}}\qty(\Theta_{\tilde{\mathcal{M}},\tilde{\mathcal{N}},\tilde{\rho}}\qty[\Phi_\mathrm{in}]\qty(x_\mathrm{out}),\Phi_\mathrm{out}\qty(x_\mathrm{out}))\notag\\
     &=d_{\diamond}\qty(\Theta_{\tilde{\mathcal{M}},\tilde{\mathcal{N}},\tilde{\rho}}\qty[\Phi_\mathrm{in}],\Phi_\mathrm{out}),
\end{align}    
\end{widetext}
where $d_{\mathrm{T}}$ and $d_\diamond$ are defined as~\eqref{eq:d_trace} and~\eqref{eq:d_diamond}, respectively.

We now analyze the conversion rate from $\Phi_\mathrm{noiseless}$ to a CQ channel $\Phi\in\mathcal{C}\qty(\mathcal{X}\to\mathcal{H})$.
As in~\eqref{eq:equivalence_relation}, we correspondingly have a QQ channel $\mathcal{N}_\Phi$ from $\mathcal{L}\qty(\mathcal{H}_\mathcal{X})$ to $\mathcal{L}\qty(\mathcal{H})$, where $\mathcal{H}_\mathcal{X}=|\mathcal{X}|$.
We introduce a reference system $\mathcal{H}_{\mathcal{X},R}$ satisfying $\dim\qty(\mathcal{H}_{\mathcal{X},R})=\dim\qty(\mathcal{H}_\mathcal{X})$.
Let $\mathcal{D}_\mathrm{pure}$ denote the set of density operators for pure states of the system $\mathcal{H}_\mathcal{X},\otimes\mathcal{H}_{\mathcal{X},R}$.
We further write its subset
\begin{align}
    &\mathcal{D}_\mathrm{pure}'\notag\\
    &\coloneqq\qty{\ket{\psi_p}\bra{\psi_p}\in\mathcal{D}_\mathrm{pure}:\ket{\psi_p}=\sum_{x\in\mathcal{X}}\sqrt{p(x)}\ket{x}\otimes\ket{x}}.
\end{align} 
When the error is measured by~\eqref{eq:error_measure}, Ref.~\cite[Theorem~3.10]{berta2011quantum} shows that the following rate is achieved:
\begin{widetext}
\begin{align}
    &\max_{\ket{\psi}\bra{\psi}\in\mathcal{D}_\mathrm{pure}}D\left(\qty(\mathcal{N}_\Phi\otimes\id)\qty(\ket{\psi}\bra{\psi})\middle\|\mathcal{N}_\Phi\qty(\Tr_{\mathcal{X},R}\qty[\ket{\psi}\bra{\psi}])\otimes\Tr_{\mathcal{X}}\qty[\ket{\psi}\bra{\psi}]\right)\notag\\
    &=\max_{\ket{\psi_p}\bra{\psi_p}\in\mathcal{D}_\mathrm{pure}'}D\left(\qty(\mathcal{N}_\Phi\otimes\id)\qty(\ket{\psi_p}\bra{\psi_p})\middle\|\mathcal{N}_\Phi\qty(\Tr_{\mathcal{X},R}\qty[\ket{\psi_p}\bra{\psi_p}])\otimes\Tr_{\mathcal{X}}\qty[\ket{\psi_p}\bra{\psi_p}]\right)\\
    &=\max_{p}\sum_{x\in\mathcal{X}}p\qty(x)D\left(\Phi\qty(x)\middle\|\sum_{x'\in\mathcal{X}}p\qty(x')\Phi\qty(x')\right)\\
    &=C\qty[\Phi],
\end{align}
\end{widetext}
and then,~\eqref{eq:error_measure} guarantees~\eqref{eq:bound_ENS}, i.e., the achievability of $C\qty[\Phi]$ in the CQ channel coding.
Note that the proof of Ref.~\cite[Theorem~3]{6757002} also essentially shows the same fact.

\bibliography{citation}

\begin{thebibliography}{63}%
\makeatletter
\providecommand \@ifxundefined [1]{%
 \@ifx{#1\undefined}
}%
\providecommand \@ifnum [1]{%
 \ifnum #1\expandafter \@firstoftwo
 \else \expandafter \@secondoftwo
 \fi
}%
\providecommand \@ifx [1]{%
 \ifx #1\expandafter \@firstoftwo
 \else \expandafter \@secondoftwo
 \fi
}%
\providecommand \natexlab [1]{#1}%
\providecommand \enquote  [1]{``#1''}%
\providecommand \bibnamefont  [1]{#1}%
\providecommand \bibfnamefont [1]{#1}%
\providecommand \citenamefont [1]{#1}%
\providecommand \href@noop [0]{\@secondoftwo}%
\providecommand \href [0]{\begingroup \@sanitize@url \@href}%
\providecommand \@href[1]{\@@startlink{#1}\@@href}%
\providecommand \@@href[1]{\endgroup#1\@@endlink}%
\providecommand \@sanitize@url [0]{\catcode `\\12\catcode `\$12\catcode `\&12\catcode `\#12\catcode `\^12\catcode `\_12\catcode `\%12\relax}%
\providecommand \@@startlink[1]{}%
\providecommand \@@endlink[0]{}%
\providecommand \url  [0]{\begingroup\@sanitize@url \@url }%
\providecommand \@url [1]{\endgroup\@href {#1}{\urlprefix }}%
\providecommand \urlprefix  [0]{URL }%
\providecommand \Eprint [0]{\href }%
\providecommand \doibase [0]{https://doi.org/}%
\providecommand \selectlanguage [0]{\@gobble}%
\providecommand \bibinfo  [0]{\@secondoftwo}%
\providecommand \bibfield  [0]{\@secondoftwo}%
\providecommand \translation [1]{[#1]}%
\providecommand \BibitemOpen [0]{}%
\providecommand \bibitemStop [0]{}%
\providecommand \bibitemNoStop [0]{.\EOS\space}%
\providecommand \EOS [0]{\spacefactor3000\relax}%
\providecommand \BibitemShut  [1]{\csname bibitem#1\endcsname}%
\let\auto@bib@innerbib\@empty
\bibitem [{\citenamefont {Cover}\ and\ \citenamefont {Thomas}(2012)}]{cover2012elements}%
  \BibitemOpen
  \bibfield  {author} {\bibinfo {author} {\bibfnamefont {T.}~\bibnamefont {Cover}}\ and\ \bibinfo {author} {\bibfnamefont {J.}~\bibnamefont {Thomas}},\ }\href {https://www.wiley.com/en-us/Elements+of+Information+Theory%2C+2nd+Edition-p-9781118585771} {\emph {\bibinfo {title} {Elements of Information Theory}}}\ (\bibinfo  {publisher} {Wiley},\ \bibinfo {year} {2012})\BibitemShut {NoStop}%
\bibitem [{\citenamefont {Shannon}(1948)}]{6773024}%
  \BibitemOpen
  \bibfield  {author} {\bibinfo {author} {\bibfnamefont {C.~E.}\ \bibnamefont {Shannon}},\ }\bibfield  {title} {\bibinfo {title} {A mathematical theory of communication},\ }\href {https://doi.org/10.1002/j.1538-7305.1948.tb01338.x} {\bibfield  {journal} {\bibinfo  {journal} {The Bell System Technical Journal}\ }\textbf {\bibinfo {volume} {27}},\ \bibinfo {pages} {379} (\bibinfo {year} {1948})}\BibitemShut {NoStop}%
\bibitem [{\citenamefont {Hayashi}(2016)}]{hayashi2016quantum}%
  \BibitemOpen
  \bibfield  {author} {\bibinfo {author} {\bibfnamefont {M.}~\bibnamefont {Hayashi}},\ }\href {https://link.springer.com/book/10.1007/978-3-662-49725-8#bibliographic-information} {\emph {\bibinfo {title} {Quantum information theory: Mathematical Foundation}}}\ (\bibinfo  {publisher} {Springer},\ \bibinfo {year} {2016})\BibitemShut {NoStop}%
\bibitem [{\citenamefont {Holevo}(2019)}]{Holevo+2019}%
  \BibitemOpen
  \bibfield  {author} {\bibinfo {author} {\bibfnamefont {A.~S.}\ \bibnamefont {Holevo}},\ }\href {https://doi.org/doi:10.1515/9783110642490} {\emph {\bibinfo {title} {Quantum Systems, Channels, Information}}}\ (\bibinfo  {publisher} {De Gruyter},\ \bibinfo {address} {Berlin, Boston},\ \bibinfo {year} {2019})\BibitemShut {NoStop}%
\bibitem [{\citenamefont {Watrous}(2018)}]{watrous_2018}%
  \BibitemOpen
  \bibfield  {author} {\bibinfo {author} {\bibfnamefont {J.}~\bibnamefont {Watrous}},\ }\href {https://doi.org/10.1017/9781316848142} {\emph {\bibinfo {title} {The Theory of Quantum Information}}}\ (\bibinfo  {publisher} {Cambridge University Press},\ \bibinfo {year} {2018})\BibitemShut {NoStop}%
\bibitem [{\citenamefont {Wilde}(2017)}]{Wilde_2017}%
  \BibitemOpen
  \bibfield  {author} {\bibinfo {author} {\bibfnamefont {M.~M.}\ \bibnamefont {Wilde}},\ }\href {https://www.cambridge.org/core/books/quantum-information-theory/247A740E156416531AA8CB97DFDAE438} {\emph {\bibinfo {title} {Quantum Information Theory}}},\ \bibinfo {edition} {2nd}\ ed.\ (\bibinfo  {publisher} {Cambridge University Press},\ \bibinfo {year} {2017})\BibitemShut {NoStop}%
\bibitem [{\citenamefont {Holevo}(1998)}]{IEEE-IT-651037}%
  \BibitemOpen
  \bibfield  {author} {\bibinfo {author} {\bibfnamefont {A.}~\bibnamefont {Holevo}},\ }\bibfield  {title} {\bibinfo {title} {The capacity of the quantum channel with general signal states},\ }\href {https://doi.org/10.1109/18.651037} {\bibfield  {journal} {\bibinfo  {journal} {IEEE Transactions on Information Theory}\ }\textbf {\bibinfo {volume} {44}},\ \bibinfo {pages} {269} (\bibinfo {year} {1998})}\BibitemShut {NoStop}%
\bibitem [{\citenamefont {Schumacher}\ and\ \citenamefont {Westmoreland}(1997)}]{PhysRevA.56.131}%
  \BibitemOpen
  \bibfield  {author} {\bibinfo {author} {\bibfnamefont {B.}~\bibnamefont {Schumacher}}\ and\ \bibinfo {author} {\bibfnamefont {M.~D.}\ \bibnamefont {Westmoreland}},\ }\bibfield  {title} {\bibinfo {title} {Sending classical information via noisy quantum channels},\ }\href {https://doi.org/10.1103/PhysRevA.56.131} {\bibfield  {journal} {\bibinfo  {journal} {Phys. Rev. A}\ }\textbf {\bibinfo {volume} {56}},\ \bibinfo {pages} {131} (\bibinfo {year} {1997})}\BibitemShut {NoStop}%
\bibitem [{\citenamefont {Bennett}\ \emph {et~al.}(1999)\citenamefont {Bennett}, \citenamefont {Shor}, \citenamefont {Smolin},\ and\ \citenamefont {Thapliyal}}]{PhysRevLett.83.3081}%
  \BibitemOpen
  \bibfield  {author} {\bibinfo {author} {\bibfnamefont {C.~H.}\ \bibnamefont {Bennett}}, \bibinfo {author} {\bibfnamefont {P.~W.}\ \bibnamefont {Shor}}, \bibinfo {author} {\bibfnamefont {J.~A.}\ \bibnamefont {Smolin}},\ and\ \bibinfo {author} {\bibfnamefont {A.~V.}\ \bibnamefont {Thapliyal}},\ }\bibfield  {title} {\bibinfo {title} {Entanglement-assisted classical capacity of noisy quantum channels},\ }\href {https://doi.org/10.1103/PhysRevLett.83.3081} {\bibfield  {journal} {\bibinfo  {journal} {Phys. Rev. Lett.}\ }\textbf {\bibinfo {volume} {83}},\ \bibinfo {pages} {3081} (\bibinfo {year} {1999})}\BibitemShut {NoStop}%
\bibitem [{\citenamefont {Bennett}\ \emph {et~al.}(2002)\citenamefont {Bennett}, \citenamefont {Shor}, \citenamefont {Smolin},\ and\ \citenamefont {Thapliyal}}]{1035117}%
  \BibitemOpen
  \bibfield  {author} {\bibinfo {author} {\bibfnamefont {C.}~\bibnamefont {Bennett}}, \bibinfo {author} {\bibfnamefont {P.}~\bibnamefont {Shor}}, \bibinfo {author} {\bibfnamefont {J.}~\bibnamefont {Smolin}},\ and\ \bibinfo {author} {\bibfnamefont {A.}~\bibnamefont {Thapliyal}},\ }\bibfield  {title} {\bibinfo {title} {Entanglement-assisted capacity of a quantum channel and the reverse shannon theorem},\ }\href {https://doi.org/10.1109/TIT.2002.802612} {\bibfield  {journal} {\bibinfo  {journal} {IEEE Transactions on Information Theory}\ }\textbf {\bibinfo {volume} {48}},\ \bibinfo {pages} {2637} (\bibinfo {year} {2002})}\BibitemShut {NoStop}%
\bibitem [{\citenamefont {Lloyd}(1997)}]{PhysRevA.55.1613}%
  \BibitemOpen
  \bibfield  {author} {\bibinfo {author} {\bibfnamefont {S.}~\bibnamefont {Lloyd}},\ }\bibfield  {title} {\bibinfo {title} {Capacity of the noisy quantum channel},\ }\href {https://doi.org/10.1103/PhysRevA.55.1613} {\bibfield  {journal} {\bibinfo  {journal} {Phys. Rev. A}\ }\textbf {\bibinfo {volume} {55}},\ \bibinfo {pages} {1613} (\bibinfo {year} {1997})}\BibitemShut {NoStop}%
\bibitem [{\citenamefont {Shor}(2002)}]{shor2002quantum}%
  \BibitemOpen
  \bibfield  {author} {\bibinfo {author} {\bibfnamefont {P.~W.}\ \bibnamefont {Shor}},\ }\bibfield  {title} {\bibinfo {title} {The quantum channel capacity and coherent information},\ }in\ \href@noop {} {\emph {\bibinfo {booktitle} {lecture notes, MSRI Workshop on Quantum Computation}}},\ Vol.~\bibinfo {volume} {5}\ (\bibinfo {year} {2002})\BibitemShut {NoStop}%
\bibitem [{\citenamefont {Devetak}(2005)}]{1377491}%
  \BibitemOpen
  \bibfield  {author} {\bibinfo {author} {\bibfnamefont {I.}~\bibnamefont {Devetak}},\ }\bibfield  {title} {\bibinfo {title} {The private classical capacity and quantum capacity of a quantum channel},\ }\href {https://doi.org/10.1109/TIT.2004.839515} {\bibfield  {journal} {\bibinfo  {journal} {IEEE Transactions on Information Theory}\ }\textbf {\bibinfo {volume} {51}},\ \bibinfo {pages} {44} (\bibinfo {year} {2005})}\BibitemShut {NoStop}%
\bibitem [{\citenamefont {Dennis}\ \emph {et~al.}(2002)\citenamefont {Dennis}, \citenamefont {Kitaev}, \citenamefont {Landahl},\ and\ \citenamefont {Preskill}}]{10.1063/1.1499754}%
  \BibitemOpen
  \bibfield  {author} {\bibinfo {author} {\bibfnamefont {E.}~\bibnamefont {Dennis}}, \bibinfo {author} {\bibfnamefont {A.}~\bibnamefont {Kitaev}}, \bibinfo {author} {\bibfnamefont {A.}~\bibnamefont {Landahl}},\ and\ \bibinfo {author} {\bibfnamefont {J.}~\bibnamefont {Preskill}},\ }\bibfield  {title} {\bibinfo {title} {Topological quantum memory},\ }\href {https://doi.org/10.1063/1.1499754} {\bibfield  {journal} {\bibinfo  {journal} {Journal of Mathematical Physics}\ }\textbf {\bibinfo {volume} {43}},\ \bibinfo {pages} {4452} (\bibinfo {year} {2002})}\BibitemShut {NoStop}%
\bibitem [{\citenamefont {Hayden}\ and\ \citenamefont {Preskill}(2007)}]{Patrick_Hayden_2007}%
  \BibitemOpen
  \bibfield  {author} {\bibinfo {author} {\bibfnamefont {P.}~\bibnamefont {Hayden}}\ and\ \bibinfo {author} {\bibfnamefont {J.}~\bibnamefont {Preskill}},\ }\bibfield  {title} {\bibinfo {title} {Black holes as mirrors: quantum information in random subsystems},\ }\href {https://doi.org/10.1088/1126-6708/2007/09/120} {\bibfield  {journal} {\bibinfo  {journal} {Journal of High Energy Physics}\ }\textbf {\bibinfo {volume} {2007}},\ \bibinfo {pages} {120} (\bibinfo {year} {2007})}\BibitemShut {NoStop}%
\bibitem [{\citenamefont {Bennett}\ \emph {et~al.}(2014)\citenamefont {Bennett}, \citenamefont {Devetak}, \citenamefont {Harrow}, \citenamefont {Shor},\ and\ \citenamefont {Winter}}]{6757002}%
  \BibitemOpen
  \bibfield  {author} {\bibinfo {author} {\bibfnamefont {C.~H.}\ \bibnamefont {Bennett}}, \bibinfo {author} {\bibfnamefont {I.}~\bibnamefont {Devetak}}, \bibinfo {author} {\bibfnamefont {A.~W.}\ \bibnamefont {Harrow}}, \bibinfo {author} {\bibfnamefont {P.~W.}\ \bibnamefont {Shor}},\ and\ \bibinfo {author} {\bibfnamefont {A.}~\bibnamefont {Winter}},\ }\bibfield  {title} {\bibinfo {title} {The quantum reverse shannon theorem and resource tradeoffs for simulating quantum channels},\ }\href {https://doi.org/10.1109/TIT.2014.2309968} {\bibfield  {journal} {\bibinfo  {journal} {IEEE Transactions on Information Theory}\ }\textbf {\bibinfo {volume} {60}},\ \bibinfo {pages} {2926} (\bibinfo {year} {2014})}\BibitemShut {NoStop}%
\bibitem [{\citenamefont {Hayashi}\ \emph {et~al.}(2025)\citenamefont {Hayashi}, \citenamefont {Cheng},\ and\ \citenamefont {Gao}}]{11005630}%
  \BibitemOpen
  \bibfield  {author} {\bibinfo {author} {\bibfnamefont {M.}~\bibnamefont {Hayashi}}, \bibinfo {author} {\bibfnamefont {H.-C.}\ \bibnamefont {Cheng}},\ and\ \bibinfo {author} {\bibfnamefont {L.}~\bibnamefont {Gao}},\ }\bibfield  {title} {\bibinfo {title} {Resolvability of classical-quantum channels},\ }\href {https://doi.org/10.1109/TIT.2025.3570569} {\bibfield  {journal} {\bibinfo  {journal} {IEEE Transactions on Information Theory}\ }\textbf {\bibinfo {volume} {71}},\ \bibinfo {pages} {6061} (\bibinfo {year} {2025})}\BibitemShut {NoStop}%
\bibitem [{\citenamefont {Kuroiwa}\ and\ \citenamefont {Yamasaki}(2020)}]{Kuroiwa2020}%
  \BibitemOpen
  \bibfield  {author} {\bibinfo {author} {\bibfnamefont {K.}~\bibnamefont {Kuroiwa}}\ and\ \bibinfo {author} {\bibfnamefont {H.}~\bibnamefont {Yamasaki}},\ }\bibfield  {title} {\bibinfo {title} {General {Q}uantum {R}esource {T}heories: {D}istillation, {F}ormation and {C}onsistent {R}esource {M}easures},\ }\href {https://doi.org/10.22331/q-2020-11-01-355} {\bibfield  {journal} {\bibinfo  {journal} {{Quantum}}\ }\textbf {\bibinfo {volume} {4}},\ \bibinfo {pages} {355} (\bibinfo {year} {2020})}\BibitemShut {NoStop}%
\bibitem [{\citenamefont {Chitambar}\ and\ \citenamefont {Gour}(2019)}]{Chitambar2018}%
  \BibitemOpen
  \bibfield  {author} {\bibinfo {author} {\bibfnamefont {E.}~\bibnamefont {Chitambar}}\ and\ \bibinfo {author} {\bibfnamefont {G.}~\bibnamefont {Gour}},\ }\bibfield  {title} {\bibinfo {title} {Quantum resource theories},\ }\href {https://doi.org/10.1103/RevModPhys.91.025001} {\bibfield  {journal} {\bibinfo  {journal} {Rev. Mod. Phys.}\ }\textbf {\bibinfo {volume} {91}},\ \bibinfo {pages} {025001} (\bibinfo {year} {2019})}\BibitemShut {NoStop}%
\bibitem [{\citenamefont {Brandão}\ and\ \citenamefont {Plenio}(2008)}]{Brand_o_2008}%
  \BibitemOpen
  \bibfield  {author} {\bibinfo {author} {\bibfnamefont {F.~G. S.~L.}\ \bibnamefont {Brandão}}\ and\ \bibinfo {author} {\bibfnamefont {M.~B.}\ \bibnamefont {Plenio}},\ }\bibfield  {title} {\bibinfo {title} {Entanglement theory and the second law of thermodynamics},\ }\href {https://doi.org/10.1038/nphys1100} {\bibfield  {journal} {\bibinfo  {journal} {Nature Physics}\ }\textbf {\bibinfo {volume} {4}},\ \bibinfo {pages} {873–877} (\bibinfo {year} {2008})}\BibitemShut {NoStop}%
\bibitem [{\citenamefont {Brandao}\ and\ \citenamefont {Plenio}(2010)}]{brandao2010reversible}%
  \BibitemOpen
  \bibfield  {author} {\bibinfo {author} {\bibfnamefont {F.~G.}\ \bibnamefont {Brandao}}\ and\ \bibinfo {author} {\bibfnamefont {M.~B.}\ \bibnamefont {Plenio}},\ }\bibfield  {title} {\bibinfo {title} {A reversible theory of entanglement and its relation to the second law},\ }\href {https://link.springer.com/article/10.1007/s00220-010-1003-1} {\bibfield  {journal} {\bibinfo  {journal} {Communications in Mathematical Physics}\ }\textbf {\bibinfo {volume} {295}},\ \bibinfo {pages} {829} (\bibinfo {year} {2010})}\BibitemShut {NoStop}%
\bibitem [{\citenamefont {Brand{\~a}o}\ and\ \citenamefont {Plenio}(2010)}]{Brandao2010}%
  \BibitemOpen
  \bibfield  {author} {\bibinfo {author} {\bibfnamefont {F.~G. S.~L.}\ \bibnamefont {Brand{\~a}o}}\ and\ \bibinfo {author} {\bibfnamefont {M.~B.}\ \bibnamefont {Plenio}},\ }\bibfield  {title} {\bibinfo {title} {A generalization of quantum stein's lemma},\ }\href {https://doi.org/10.1007/s00220-010-1005-z} {\bibfield  {journal} {\bibinfo  {journal} {Communications in Mathematical Physics}\ }\textbf {\bibinfo {volume} {295}},\ \bibinfo {pages} {791} (\bibinfo {year} {2010})}\BibitemShut {NoStop}%
\bibitem [{\citenamefont {Hayashi}\ and\ \citenamefont {Yamasaki}(2024)}]{hayashi2025generalizedquantumsteinslemma}%
  \BibitemOpen
  \bibfield  {author} {\bibinfo {author} {\bibfnamefont {M.}~\bibnamefont {Hayashi}}\ and\ \bibinfo {author} {\bibfnamefont {H.}~\bibnamefont {Yamasaki}},\ }\href {https://arxiv.org/abs/2408.02722} {\bibinfo {title} {Generalized quantum stein's lemma and second law of quantum resource theories}} (\bibinfo {year} {2024}),\ \Eprint {https://arxiv.org/abs/2408.02722} {arXiv:2408.02722 [quant-ph]} \BibitemShut {NoStop}%
\bibitem [{\citenamefont {Lami}(2025)}]{10898013}%
  \BibitemOpen
  \bibfield  {author} {\bibinfo {author} {\bibfnamefont {L.}~\bibnamefont {Lami}},\ }\bibfield  {title} {\bibinfo {title} {A solution of the generalized quantum stein's lemma},\ }\href {https://doi.org/10.1109/TIT.2025.3543610} {\bibfield  {journal} {\bibinfo  {journal} {IEEE Transactions on Information Theory}\ }\textbf {\bibinfo {volume} {71}},\ \bibinfo {pages} {4454} (\bibinfo {year} {2025})}\BibitemShut {NoStop}%
\bibitem [{\citenamefont {Berta}\ \emph {et~al.}(2023)\citenamefont {Berta}, \citenamefont {Brand{\~{a}}o}, \citenamefont {Gour}, \citenamefont {Lami}, \citenamefont {Plenio}, \citenamefont {Regula},\ and\ \citenamefont {Tomamichel}}]{berta2023gap}%
  \BibitemOpen
  \bibfield  {author} {\bibinfo {author} {\bibfnamefont {M.}~\bibnamefont {Berta}}, \bibinfo {author} {\bibfnamefont {F.~G. S.~L.}\ \bibnamefont {Brand{\~{a}}o}}, \bibinfo {author} {\bibfnamefont {G.}~\bibnamefont {Gour}}, \bibinfo {author} {\bibfnamefont {L.}~\bibnamefont {Lami}}, \bibinfo {author} {\bibfnamefont {M.~B.}\ \bibnamefont {Plenio}}, \bibinfo {author} {\bibfnamefont {B.}~\bibnamefont {Regula}},\ and\ \bibinfo {author} {\bibfnamefont {M.}~\bibnamefont {Tomamichel}},\ }\bibfield  {title} {\bibinfo {title} {On a gap in the proof of the generalised quantum {S}tein's lemma and its consequences for the reversibility of quantum resources},\ }\href {https://doi.org/10.22331/q-2023-09-07-1103} {\bibfield  {journal} {\bibinfo  {journal} {{Quantum}}\ }\textbf {\bibinfo {volume} {7}},\ \bibinfo {pages} {1103} (\bibinfo {year} {2023})}\BibitemShut {NoStop}%
\bibitem [{\citenamefont {Yamasaki}\ and\ \citenamefont {Kuroiwa}(2024)}]{yamasaki2024generalized}%
  \BibitemOpen
  \bibfield  {author} {\bibinfo {author} {\bibfnamefont {H.}~\bibnamefont {Yamasaki}}\ and\ \bibinfo {author} {\bibfnamefont {K.}~\bibnamefont {Kuroiwa}},\ }\href {https://arxiv.org/abs/2401.01926} {\bibinfo {title} {Generalized quantum stein's lemma: Redeeming second law of resource theories}} (\bibinfo {year} {2024}),\ \Eprint {https://arxiv.org/abs/2401.01926v2} {arXiv:2401.01926v2 [quant-ph]} \BibitemShut {NoStop}%
\bibitem [{\citenamefont {Fang}\ \emph {et~al.}(2025)\citenamefont {Fang}, \citenamefont {Fawzi},\ and\ \citenamefont {Fawzi}}]{fang2025generalizedquantumasymptoticequipartition}%
  \BibitemOpen
  \bibfield  {author} {\bibinfo {author} {\bibfnamefont {K.}~\bibnamefont {Fang}}, \bibinfo {author} {\bibfnamefont {H.}~\bibnamefont {Fawzi}},\ and\ \bibinfo {author} {\bibfnamefont {O.}~\bibnamefont {Fawzi}},\ }\href {https://arxiv.org/abs/2411.04035} {\bibinfo {title} {Generalized quantum asymptotic equipartition}} (\bibinfo {year} {2025}),\ \Eprint {https://arxiv.org/abs/2411.04035} {arXiv:2411.04035 [quant-ph]} \BibitemShut {NoStop}%
\bibitem [{\citenamefont {Fang}(2025)}]{fang2025errorexponentsquantumstate}%
  \BibitemOpen
  \bibfield  {author} {\bibinfo {author} {\bibfnamefont {K.}~\bibnamefont {Fang}},\ }\href {https://arxiv.org/abs/2508.12901} {\bibinfo {title} {Error exponents of quantum state discrimination with composite correlated hypotheses}} (\bibinfo {year} {2025}),\ \Eprint {https://arxiv.org/abs/2508.12901} {arXiv:2508.12901 [quant-ph]} \BibitemShut {NoStop}%
\bibitem [{\citenamefont {Brand{\~a}o}\ and\ \citenamefont {Gour}(2015)}]{Brandao2015}%
  \BibitemOpen
  \bibfield  {author} {\bibinfo {author} {\bibfnamefont {F.~G. S.~L.}\ \bibnamefont {Brand{\~a}o}}\ and\ \bibinfo {author} {\bibfnamefont {G.}~\bibnamefont {Gour}},\ }\bibfield  {title} {\bibinfo {title} {Reversible framework for quantum resource theories},\ }\href {https://doi.org/10.1103/PhysRevLett.115.070503} {\bibfield  {journal} {\bibinfo  {journal} {Phys. Rev. Lett.}\ }\textbf {\bibinfo {volume} {115}},\ \bibinfo {pages} {070503} (\bibinfo {year} {2015})}\BibitemShut {NoStop}%
\bibitem [{\citenamefont {Regula}\ and\ \citenamefont {Lami}(2024)}]{regula2023reversibility}%
  \BibitemOpen
  \bibfield  {author} {\bibinfo {author} {\bibfnamefont {B.}~\bibnamefont {Regula}}\ and\ \bibinfo {author} {\bibfnamefont {L.}~\bibnamefont {Lami}},\ }\bibfield  {title} {\bibinfo {title} {Reversibility of quantum resources through probabilistic protocols},\ }\href {https://www.nature.com/articles/s41467-024-47243-2} {\bibfield  {journal} {\bibinfo  {journal} {Nature Communications}\ }\textbf {\bibinfo {volume} {15}},\ \bibinfo {pages} {3096} (\bibinfo {year} {2024})}\BibitemShut {NoStop}%
\bibitem [{\citenamefont {Takagi}\ \emph {et~al.}(2020)\citenamefont {Takagi}, \citenamefont {Wang},\ and\ \citenamefont {Hayashi}}]{Takagi2020}%
  \BibitemOpen
  \bibfield  {author} {\bibinfo {author} {\bibfnamefont {R.}~\bibnamefont {Takagi}}, \bibinfo {author} {\bibfnamefont {K.}~\bibnamefont {Wang}},\ and\ \bibinfo {author} {\bibfnamefont {M.}~\bibnamefont {Hayashi}},\ }\bibfield  {title} {\bibinfo {title} {Application of the resource theory of channels to communication scenarios},\ }\href {https://doi.org/10.1103/PhysRevLett.124.120502} {\bibfield  {journal} {\bibinfo  {journal} {Phys. Rev. Lett.}\ }\textbf {\bibinfo {volume} {124}},\ \bibinfo {pages} {120502} (\bibinfo {year} {2020})}\BibitemShut {NoStop}%
\bibitem [{\citenamefont {Aharonov}\ \emph {et~al.}(1998)\citenamefont {Aharonov}, \citenamefont {Kitaev},\ and\ \citenamefont {Nisan}}]{10.1145/276698.276708}%
  \BibitemOpen
  \bibfield  {author} {\bibinfo {author} {\bibfnamefont {D.}~\bibnamefont {Aharonov}}, \bibinfo {author} {\bibfnamefont {A.}~\bibnamefont {Kitaev}},\ and\ \bibinfo {author} {\bibfnamefont {N.}~\bibnamefont {Nisan}},\ }\bibfield  {title} {\bibinfo {title} {Quantum circuits with mixed states},\ }in\ \href {https://doi.org/10.1145/276698.276708} {\emph {\bibinfo {booktitle} {Proceedings of the Thirtieth Annual ACM Symposium on Theory of Computing}}},\ \bibinfo {series and number} {STOC '98}\ (\bibinfo  {publisher} {Association for Computing Machinery},\ \bibinfo {address} {New York, NY, USA},\ \bibinfo {year} {1998})\ p.\ \bibinfo {pages} {20–30}\BibitemShut {NoStop}%
\bibitem [{\citenamefont {Cooney}\ \emph {et~al.}(2016)\citenamefont {Cooney}, \citenamefont {Mosonyi},\ and\ \citenamefont {Wilde}}]{cooney2016strong}%
  \BibitemOpen
  \bibfield  {author} {\bibinfo {author} {\bibfnamefont {T.}~\bibnamefont {Cooney}}, \bibinfo {author} {\bibfnamefont {M.}~\bibnamefont {Mosonyi}},\ and\ \bibinfo {author} {\bibfnamefont {M.~M.}\ \bibnamefont {Wilde}},\ }\bibfield  {title} {\bibinfo {title} {Strong converse exponents for a quantum channel discrimination problem and quantum-feedback-assisted communication},\ }\href {https://link.springer.com/article/10.1007/s00220-016-2645-4} {\bibfield  {journal} {\bibinfo  {journal} {Communications in Mathematical Physics}\ }\textbf {\bibinfo {volume} {344}},\ \bibinfo {pages} {797} (\bibinfo {year} {2016})}\BibitemShut {NoStop}%
\bibitem [{\citenamefont {Gour}\ and\ \citenamefont {Winter}(2019)}]{Gour2019a}%
  \BibitemOpen
  \bibfield  {author} {\bibinfo {author} {\bibfnamefont {G.}~\bibnamefont {Gour}}\ and\ \bibinfo {author} {\bibfnamefont {A.}~\bibnamefont {Winter}},\ }\bibfield  {title} {\bibinfo {title} {How to quantify a dynamical quantum resource},\ }\href {https://doi.org/10.1103/PhysRevLett.123.150401} {\bibfield  {journal} {\bibinfo  {journal} {Phys. Rev. Lett.}\ }\textbf {\bibinfo {volume} {123}},\ \bibinfo {pages} {150401} (\bibinfo {year} {2019})}\BibitemShut {NoStop}%
\bibitem [{\citenamefont {Hayashi}(2009)}]{5165184}%
  \BibitemOpen
  \bibfield  {author} {\bibinfo {author} {\bibfnamefont {M.}~\bibnamefont {Hayashi}},\ }\bibfield  {title} {\bibinfo {title} {Discrimination of two channels by adaptive methods and its application to quantum system},\ }\href {https://doi.org/10.1109/TIT.2009.2023726} {\bibfield  {journal} {\bibinfo  {journal} {IEEE Transactions on Information Theory}\ }\textbf {\bibinfo {volume} {55}},\ \bibinfo {pages} {3807} (\bibinfo {year} {2009})}\BibitemShut {NoStop}%
\bibitem [{\citenamefont {Wilde}\ \emph {et~al.}(2020)\citenamefont {Wilde}, \citenamefont {Berta}, \citenamefont {Hirche},\ and\ \citenamefont {Kaur}}]{wilde2020amortized}%
  \BibitemOpen
  \bibfield  {author} {\bibinfo {author} {\bibfnamefont {M.~M.}\ \bibnamefont {Wilde}}, \bibinfo {author} {\bibfnamefont {M.}~\bibnamefont {Berta}}, \bibinfo {author} {\bibfnamefont {C.}~\bibnamefont {Hirche}},\ and\ \bibinfo {author} {\bibfnamefont {E.}~\bibnamefont {Kaur}},\ }\bibfield  {title} {\bibinfo {title} {Amortized channel divergence for asymptotic quantum channel discrimination},\ }\href {https://link.springer.com/article/10.1007/s11005-020-01297-7} {\bibfield  {journal} {\bibinfo  {journal} {Letters in Mathematical Physics}\ }\textbf {\bibinfo {volume} {110}},\ \bibinfo {pages} {2277} (\bibinfo {year} {2020})}\BibitemShut {NoStop}%
\bibitem [{\citenamefont {Salek}\ \emph {et~al.}(2022)\citenamefont {Salek}, \citenamefont {Hayashi},\ and\ \citenamefont {Winter}}]{PhysRevA.105.022419}%
  \BibitemOpen
  \bibfield  {author} {\bibinfo {author} {\bibfnamefont {F.}~\bibnamefont {Salek}}, \bibinfo {author} {\bibfnamefont {M.}~\bibnamefont {Hayashi}},\ and\ \bibinfo {author} {\bibfnamefont {A.}~\bibnamefont {Winter}},\ }\bibfield  {title} {\bibinfo {title} {Usefulness of adaptive strategies in asymptotic quantum channel discrimination},\ }\href {https://doi.org/10.1103/PhysRevA.105.022419} {\bibfield  {journal} {\bibinfo  {journal} {Phys. Rev. A}\ }\textbf {\bibinfo {volume} {105}},\ \bibinfo {pages} {022419} (\bibinfo {year} {2022})}\BibitemShut {NoStop}%
\bibitem [{\citenamefont {Hiai}\ and\ \citenamefont {Petz}(1991)}]{hiai1991proper}%
  \BibitemOpen
  \bibfield  {author} {\bibinfo {author} {\bibfnamefont {F.}~\bibnamefont {Hiai}}\ and\ \bibinfo {author} {\bibfnamefont {D.}~\bibnamefont {Petz}},\ }\bibfield  {title} {\bibinfo {title} {The proper formula for relative entropy and its asymptotics in quantum probability},\ }\href {https://link.springer.com/article/10.1007/BF02100287} {\bibfield  {journal} {\bibinfo  {journal} {Communications in mathematical physics}\ }\textbf {\bibinfo {volume} {143}},\ \bibinfo {pages} {99} (\bibinfo {year} {1991})}\BibitemShut {NoStop}%
\bibitem [{\citenamefont {Ogawa}\ and\ \citenamefont {Nagaoka}(2000)}]{887855}%
  \BibitemOpen
  \bibfield  {author} {\bibinfo {author} {\bibfnamefont {T.}~\bibnamefont {Ogawa}}\ and\ \bibinfo {author} {\bibfnamefont {H.}~\bibnamefont {Nagaoka}},\ }\bibfield  {title} {\bibinfo {title} {Strong converse and stein's lemma in quantum hypothesis testing},\ }\href {https://doi.org/10.1109/18.887855} {\bibfield  {journal} {\bibinfo  {journal} {IEEE Transactions on Information Theory}\ }\textbf {\bibinfo {volume} {46}},\ \bibinfo {pages} {2428} (\bibinfo {year} {2000})}\BibitemShut {NoStop}%
\bibitem [{\citenamefont {Hayashi}(2002)}]{hayashi2002optimal}%
  \BibitemOpen
  \bibfield  {author} {\bibinfo {author} {\bibfnamefont {M.}~\bibnamefont {Hayashi}},\ }\bibfield  {title} {\bibinfo {title} {Optimal sequence of quantum measurements in the sense of stein's lemma in quantum hypothesis testing},\ }\href {https://iopscience.iop.org/article/10.1088/0305-4470/35/50/307} {\bibfield  {journal} {\bibinfo  {journal} {Journal of Physics A: Mathematical and General}\ }\textbf {\bibinfo {volume} {35}},\ \bibinfo {pages} {10759} (\bibinfo {year} {2002})}\BibitemShut {NoStop}%
\bibitem [{\citenamefont {Nagaoka}\ and\ \citenamefont {Hayashi}(2007)}]{4069150}%
  \BibitemOpen
  \bibfield  {author} {\bibinfo {author} {\bibfnamefont {H.}~\bibnamefont {Nagaoka}}\ and\ \bibinfo {author} {\bibfnamefont {M.}~\bibnamefont {Hayashi}},\ }\bibfield  {title} {\bibinfo {title} {An information-spectrum approach to classical and quantum hypothesis testing for simple hypotheses},\ }\href {https://doi.org/10.1109/TIT.2006.889463} {\bibfield  {journal} {\bibinfo  {journal} {IEEE Transactions on Information Theory}\ }\textbf {\bibinfo {volume} {53}},\ \bibinfo {pages} {534} (\bibinfo {year} {2007})}\BibitemShut {NoStop}%
\bibitem [{\citenamefont {Umegaki}(1962)}]{10.2996/kmj/1138844604}%
  \BibitemOpen
  \bibfield  {author} {\bibinfo {author} {\bibfnamefont {H.}~\bibnamefont {Umegaki}},\ }\bibfield  {title} {\bibinfo {title} {{Conditional expectation in an operator algebra. IV. Entropy and information}},\ }\href {https://doi.org/10.2996/kmj/1138844604} {\bibfield  {journal} {\bibinfo  {journal} {Kodai Mathematical Seminar Reports}\ }\textbf {\bibinfo {volume} {14}},\ \bibinfo {pages} {59 } (\bibinfo {year} {1962})}\BibitemShut {NoStop}%
\bibitem [{\citenamefont {Müller-Lennert}\ \emph {et~al.}(2013)\citenamefont {Müller-Lennert}, \citenamefont {Dupuis}, \citenamefont {Szehr}, \citenamefont {Fehr},\ and\ \citenamefont {Tomamichel}}]{10.1063/1.4838856}%
  \BibitemOpen
  \bibfield  {author} {\bibinfo {author} {\bibfnamefont {M.}~\bibnamefont {Müller-Lennert}}, \bibinfo {author} {\bibfnamefont {F.}~\bibnamefont {Dupuis}}, \bibinfo {author} {\bibfnamefont {O.}~\bibnamefont {Szehr}}, \bibinfo {author} {\bibfnamefont {S.}~\bibnamefont {Fehr}},\ and\ \bibinfo {author} {\bibfnamefont {M.}~\bibnamefont {Tomamichel}},\ }\bibfield  {title} {\bibinfo {title} {{On quantum Rényi entropies: A new generalization and some properties}},\ }\href {https://doi.org/10.1063/1.4838856} {\bibfield  {journal} {\bibinfo  {journal} {Journal of Mathematical Physics}\ }\textbf {\bibinfo {volume} {54}},\ \bibinfo {pages} {122203} (\bibinfo {year} {2013})}\BibitemShut {NoStop}%
\bibitem [{\citenamefont {Wilde}\ \emph {et~al.}(2014)\citenamefont {Wilde}, \citenamefont {Winter},\ and\ \citenamefont {Yang}}]{wilde2014strong}%
  \BibitemOpen
  \bibfield  {author} {\bibinfo {author} {\bibfnamefont {M.~M.}\ \bibnamefont {Wilde}}, \bibinfo {author} {\bibfnamefont {A.}~\bibnamefont {Winter}},\ and\ \bibinfo {author} {\bibfnamefont {D.}~\bibnamefont {Yang}},\ }\bibfield  {title} {\bibinfo {title} {Strong converse for the classical capacity of entanglement-breaking and hadamard channels via a sandwiched r{\'e}nyi relative entropy},\ }\href {https://link.springer.com/article/10.1007/s00220-014-2122-x} {\bibfield  {journal} {\bibinfo  {journal} {Communications in Mathematical Physics}\ }\textbf {\bibinfo {volume} {331}},\ \bibinfo {pages} {593} (\bibinfo {year} {2014})}\BibitemShut {NoStop}%
\bibitem [{\citenamefont {Petz}(1986)}]{PETZ198657}%
  \BibitemOpen
  \bibfield  {author} {\bibinfo {author} {\bibfnamefont {D.}~\bibnamefont {Petz}},\ }\bibfield  {title} {\bibinfo {title} {Quasi-entropies for finite quantum systems},\ }\href {https://doi.org/https://doi.org/10.1016/0034-4877(86)90067-4} {\bibfield  {journal} {\bibinfo  {journal} {Reports on Mathematical Physics}\ }\textbf {\bibinfo {volume} {23}},\ \bibinfo {pages} {57} (\bibinfo {year} {1986})}\BibitemShut {NoStop}%
\bibitem [{\citenamefont {Lieb}\ and\ \citenamefont {Thirring}(1976)}]{lieb1976inequalities}%
  \BibitemOpen
  \bibfield  {author} {\bibinfo {author} {\bibfnamefont {E.~H.}\ \bibnamefont {Lieb}}\ and\ \bibinfo {author} {\bibfnamefont {W.~E.}\ \bibnamefont {Thirring}},\ }\bibinfo {title} {Inequalities for the moments of the eigenvalues of the schr{\"o}dinger hamiltonian and their relation to sobolev inequalities},\ in\ \href {https://press.princeton.edu/books/hardcover/9780691644264/studies-in-mathematical-physics} {\emph {\bibinfo {booktitle} {Studies in Mathematical Physics}}},\ \bibinfo {editor} {edited by\ \bibinfo {editor} {\bibfnamefont {E.~H.}\ \bibnamefont {Lieb}}, \bibinfo {editor} {\bibfnamefont {B.}~\bibnamefont {Simon}},\ and\ \bibinfo {editor} {\bibfnamefont {A.~S.}\ \bibnamefont {Wightman}}}\ (\bibinfo  {publisher} {Princeton University Press},\ \bibinfo {address} {Princeton, NJ},\ \bibinfo {year} {1976})\ pp.\ \bibinfo {pages} {269--303}\BibitemShut {NoStop}%
\bibitem [{\citenamefont {Araki}(1990)}]{araki1990inequality}%
  \BibitemOpen
  \bibfield  {author} {\bibinfo {author} {\bibfnamefont {H.}~\bibnamefont {Araki}},\ }\bibfield  {title} {\bibinfo {title} {On an inequality of lieb and thirring},\ }\href {https://link.springer.com/article/10.1007/BF01045887} {\bibfield  {journal} {\bibinfo  {journal} {Letters in Mathematical Physics}\ }\textbf {\bibinfo {volume} {19}},\ \bibinfo {pages} {167} (\bibinfo {year} {1990})}\BibitemShut {NoStop}%
\bibitem [{\citenamefont {Bluhm}\ \emph {et~al.}(2023)\citenamefont {Bluhm}, \citenamefont {Capel}, \citenamefont {Gondolf},\ and\ \citenamefont {P\'{e}rez-Hern\'{a}ndez}}]{10129917}%
  \BibitemOpen
  \bibfield  {author} {\bibinfo {author} {\bibfnamefont {A.}~\bibnamefont {Bluhm}}, \bibinfo {author} {\bibfnamefont {A.}~\bibnamefont {Capel}}, \bibinfo {author} {\bibfnamefont {P.}~\bibnamefont {Gondolf}},\ and\ \bibinfo {author} {\bibfnamefont {A.}~\bibnamefont {P\'{e}rez-Hern\'{a}ndez}},\ }\bibfield  {title} {\bibinfo {title} {Continuity of quantum entropic quantities via almost convexity},\ }\href {https://doi.org/10.1109/TIT.2023.3277892} {\bibfield  {journal} {\bibinfo  {journal} {IEEE Transactions on Information Theory}\ }\textbf {\bibinfo {volume} {69}},\ \bibinfo {pages} {5869} (\bibinfo {year} {2023})}\BibitemShut {NoStop}%
\bibitem [{\citenamefont {Bluhm}\ \emph {et~al.}(2024)\citenamefont {Bluhm}, \citenamefont {Capel}, \citenamefont {Gondolf},\ and\ \citenamefont {P\'{e}rez-Hern\'{a}ndez}}]{10504886}%
  \BibitemOpen
  \bibfield  {author} {\bibinfo {author} {\bibfnamefont {A.}~\bibnamefont {Bluhm}}, \bibinfo {author} {\bibfnamefont {A.}~\bibnamefont {Capel}}, \bibinfo {author} {\bibfnamefont {P.}~\bibnamefont {Gondolf}},\ and\ \bibinfo {author} {\bibfnamefont {A.}~\bibnamefont {P\'{e}rez-Hern\'{a}ndez}},\ }\bibfield  {title} {\bibinfo {title} {Corrections to “continuity of quantum entropic quantities via almost convexity”},\ }\href {https://doi.org/10.1109/TIT.2024.3390839} {\bibfield  {journal} {\bibinfo  {journal} {IEEE Transactions on Information Theory}\ }\textbf {\bibinfo {volume} {70}},\ \bibinfo {pages} {5410} (\bibinfo {year} {2024})}\BibitemShut {NoStop}%
\bibitem [{\citenamefont {Chiribella}\ \emph {et~al.}(2008{\natexlab{a}})\citenamefont {Chiribella}, \citenamefont {D'Ariano},\ and\ \citenamefont {Perinotti}}]{PhysRevLett.101.060401}%
  \BibitemOpen
  \bibfield  {author} {\bibinfo {author} {\bibfnamefont {G.}~\bibnamefont {Chiribella}}, \bibinfo {author} {\bibfnamefont {G.~M.}\ \bibnamefont {D'Ariano}},\ and\ \bibinfo {author} {\bibfnamefont {P.}~\bibnamefont {Perinotti}},\ }\bibfield  {title} {\bibinfo {title} {Quantum circuit architecture},\ }\href {https://doi.org/10.1103/PhysRevLett.101.060401} {\bibfield  {journal} {\bibinfo  {journal} {Phys. Rev. Lett.}\ }\textbf {\bibinfo {volume} {101}},\ \bibinfo {pages} {060401} (\bibinfo {year} {2008}{\natexlab{a}})}\BibitemShut {NoStop}%
\bibitem [{\citenamefont {Chiribella}\ \emph {et~al.}(2008{\natexlab{b}})\citenamefont {Chiribella}, \citenamefont {D'Ariano},\ and\ \citenamefont {Perinotti}}]{Chiribella_2008}%
  \BibitemOpen
  \bibfield  {author} {\bibinfo {author} {\bibfnamefont {G.}~\bibnamefont {Chiribella}}, \bibinfo {author} {\bibfnamefont {G.~M.}\ \bibnamefont {D'Ariano}},\ and\ \bibinfo {author} {\bibfnamefont {P.}~\bibnamefont {Perinotti}},\ }\bibfield  {title} {\bibinfo {title} {Transforming quantum operations: Quantum supermaps},\ }\href {https://dx.doi.org/10.1209/0295-5075/83/30004} {\bibfield  {journal} {\bibinfo  {journal} {Europhysics Letters}\ }\textbf {\bibinfo {volume} {83}},\ \bibinfo {pages} {30004} (\bibinfo {year} {2008}{\natexlab{b}})}\BibitemShut {NoStop}%
\bibitem [{\citenamefont {Chiribella}\ \emph {et~al.}(2009)\citenamefont {Chiribella}, \citenamefont {D'Ariano},\ and\ \citenamefont {Perinotti}}]{PhysRevA.80.022339}%
  \BibitemOpen
  \bibfield  {author} {\bibinfo {author} {\bibfnamefont {G.}~\bibnamefont {Chiribella}}, \bibinfo {author} {\bibfnamefont {G.~M.}\ \bibnamefont {D'Ariano}},\ and\ \bibinfo {author} {\bibfnamefont {P.}~\bibnamefont {Perinotti}},\ }\bibfield  {title} {\bibinfo {title} {Theoretical framework for quantum networks},\ }\href {https://doi.org/10.1103/PhysRevA.80.022339} {\bibfield  {journal} {\bibinfo  {journal} {Phys. Rev. A}\ }\textbf {\bibinfo {volume} {80}},\ \bibinfo {pages} {022339} (\bibinfo {year} {2009})}\BibitemShut {NoStop}%
\bibitem [{\citenamefont {Winter}(2016)}]{Winter2016}%
  \BibitemOpen
  \bibfield  {author} {\bibinfo {author} {\bibfnamefont {A.}~\bibnamefont {Winter}},\ }\bibfield  {title} {\bibinfo {title} {Tight uniform continuity bounds for quantum entropies: Conditional entropy, relative entropy distance and energy constraints},\ }\href {https://doi.org/10.1007/s00220-016-2609-8} {\bibfield  {journal} {\bibinfo  {journal} {Communications in Mathematical Physics}\ }\textbf {\bibinfo {volume} {347}},\ \bibinfo {pages} {291} (\bibinfo {year} {2016})}\BibitemShut {NoStop}%
\bibitem [{\citenamefont {Fekete}(1923)}]{fekete1923verteilung}%
  \BibitemOpen
  \bibfield  {author} {\bibinfo {author} {\bibfnamefont {M.}~\bibnamefont {Fekete}},\ }\bibfield  {title} {\bibinfo {title} {{\"U}ber die verteilung der wurzeln bei gewissen algebraischen gleichungen mit ganzzahligen koeffizienten},\ }\href {https://link.springer.com/article/10.1007/BF01504345} {\bibfield  {journal} {\bibinfo  {journal} {Mathematische Zeitschrift}\ }\textbf {\bibinfo {volume} {17}},\ \bibinfo {pages} {228} (\bibinfo {year} {1923})}\BibitemShut {NoStop}%
\bibitem [{\citenamefont {v.~Neumann}(1928)}]{v1928theorie}%
  \BibitemOpen
  \bibfield  {author} {\bibinfo {author} {\bibfnamefont {J.}~\bibnamefont {v.~Neumann}},\ }\bibfield  {title} {\bibinfo {title} {Zur theorie der gesellschaftsspiele},\ }\href@noop {} {\bibfield  {journal} {\bibinfo  {journal} {Mathematische annalen}\ }\textbf {\bibinfo {volume} {100}},\ \bibinfo {pages} {295} (\bibinfo {year} {1928})}\BibitemShut {NoStop}%
\bibitem [{\citenamefont {Sion}(1958)}]{sion1958general}%
  \BibitemOpen
  \bibfield  {author} {\bibinfo {author} {\bibfnamefont {M.}~\bibnamefont {Sion}},\ }\bibfield  {title} {\bibinfo {title} {On general minimax theorems.},\ }\href {https://msp.org/pjm/1958/8-1/p14.xhtml} {\bibfield  {journal} {\bibinfo  {journal} {Pacific Journal of Mathematics}\ }\textbf {\bibinfo {volume} {8}},\ \bibinfo {pages} {171–176} (\bibinfo {year} {1958})}\BibitemShut {NoStop}%
\bibitem [{\citenamefont {Komiya}(1988)}]{10.2996/kmj/1138038812}%
  \BibitemOpen
  \bibfield  {author} {\bibinfo {author} {\bibfnamefont {H.}~\bibnamefont {Komiya}},\ }\bibfield  {title} {\bibinfo {title} {{Elementary proof for Sion's minimax theorem}},\ }\href {https://doi.org/10.2996/kmj/1138038812} {\bibfield  {journal} {\bibinfo  {journal} {Kodai Mathematical Journal}\ }\textbf {\bibinfo {volume} {11}},\ \bibinfo {pages} {5 } (\bibinfo {year} {1988})}\BibitemShut {NoStop}%
\bibitem [{\citenamefont {L{\"o}wner}(1934)}]{Lowner1934}%
  \BibitemOpen
  \bibfield  {author} {\bibinfo {author} {\bibfnamefont {K.}~\bibnamefont {L{\"o}wner}},\ }\bibfield  {title} {\bibinfo {title} {{\"U}ber monotone matrixfunktionen},\ }\href {https://doi.org/10.1007/BF01170633} {\bibfield  {journal} {\bibinfo  {journal} {Mathematische Zeitschrift}\ }\textbf {\bibinfo {volume} {38}},\ \bibinfo {pages} {177} (\bibinfo {year} {1934})}\BibitemShut {NoStop}%
\bibitem [{\citenamefont {Hiai}(2010)}]{FumioHIAI2010IIS160201}%
  \BibitemOpen
  \bibfield  {author} {\bibinfo {author} {\bibfnamefont {F.}~\bibnamefont {Hiai}},\ }\bibfield  {title} {\bibinfo {title} {Matrix analysis: Matrix monotone functions, matrix means, and majorization},\ }\href {https://doi.org/10.4036/iis.2010.139} {\bibfield  {journal} {\bibinfo  {journal} {Interdisciplinary Information Sciences}\ }\textbf {\bibinfo {volume} {16}},\ \bibinfo {pages} {139} (\bibinfo {year} {2010})}\BibitemShut {NoStop}%
\bibitem [{\citenamefont {Datta}(2009)}]{doi:10.1142/S0219749909005298}%
  \BibitemOpen
  \bibfield  {author} {\bibinfo {author} {\bibfnamefont {N.}~\bibnamefont {Datta}},\ }\bibfield  {title} {\bibinfo {title} {Max-relative entropy of entanglement, alias log robustness},\ }\href {https://doi.org/10.1142/S0219749909005298} {\bibfield  {journal} {\bibinfo  {journal} {International Journal of Quantum Information}\ }\textbf {\bibinfo {volume} {07}},\ \bibinfo {pages} {475} (\bibinfo {year} {2009})}\BibitemShut {NoStop}%
\bibitem [{\citenamefont {Fang}\ \emph {et~al.}(2020)\citenamefont {Fang}, \citenamefont {Wang}, \citenamefont {Tomamichel},\ and\ \citenamefont {Berta}}]{IEEE-IT-8850073}%
  \BibitemOpen
  \bibfield  {author} {\bibinfo {author} {\bibfnamefont {K.}~\bibnamefont {Fang}}, \bibinfo {author} {\bibfnamefont {X.}~\bibnamefont {Wang}}, \bibinfo {author} {\bibfnamefont {M.}~\bibnamefont {Tomamichel}},\ and\ \bibinfo {author} {\bibfnamefont {M.}~\bibnamefont {Berta}},\ }\bibfield  {title} {\bibinfo {title} {Quantum channel simulation and the channel's smooth max-information},\ }\href {https://doi.org/10.1109/TIT.2019.2943858} {\bibfield  {journal} {\bibinfo  {journal} {IEEE Transactions on Information Theory}\ }\textbf {\bibinfo {volume} {66}},\ \bibinfo {pages} {2129} (\bibinfo {year} {2020})}\BibitemShut {NoStop}%
\bibitem [{\citenamefont {Oufkir}\ \emph {et~al.}(2024)\citenamefont {Oufkir}, \citenamefont {Tomamichel},\ and\ \citenamefont {Berta}}]{OTB2024}%
  \BibitemOpen
  \bibfield  {author} {\bibinfo {author} {\bibfnamefont {A.}~\bibnamefont {Oufkir}}, \bibinfo {author} {\bibfnamefont {M.}~\bibnamefont {Tomamichel}},\ and\ \bibinfo {author} {\bibfnamefont {M.}~\bibnamefont {Berta}},\ }\href {https://arxiv.org/abs/2410.01084} {\bibinfo {title} {Error exponent of activated non-signaling assisted classical-quantum channel coding}} (\bibinfo {year} {2024}),\ \Eprint {https://arxiv.org/abs/2410.01084} {arXiv:2410.01084 [quant-ph]} \BibitemShut {NoStop}%
\bibitem [{\citenamefont {Berta}\ \emph {et~al.}(2011)\citenamefont {Berta}, \citenamefont {Christandl},\ and\ \citenamefont {Renner}}]{berta2011quantum}%
  \BibitemOpen
  \bibfield  {author} {\bibinfo {author} {\bibfnamefont {M.}~\bibnamefont {Berta}}, \bibinfo {author} {\bibfnamefont {M.}~\bibnamefont {Christandl}},\ and\ \bibinfo {author} {\bibfnamefont {R.}~\bibnamefont {Renner}},\ }\bibfield  {title} {\bibinfo {title} {The quantum reverse shannon theorem based on one-shot information theory},\ }\href {https://link.springer.com/article/10.1007/s00220-011-1309-7} {\bibfield  {journal} {\bibinfo  {journal} {Communications in Mathematical Physics}\ }\textbf {\bibinfo {volume} {306}},\ \bibinfo {pages} {579} (\bibinfo {year} {2011})}\BibitemShut {NoStop}%
\end{thebibliography}%

\end{document}